\documentclass[12pt]{article}
\pdfoutput=1
\usepackage[height=8.85in,width=6.45in]{geometry}
\usepackage{pifont}
\usepackage{amsthm}
\usepackage[mathscr]{euscript}

\usepackage{times}
\usepackage{pdflscape}
\usepackage{tikz}
\usetikzlibrary{arrows.meta, positioning, calc}
\usepackage{xfrac}
\usepackage{booktabs}
\usepackage[utf8]{inputenc}
\usepackage{amsmath}
\usepackage{amssymb}
\usepackage{mathtools}
\numberwithin{equation}{section}
\usepackage{slashed}
\usepackage{braket}
\usepackage[svgnames]{xcolor}
\usepackage{url}
\usepackage[colorlinks,citecolor=DarkGreen,linkcolor=FireBrick,linktocpage]{hyperref}
\usepackage{cite}
\usepackage{graphicx}
\usepackage{float}
\usepackage{tikz-cd}
\usepackage{stmaryrd}
\usepackage{courier}
\usepackage{dashbox}
\usepackage{caption}
\usepackage{subcaption}
\usepackage{enumitem}
\usepackage{footmisc}
\usepackage{tcolorbox}
\usepackage{fix-cm}

\usepackage{anyfontsize}

\usepackage{dsfont}
\usepackage{array}
\usepackage{multirow}
\usepackage{rotating}

 \def\ri{\mathrm{i}}

\newcolumntype{H}{>{\setbox0=\hbox\bgroup}c<{\egroup}@{}}

\renewcommand{\title}[1]{\vbox{\center\LARGE{#1}}\vspace{5mm}}
\renewcommand{\author}[1]{\vbox{\center#1}\vspace{5mm}}
\newcommand{\address}[1]{\vbox{\center\em#1}}

\makeatletter
\newsavebox{\@brx}
\newcommand{\llangle}[1][]{\savebox{\@brx}{\(\m@th{#1\langle}\)}%
  \mathopen{\copy\@brx\kern-0.5\wd\@brx\usebox{\@brx}}}
\newcommand{\rrangle}[1][]{\savebox{\@brx}{\(\m@th{#1\rangle}\)}%
  \mathclose{\copy\@brx\kern-0.5\wd\@brx\usebox{\@brx}}}
\makeatother

\def\Pr{\mathrm{Proj}_{\mathbb{Z}}}
\def\DD{D}

\newcommand{\Rep}{\mathrm{Rep}}

\newtheorem{theorem}{Theorem}[section]
\newtheorem{definition}[theorem]{Definition}
\newtheorem{corollary}[theorem]{Corollary}
\newtheorem{proposition}[theorem]{Proposition}
\newtheorem{conjecture}[theorem]{Conjecture}
\newtheorem{example}[theorem]{Example}
\newtheorem{lemma}[theorem]{Lemma}
\newtheorem{proposal}[theorem]{Proposal}

\tcbuselibrary{theorems, skins, breakable}

\newcounter{mainresultcounter}

\newtcolorbox{mainresult}[1][]{
  enhanced,
  colback=blue!5!white,
  colframe=blue!60!black,
  fonttitle=\bfseries,
  title=Main Result~\refstepcounter{mainresultcounter}\themainresultcounter,
  attach boxed title to top left={yshift=-2mm, xshift=4mm},
  boxed title style={
    colback=blue!60!black,
    colframe=blue!60!black,
    rounded corners
  },
  rounded corners,
  drop shadow,
  #1
}

\begin{document}

\begin{titlepage}

\title{Classification of Rational $c=1$ Vertex Operator Algebras and Vertex Operator Superalgebras}

\author{Terry Gannon${}^1$ and Brandon C.\ Rayhaun${}^{2}$}

        \address{${}^{1}$Department of Mathematics, University of Alberta, Edmonton, Alberta, Canada\\
        ${}^{2}$School of Natural Sciences, Institute for Advanced Study, Princeton, NJ, USA}

\begin{abstract}

\noindent\textbf{For mathematicians.} In this first in a series of two papers, we give a mathematically rigorous classification of (sufficiently nice) $c=1$ vertex operator algebras (VOAs) and vertex operator superalgebras (VOSAs). We confirm the lore that any such VO(S)A is either a lattice VO(S)A $V_L$ associated to a rank-1 integral lattice $L$, or can be obtained as an orbifold thereof, i.e.\ a $G$-invariant subalgebra $V_L^G$ for some finite group $G$ of automorphisms. All such $G$ are known, allowing for an explicit enumeration of nice $c=1$ VO(S)As. A key ingredient in our approach is to establish a general criterion for nice VOAs, requiring knowledge only of the vacuum character, for testing when the simple modules with integer conformal dimension span a symmetric fusion subcategory which is braided tensor equivalent to $\Rep(G)$. In our companion paper, we calculate the ribbon auto-equivalences of the representation categories of the nice $c=1$ VOAs, and leverage this to obtain the classification of nice bosonic and fermionic $c=1$ full conformal field theories. \\

\noindent\textbf{For physicists.} We rigorously classify the chiral algebras that can arise in the holomorphic sector of a bosonic or fermionic rational $c=1$ conformal field theory (CFT) whose non-identity primaries all have positive conformal dimension. Thinking of chiral algebras as gapless boundary conditions of 3D topological quantum field theories (TQFTs), our result says that any such chiral algebra is either a holomorphic boundary of $U(1)_k$ Chern-Simons theory, or can be obtained by passing to the $G$-invariant states thereof for some finite group $G$ of symmetries. We explicitly enumerate these chiral algebras and also discuss their non-invertible symmetries. In a companion paper, we build on these results using techniques from the study of 3D TQFTs to classify full bosonic and fermionic rational $c=1$ CFTs.

\end{abstract}

\end{titlepage}

\eject

\setcounter{tocdepth}{2}
\tableofcontents

\section{Introduction}

The classification of conformal field theories (CFTs) in two spacetime dimensions is of broad interest, having implications for a variety of subjects in physics and mathematics. Vertex operator algebras (VOAs), known to physicists as chiral algebras \cite{Zamolodchikov:1985wn}, furnish a rigorous axiomatic approach to the algebra of holomorphic local operators in a bosonic 2D CFT. In addition to their intrinsic value, VOAs serve as useful organizing devices for the physics of full CFTs, which generally contain both left-movers and right-movers. 

This paper is concerned with a certain class of VOAs which we refer to as ``nice'': those which are strongly rational and pseudo-unitary. (See the beginning of Appendix \ref{app:bestiary}, and in particular Definition \ref{def:niceVOA}, for our precise hypotheses.) Strong rationality refers to a collection of conditions which, in particular, are sufficient to guarantee that the VOA has finitely many simple modules into which every module decomposes as a direct sum. Pseudo-unitarity means that the non-identity primary operators (equivalently, non-vacuum simple modules) all have positive conformal dimension. 
For physicists, we remark that the conditions we impose are general enough to include the chiral algebra of any unitary rational conformal field theory.

Without further adjectives, the classification of nice VOAs is widely believed to be impossible. One way to make the problem tractable is to impose an upper bound on the central charge $c$ and the number $p$ of primary operators (i.e.\ simple modules). The state of the art \cite{Schellekens:1992db,Hohn:2017dsm,Moller:2019tlx,Moller:2021wva,Hohn:2020xfe,lam2023lattice,Mukhi:2022bte,Rayhaun:2023pgc,Hohn:2023auw,Moller:2024plb,Moller:2024xtt} is the classification of nice VOAs through (roughly) $c\lesssim 24$ and $p\lesssim 4$. Many results in this program have also been obtained at the level of genus-1 data using modular form technology, starting with \cite{Mathur:1988na} and continuing to the present day (see \cite{Das:2026jyp} and references therein).

\begin{figure}
\begin{center}
\tikzset{every picture/.style={line width=0.75pt}} 
\begin{tikzpicture}[x=0.75pt,y=0.75pt,yscale=-1,xscale=1]
\draw    (41,191) -- (307,191) ;
\draw [shift={(309,191)}, rotate = 180] [color={rgb, 255:red, 0; green, 0; blue, 0 }  ][line width=0.75]    (10.93,-3.29) .. controls (6.95,-1.4) and (3.31,-0.3) .. (0,0) .. controls (3.31,0.3) and (6.95,1.4) .. (10.93,3.29)   ;
\draw    (138,191) -- (138,8.17) ;
\draw [shift={(138,6.17)}, rotate = 90] [color={rgb, 255:red, 0; green, 0; blue, 0 }  ][line width=0.75]    (10.93,-3.29) .. controls (6.95,-1.4) and (3.31,-0.3) .. (0,0) .. controls (3.31,0.3) and (6.95,1.4) .. (10.93,3.29)   ;
\draw  [fill={rgb, 255:red, 0; green, 0; blue, 0 }  ,fill opacity=1 ] (38.5,191) .. controls (38.5,190.31) and (39.06,189.75) .. (39.75,189.75) .. controls (40.44,189.75) and (41,190.31) .. (41,191) .. controls (41,191.69) and (40.44,192.25) .. (39.75,192.25) .. controls (39.06,192.25) and (38.5,191.69) .. (38.5,191) -- cycle ;
\draw  [fill={rgb, 255:red, 0; green, 0; blue, 0 }  ,fill opacity=1 ] (136.75,191) .. controls (136.75,190.31) and (137.31,189.75) .. (138,189.75) .. controls (138.69,189.75) and (139.25,190.31) .. (139.25,191) .. controls (139.25,191.69) and (138.69,192.25) .. (138,192.25) .. controls (137.31,192.25) and (136.75,191.69) .. (136.75,191) -- cycle ;
\draw  [fill={rgb, 255:red, 0; green, 0; blue, 0 }  ,fill opacity=1 ] (218.5,30.5) .. controls (218.5,29.81) and (219.06,29.25) .. (219.75,29.25) .. controls (220.44,29.25) and (221,29.81) .. (221,30.5) .. controls (221,31.19) and (220.44,31.75) .. (219.75,31.75) .. controls (219.06,31.75) and (218.5,31.19) .. (218.5,30.5) -- cycle ;
\draw  [fill={rgb, 255:red, 0; green, 0; blue, 0 }  ,fill opacity=1 ] (219,60) .. controls (219,59.31) and (219.56,58.75) .. (220.25,58.75) .. controls (220.94,58.75) and (221.5,59.31) .. (221.5,60) .. controls (221.5,60.69) and (220.94,61.25) .. (220.25,61.25) .. controls (219.56,61.25) and (219,60.69) .. (219,60) -- cycle ;
\draw  [fill={rgb, 255:red, 0; green, 0; blue, 0 }  ,fill opacity=1 ] (219,90.5) .. controls (219,89.81) and (219.56,89.25) .. (220.25,89.25) .. controls (220.94,89.25) and (221.5,89.81) .. (221.5,90.5) .. controls (221.5,91.19) and (220.94,91.75) .. (220.25,91.75) .. controls (219.56,91.75) and (219,91.19) .. (219,90.5) -- cycle ;
\draw  [fill={rgb, 255:red, 0; green, 0; blue, 0 }  ,fill opacity=1 ] (136.67,150.5) .. controls (136.67,149.81) and (137.23,149.25) .. (137.92,149.25) .. controls (138.61,149.25) and (139.17,149.81) .. (139.17,150.5) .. controls (139.17,151.19) and (138.61,151.75) .. (137.92,151.75) .. controls (137.23,151.75) and (136.67,151.19) .. (136.67,150.5) -- cycle ;
\draw  [fill={rgb, 255:red, 0; green, 0; blue, 0 }  ,fill opacity=1 ] (79.33,190.83) .. controls (79.33,190.14) and (79.89,189.58) .. (80.58,189.58) .. controls (81.27,189.58) and (81.83,190.14) .. (81.83,190.83) .. controls (81.83,191.52) and (81.27,192.08) .. (80.58,192.08) .. controls (79.89,192.08) and (79.33,191.52) .. (79.33,190.83) -- cycle ;
\draw  [fill={rgb, 255:red, 0; green, 0; blue, 0 }  ,fill opacity=1 ] (136.75,120.5) .. controls (136.75,119.81) and (137.31,119.25) .. (138,119.25) .. controls (138.69,119.25) and (139.25,119.81) .. (139.25,120.5) .. controls (139.25,121.19) and (138.69,121.75) .. (138,121.75) .. controls (137.31,121.75) and (136.75,121.19) .. (136.75,120.5) -- cycle ;

\draw (285.5,195.4) node [anchor=north west][inner sep=0.75pt]    {$R_{\mathrm{circle}}$};
\draw (21,193) node [anchor=north west][inner sep=0.75pt]    {$\sqrt{2}$};
\draw (119.5,193) node [anchor=north west][inner sep=0.75pt]    {$\sqrt{8}$};
\draw (143.5,11.9) node [anchor=north west][inner sep=0.75pt]    {$R_{\mathrm{orbifold}}$};
\draw (10.5,215) node [anchor=north west][inner sep=0.75pt]    {$SU( 2)_{1}$};
\draw (124.5,215) node [anchor=north west][inner sep=0.75pt]    {$\mathrm{KT}$};
\draw (229,20.9) node [anchor=north west][inner sep=0.75pt]    {$SU( 2)_{1} /A_{4}$};
\draw (229.33,50.57) node [anchor=north west][inner sep=0.75pt]    {$SU( 2)_{1} /S_{4}$};
\draw (229.33,80.9) node [anchor=north west][inner sep=0.75pt]    {$SU( 2)_{1} /A_{5}$};
\draw (144.33,139.57) node [anchor=north west][inner sep=0.75pt]    {$\mathrm{Ising}^{2}$};
\draw (122,142.9) node [anchor=north west][inner sep=0.75pt]    {$2$};
\draw (75.33,197) node [anchor=north west][inner sep=0.75pt]    {$2$};
\draw (60,169.9) node [anchor=north west][inner sep=0.75pt]    {$\mathrm{Dirac}$};
\draw (108.33,107.9) node [anchor=north west][inner sep=0.75pt]    {$\sqrt{6}$};
\draw (143.67,112.23) node [anchor=north west][inner sep=0.75pt]    {$\mathbb{Z}_{4} \ \mathrm{parafermion}$};
\end{tikzpicture}
\end{center}
\caption{The (conjecturally complete) conformal moduli space of unitary 2D CFTs with compact spectrum and central charge $c=1$.}\label{fig:conformalmanifold}
\end{figure}
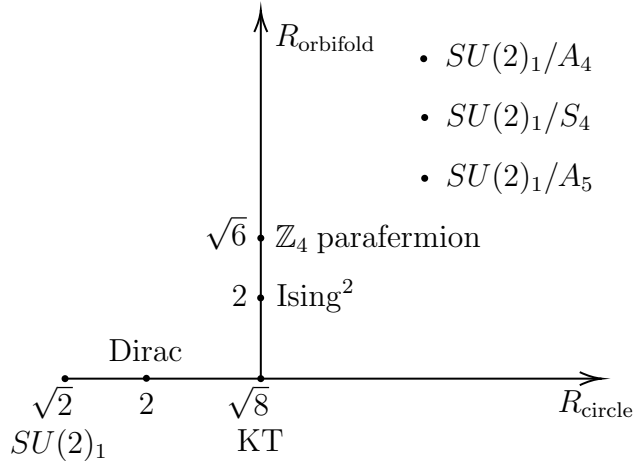

More ambitiously, one may attempt to classify nice VOAs in a range of central charges without any restriction on $p$. Here, the assumption of pseudo-unitarity can be viewed as a natural way to control an otherwise unfathomably large zoo: indeed, without pseudo-unitarity, every central charge is polluted by the possibility of tensoring with strongly rational VOAs of central charge $c=0$, of which there are an enormous infinity.\footnote{For example, $\mathrm{LY}^{5n}\otimes V_L$ has $c=0$, where $\mathrm{LY}=L(-\sfrac{22}{5})$ is the simple quotient of the Virasoro VOA at $c=-\sfrac{22}{5}$, and $V_L$ is the VOA associated to any even integral lattice $L$ of rank $22n$, with $n$ an arbitrary positive integer. Thus, the classification of strongly rational VOAs at a given central charge, without imposing pseudo-unitarity, contains as a subproblem the classification of lattices of unbounded rank!}

A beautiful result in this vein is the classification of nice 2D CFTs with central charge $c<1$, which are referred to as the minimal models \cite{Belavin:1984vu,kac1978highest,feigin1984verma,friedan1984conformal,cappelli1987modular,cappelli1987ade}. A parallel mathematical classification of nice VOAs with $c<1$ was obtained in \cite{dong2008classification,Dong:2014iza}.

Not long after these landmark results, Ginsparg \cite{Ginsparg:1987eb} explored the conformal moduli space\footnote{The term ``conformal manifold'' is typically used by physicists to describe the moduli space of conformal field theories at a fixed central charge, however we prefer the term ``conformal moduli space'' because, as inspection of Figure \ref{fig:conformalmanifold} reveals, it is often not a manifold at all!} of CFTs with central charge $c=1$ using the orbifold procedure (in the sense of physics) \cite{Dixon:1985jw,Dixon:1986jc}. This culminated in the conjecture that every unitary $c=1$ CFT with compact spectrum is either
\begin{enumerate}[label=\arabic*)]
\item a free compact boson (parametrized by a radius $R$, up to T-duality),
\item a sigma model into an interval (obtained by gauging the $\mathbb{Z}_2$ charge conjugation symmetry of a free compact boson theory for some $R$), or,
\item one of three isolated theories obtained by gauging the exceptional finite subgroups $A_4$, $S_4$, and $A_5$, respectively, of the diagonal $SO(3)$ subgroup of the $SO(4)\cong SU(2)_L\times SU(2)_R/\mathbb{Z}_2$ symmetry of the $SU(2)_1$ Wess-Zumino-Witten (WZW) model.
\end{enumerate}
This collection of theories is often depicted as in Figure \ref{fig:conformalmanifold}.

The conformal moduli space of unitary $c=1$ CFTs with compact spectrum contains both rational and irrational points. By examining the chiral algebras of the rational CFTs appearing in this list (which are expected to be nice when thought of as VOAs), one arrives at the expectation that any nice $c=1$ VOA is isomorphic to one of the following:
\begin{enumerate}[label=\arabic*)]
\item The VOA $V_{2m}$ $(m\geq 1)$ associated to the even integral lattice $\sqrt{2m}\mathbb{Z}$. Physically, $V_{2m}$ is the chiral algebra of the compact boson CFT of radius $R=\sqrt{2p/q}$ for any coprime integers $p,q$ with $pq=m$. It can also be thought of as a gapless chiral boundary condition for $U(1)_{2m}$ Chern-Simons theory. Note that $V_2\cong \widehat{\mathfrak{su}}(2)_1$, which is the current algebra living on the boundary of $SU(2)_1$ $(\cong U(1)_2)$ Chern-Simons theory.
\item The VOA $V_{2m}^+$ $(m\geq 2)$ obtained from $V_{2m}$ by passing to the subalgebra invariant under the $\mathbb{Z}_2^{\mathrm{C}}$ charge conjugation automorphism of $V_{2m}$, i.e.\ the standard lift of the $-1$ automorphism of the lattice $\sqrt{2m}\mathbb{Z}$.\footnote{We restrict to $m\geq 2$ because of the exceptional isomorphism $V_2^+\cong V_8$, which is related to the fact that the circle and orbifold branches of the $c=1$ conformal manifold meet at the Kosterlitz-Thouless point, see Figure \ref{fig:conformalmanifold}.} Physically, $V_{2m}^+$ is the chiral algebra of the theory obtained by gauging the $\mathbb{Z}_2$ charge conjugation symmetry of the compact boson of radius $R=\sqrt{2p/q}$ with $pq=m$. It can also be thought of as a gapless chiral boundary condition for $O(2)_{2m}$ Chern-Simons theory.
\item The three exceptional VOAs $V_T:=V_2^{A_4}$, $V_O:=V_2^{S_4}$, and $V_I:=V_2^{A_5}$, obtained from $V_2\cong \widehat{\mathfrak{su}}(2)_1$ by passing to the fixed points of the exceptional finite subgroups $A_4$, $S_4$, and $A_5$, respectively, of $\mathrm{Aut}(V_2)\cong \textsl{PSL}(2,\mathbb{C})$. 
\end{enumerate}

\subsection{Main result: classification of nice \texorpdfstring{$c=1$}{c=1} VOAs}

The main result of this paper is a proof of this folklore expectation.

\begin{mainresult}
\textbf{Theorem \ref{thm:c=1classification}.}
If $V$ is a nice $c=1$ VOA, then $V$ is isomorphic to one of the following,
\begin{align}
    V_{2m} \ (m\geq 1), \ \ \ V_{2m}^+ \ (m\geq 2), \ \ V_T, \ \ V_O, \ \ V_I.
\end{align}
\end{mainresult}
We remark that, conditional on the results of the recent preprint \cite{xu2026c2cofiniteness}, which claims to establish the strong rationality of the icosahedral VOA $V_I$, the converse of Theorem \ref{thm:c=1classification} can also be shown to be true. That is, every VOA appearing in the statement of Theorem \ref{thm:c=1classification} is nice and has central charge $c=1$ (see Proposition \ref{prop:conversec=1classification}).

We also remark for experts that Theorem \ref{thm:c=1classification} can be strengthened. In particular, we show in Corollary \ref{cor:strengthened} that the classification remains true if one relaxes the hypothesis of rationality (i.e.\ semisimplicity of the representation category) to rigidity (i.e.\ existence of duals in $\Rep(V)$), provided one assumes an appropriately modified notion of pseudo-unitarity. Rigidity is much weaker than rationality: there are $C_2$-cofinite VOAs that are not rational (like the triplet algebras), but all $C_2$-cofinite VOAs are expected to satisfy rigidity.

\bigskip 

While a classification of chiral algebras may appear to be of limited interest to a physicist interested in full 2D conformal field theories with both left- and right-movers, it is often quite feasible to turn a result about the former into a result about the latter. Indeed, in our companion paper \cite{grCFTs}, we will rigorously enumerate the consistent ways of constructing full CFTs whose holomorphic sector is one of the chiral algebras appearing in Theorem \ref{thm:c=1classification}. In particular, we will show that every bosonic rational $c=1$ CFT whose chiral algebra is nice appears somewhere in the conformal moduli space of Figure \ref{fig:conformalmanifold}. That is, Ginsparg's picture \cite{Ginsparg:1987eb}, put forward nearly 40 years ago, is complete, at least as far as bosonic rational CFTs are concerned.

Another pleasant byproduct of this theorem is that we expect that the techniques we use to prove it will generalize. Indeed, in work in preparation, we will show that essentially the same strategy yields the classification of nice $c=\sfrac32$ vertex operator superalgebras with $N=1$ supersymmetry. (See \cite{Dixon:1988ac,Wendland:2004pp} for expectations coming from physics.)

Although Theorem \ref{thm:c=1classification} has been anticipated for quite some time, it is somewhat surprising that it has any right to be true. Ginsparg obtained his conjectural picture of the conformal moduli space of unitary, compact $c=1$ CFTs by applying the orbifold procedure to \emph{invertible} symmetries of the theories that were known to him at the time. While the orbifold procedure is a powerful tool for producing new CFTs from known ones, orbifolding only invertible symmetries is known to be insufficient in general if one wishes to traverse the conformal moduli space of CFTs at a given central charge.

For example, at $c=\sfrac{21}{22}$, there are three minimal models, which we refer to as $A$, $D$, and $E$. As far as invertible symmetries go, $A$ and $D$ each possess only a $\mathbb{Z}_2$ symmetry. Orbifolding the $\mathbb{Z}_2$ symmetry of $A$ recovers $D$, and vice versa, so $E$ cannot be reached from $A$ or $D$ by orbifolding a sequence of invertible symmetries. In particular, if one had attempted to employ Ginsparg's strategy at $c=\sfrac{21}{22}$ starting with the diagonal model $A$, one would have discovered $D$ but missed the exceptional minimal model $E$.\footnote{Though, we note that $E$ \emph{can} be reached by gauging a Frobenius algebra object in the non-invertible category of Verlinde lines of $A$ \cite{Fuchs:2002cm,Bhardwaj:2017xup}.}

The upshot is that $c=1$ is in some ways more ``group-theoretical'' than other central charges. In fact, as we explain below, our proof makes crucial use of this feature of $c=1$ conformal field theory. Further discussion of the group-theoretical nature of the $c=1$ conformal moduli space will appear in \cite{flatroads,ncb}.

\subsection{Proof sketch}

Let us sketch our proof of Theorem \ref{thm:c=1classification}, first in terms more familiar to mathematicians, and then using physical language.

\paragraph{Proof sketch for mathematicians.} In our proof, a central role is played by the notion of a ``Tannakian VOA'', which we introduce in Section \ref{subsec:tannakian}. By definition, we say that a strongly rational VOA $V$ is Tannakian if its simple modules of integer conformal dimension span a symmetric fusion subcategory of $\Rep(V)$ which is braided tensor equivalent to $\Rep(G)$ for some finite group $G$. 

The most useful property of Tannakian VOAs for our purposes is derived in Proposition \ref{prop:Tannakianext}: if a VOA $V$ is nice and Tannakian, then it admits a canonical maximal conformal extension $V\subset V_{\Rep(G)}$ which is also nice. Furthermore, $\mathrm{Aut}(V_{\Rep(G)})$ admits a distinguished subgroup which is isomorphic to $G$, with the property that the fixed-points recover $V$, i.e.\ $V_{\Rep(G)}^G\cong V$.\footnote{We abusively use $G$ to also refer to the subgroup of $\mathrm{Aut}(V_{\Rep(G)})$.}

The crux of our proof of Theorem \ref{thm:c=1classification} involves demonstrating that any nice $c=1$ VOA $V$ is Tannakian, and that the corresponding maximal conformal extension is a lattice VOA $V_{2t_\star}$ for some square-free positive integer $t_\star$. It then follows from the discussion of the previous paragraph that any nice $c=1$ VOA can be obtained as $V_{2t_\star}^G$ for some square-free $t_\star$ and some finite subgroup $G\subset \mathrm{Aut}(V_{2t_\star})$. The automorphism groups of lattice VOAs are under good control (see e.g.\ Section \ref{subsec:lattice} for a review), so the nice $c=1$ VOAs can be explicitly enumerated using this characterization.

It turns out that one can determine whether or not a nice VOA is Tannakian once one knows its vacuum character and how it transforms under the modular transformation $\tau\mapsto -1/\tau$, as we demonstrate in Theorem \ref{thm:TannakianVacChar}. Indeed, let 
\begin{align}\label{eqn:vaccharV}
    \mathrm{ch}_V(\tau) := \mathrm{Tr}_Vq^{L_0-c/24}, \ \ \ \ q=e^{2\pi i \tau},
\end{align}
be the graded-dimension of $V$, and define the function 
\begin{align}\label{eqn:introPVdef}
    P_V(\tau) := \frac{1}{S_{00}}\Pr[\mathrm{ch}_V(-1/\tau)].
\end{align}
Here, $S_{00}$ is the vacuum-vacuum entry of the modular S-matrix of $V$, which can be extracted as the leading coefficient in the $q$-expansion of $\mathrm{ch}_V(-1/\tau)$, and 
\begin{align}
    \mathrm{Proj}_{\mathbb{Z}}\left[\sum_{x\in\mathbb{Q}}C(x)q^x\right]:=\sum_{n\in\mathbb{Z}}C(n-\sfrac c{24})q^{n-\sfrac{c}{24}}
\end{align}
projects onto the contributions of states with integer conformal dimensions. The function $P_V(\tau)$ contains non-trivial information about the modules of $V$ with integer conformal dimension. Indeed, one readily sees that 
\begin{align}
    P_V(\tau) = \sum_{a\in \Rep(V)_{h\in\mathbb{Z}}}d_a\mathrm{ch}_{V_a}(\tau),
\end{align}
where the sum over $\Rep(V)_{h\in\mathbb{Z}}$ is a sum over simple modules $V_a$ with integer conformal dimension, and $d_a=S_{0a}/S_{00}$ is the quantum dimension of $V_a$. In fact, Tannakianness can be completely recast as an easily checkable property of this function.\\

\noindent \textbf{Theorem \ref{thm:TannakianVacChar}.} \emph{A nice VOA is Tannakian with $\Rep(V)_{h\in\mathbb{Z}}\cong \Rep(G)$ if and only if }
\begin{align}\label{eqn:tannakiancondition}
    \frac{1}{S_{00}}\Pr[P_V(-1/\tau)]=|G|P_V(\tau),
\end{align}
\emph{where }$|G|$ \emph{is the order of the group} $G$.\\

The non-trivial content of Theorem \ref{thm:TannakianVacChar} is that, while the vacuum character generally distills only partial information about $\Rep(V)$ as a modular fusion category, it evidently sees enough to determine whether or not $V$ is Tannakian and, if so, what the order of the group $G$ is. We expect that this result will have a broad range of applicability beyond $c=1$ VOAs. 

Now, the vacuum character of a nice $c=1$ VOA is extremely constrained by modular invariance \cite{Zhu:1996gaq,Coste:1999yc,Ng:2010win,dong2015congruence}, and in particular by the Serre-Stark theorem \cite{serre2006modular}, as was first noticed by Kiritsis \cite{Kiritsis:1988et,Kiritsis:1988es}. Indeed, at $c=1$, modularity implies that $\mathrm{ch}_V(\tau)$ can be expanded in theta functions, as we describe in Lemma \ref{lem:c=1VacChar}. It turns out that we do not actually need to carry out an explicit enumeration of healthy-looking vacuum characters to finish our argument: the general expression from Lemma \ref{lem:c=1VacChar}, coupled with an application of Galois symmetry, is enough to invoke Theorem \ref{thm:TannakianVacChar} and conclude that $V$ is Tannakian. What is more, the general shape of the vacuum character is constrained enough that we are able to conclude that its maximal conformal extension $V_{\Rep(G)}$ has a vacuum character which agrees with that of $V_{2t_\star}$ for some square-free integer $t_\star$. Invoking a celebrated characterization of lattice VOAs by Dong and Mason \cite{Dong:2002fs}, this actually implies that $V_{\Rep(G)}$ is isomorphic to $V_{2t_\star}$, which finishes the argument.

\begin{figure}
    \begin{center}
\begin{footnotesize}
\tikzset{every picture/.style={line width=0.75pt}}    
\begin{tikzpicture}[x=0.75pt,y=0.75pt,yscale=-1,xscale=1,scale=.8]
\draw  [draw opacity=0][fill={rgb, 255:red, 184; green, 233; blue, 134 }  ,fill opacity=1 ][line width=1.5]  (66.2,119.9) -- (66.2,10) -- (162.2,57.1) -- (162.2,167) -- cycle ;
\draw [line width=1.5]  [dash pattern={on 1.69pt off 2.76pt}]  (6.5,10) -- (66.2,10) ;
\draw [line width=1.5]  [dash pattern={on 1.69pt off 2.76pt}]  (106.5,167) -- (162.2,167) ;
\draw [line width=1.5]  [dash pattern={on 1.69pt off 2.76pt}]  (6.5,119.9) -- (66.2,119.9) ;
\draw [line width=1.5]  [dash pattern={on 1.69pt off 2.76pt}]  (106.5,57.1) -- (162.2,57.1) ;
\draw    (197,155.5) -- (254.75,155.5) ;
\draw [shift={(256.75,155.5)}, rotate = 180] [color={rgb, 255:red, 0; green, 0; blue, 0 }  ][line width=0.75]    (10.93,-3.29) .. controls (6.95,-1.4) and (3.31,-0.3) .. (0,0) .. controls (3.31,0.3) and (6.95,1.4) .. (10.93,3.29)   ;
\draw  [draw opacity=0][fill={rgb, 255:red, 245; green, 166; blue, 35 }  ,fill opacity=.7 ][line width=1.5]  (372.2,118.9) -- (372.2,9) -- (468.2,56.1) -- (468.2,166) -- cycle ;
\draw [line width=1.5]  [dash pattern={on 1.69pt off 2.76pt}]  (312.5,9) -- (372.2,9) ;
\draw [line width=1.5]  [dash pattern={on 1.69pt off 2.76pt}]  (412.5,166) -- (468.2,166) ;
\draw [line width=1.5]  [dash pattern={on 1.69pt off 2.76pt}]  (312.5,118.9) -- (372.2,118.9) ;
\draw [line width=1.5]  [dash pattern={on 1.69pt off 2.76pt}]  (412.5,56.1) -- (468.2,56.1) ;
\draw    (490,155) -- (547.75,155) ;
\draw [shift={(549.75,155)}, rotate = 180] [color={rgb, 255:red, 0; green, 0; blue, 0 }  ][line width=0.75]    (10.93,-3.29) .. controls (6.95,-1.4) and (3.31,-0.3) .. (0,0) .. controls (3.31,0.3) and (6.95,1.4) .. (10.93,3.29)   ;
\draw  [draw opacity=0][fill={rgb, 255:red, 184; green, 233; blue, 134 }  ,fill opacity=.7 ][line width=1.5]  (658.2,118.9) -- (658.2,9) -- (754.2,56.1) -- (754.2,166) -- cycle ;
\draw [line width=1.5]  [dash pattern={on 1.69pt off 2.76pt}]  (598.5,9) -- (658.2,9) ;
\draw [line width=1.5]  [dash pattern={on 1.69pt off 2.76pt}]  (698.5,166) -- (754.2,166) ;
\draw [line width=1.5]  [dash pattern={on 1.69pt off 2.76pt}]  (598.5,118.9) -- (658.2,118.9) ;
\draw [line width=1.5]  [dash pattern={on 1.69pt off 2.76pt}]  (698.5,56.1) -- (754.2,56.1) ;

\draw (2,73.1) node [anchor=north west][inner sep=0.75pt]    {$\mathrm{TQFT}_{V}$};
\draw (105.2,79.9) node [anchor=north west][inner sep=0.75pt]  [color={rgb, 255:red, 65; green, 117; blue, 5 }  ,opacity=1 ]  {$V$};
\draw (172,125) node [anchor=north west][inner sep=0.75pt]    {$\text{gauge} \ \mathrm{Rep}( G)^{( 1)}$};
\draw (220,53.1) node [anchor=north west][inner sep=0.75pt]    {$ \begin{array}{l}
\mathrm{TQFT}_{V} /\mathrm{Rep}( G)^{( 1)}\\
\ \ \ \ \ \ \cong U( 1)_{2t_\star} \ \text{CS}
\end{array}$};
\draw (383.2,65.9) node [anchor=north west][inner sep=0.75pt]  [color={rgb, 255:red, 208; green, 2; blue, 27 }  ,opacity=1 ]  {$ \begin{array}{l}
V_{\mathrm{Rep}( G)}\\
\ \ \ \cong V_{2t_\star}
\end{array}$};
\draw (486,125) node [anchor=north west][inner sep=0.75pt]    {$\text{gauge} \ G^{( 0)}$};
\draw (512,55.1) node [anchor=north west][inner sep=0.75pt]    {$ \begin{array}{l}
\left( U( 1)_{2t_\star} \ \text{CS}\right) /G^{( 0)}\\
\ \ \ \ \ \ \ \cong \mathrm{TQFT}_{V}
\end{array}$};
\draw (678.2,70.9) node [anchor=north west][inner sep=0.75pt]  [color={rgb, 255:red, 65; green, 117; blue, 5 }  ,opacity=1 ]  {$ \begin{array}{l}
V_{2t_\star}^{G}\\
\ \cong V
\end{array}$};
\end{tikzpicture}
\end{footnotesize}
    \end{center}
\caption{A physical interpretation of our proof of the classification of nice $c=1$ chiral algebras. We argue that any such chiral algebra $V$ occurs on the boundary of a TQFT whose spinless anyons generate a $\mathrm{Rep}(G)$ one-form symmetry. When $\mathrm{Rep}(G)$ is gauged, the bulk goes over to $U(1)_{2t_\star}$ Chern-Simons theory for some square-free $t_\star$, and the boundary goes over to the free chiral boson theory $V_{2t_\star}$. This bulk/boundary theory inherits a dual/quantum $G$ zero-form symmetry which, on the one hand, ``undoes'' the $\mathrm{Rep}(G)$ gauging and, on the other, has the effect of projecting the local operators on the boundary to the $G$-invariant subalgebra $V_{2t_\star}^G$. Thus, every nice $c=1$ chiral algebra $V$ is of the form $V_{2t_\star}^G$ for some square-free $t_\star$ and some finite subgroup $G\subset \mathrm{Aut}(V_{2t_\star})$.}\label{fig:GRep(G)gauging}
\end{figure}
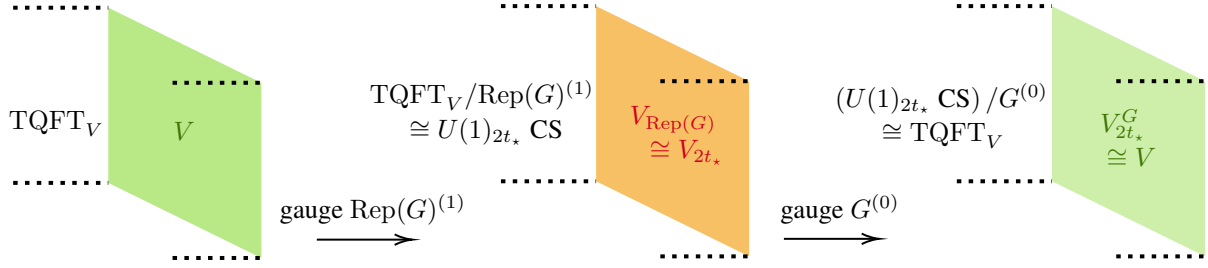

\paragraph{Proof sketch for physicists.} We adopt the point of view that a rational chiral algebra is a gapless chiral boundary condition of a 3D semisimple TQFT \cite{Witten:1988hf,Elitzur:1989nr}. We use $\mathrm{TQFT}_V$ to denote this bulk theory which hosts $V$ as its edge modes. See the left of Figure \ref{fig:GRep(G)gauging}.

The first step of our argument is to show that, for any nice $c=1$ chiral algebra $V$, the  spinless anyons of $\mathrm{TQFT}_V$ generate a $\Rep(G)$ non-invertible 1-form symmetry. We call nice VOAs with this property \emph{Tannakian}. 

A surprising fact which we establish in Theorem \ref{thm:TannakianVacChar} is that whether or not a nice VOA $V$ is Tannakian can be read off purely from its vacuum character, Equation \eqref{eqn:vaccharV}, 
which can be thought of as the partition function of $\mathrm{TQFT}_V$ on $D^2\times S^1$, with $V$ imposed as a boundary condition on $\partial(D^2\times S^1)\cong T^2$. (The complex structure of $T^2$ is taken to be $\tau$.) In particular, one needs only check that Equation  \eqref{eqn:tannakiancondition} is satisfied to conclude that $V$ is Tannakian, which is straightforward if one understands how $\mathrm{ch}_V(\tau)$ transforms under $\tau\mapsto-1/\tau$.

The vacuum character in turn is extremely constrained by modular covariance. Indeed, letting $\eta(\tau) = q^{1/24}\prod_{n=1}^\infty (1-q^n)$ be the Dedekind eta function, one can show that $\eta(\tau)\mathrm{ch}_V(\tau)$ is a weight-$\sfrac12$ holomorphic modular form for $\Gamma_1(C_V)$, where $C_V$ is the so-called ``conductor'' of $V$ (see Lemma \ref{lem:FVmodularity}). Therefore, $\eta(\tau)\mathrm{ch}_V(\tau)$ is precisely the kind of object that is classified by the Serre-Stark theorem \cite{serre2006modular}, which asserts that it can be expanded in theta functions (see Lemma \ref{lem:c=1VacChar} for the precise statement). It turns out that even this rough shape, after an application of Galois symmetry, gives us enough information about $V$ to invoke Theorem \ref{thm:TannakianVacChar} and conclude that it is Tannakian.

The second step of our argument involves gauging the $\Rep(G)$ one-form symmetry of $\mathrm{TQFT}_V$ in the presence of its $V$ chiral boundary condition, as in the passage from the left to the middle of Figure \ref{fig:GRep(G)gauging}. In addition to modifying the bulk TQFT, the effect of gauging on the boundary is to extend $V$ to a larger rational chiral algebra which we call $V_{\Rep(G)}$, whose vacuum character is precisely the function $P_V(\tau)$ defined in \eqref{eqn:introPVdef}. 

Again, $\mathrm{ch}_V(\tau)$ is constrained heavily enough by the Serre-Stark theorem to conclude that the coefficient of $q^{1-1/24}$ in $P_V(\tau)$ is non-vanishing, and therefore that $V_{\Rep(G)}$ has a non-zero number of spin-1 Noether currents. In particular, the rank $r$ of the Lie algebra of the global symmetry group of $V_{\Rep(G)}$ is at least $1$. In (pseudo)-unitary chiral algebras, the rank $r$ is always bounded from above by the central charge $c$, with equality if and only if the theory is a collection of $c$-many free chiral bosons with vertex operators belonging to a rank-$c$ charge lattice.\footnote{More generally, the rank $r$ is always bounded from above by the effective central charge $c-24h_{\mathrm{min}}$, where $h_{\mathrm{min}}$ is the minimum taken over the conformal dimensions of the primaries \cite{Dong:2002fs}.} Thus, at $c=1$, having a non-zero number of spin-1 currents is enough to conclude that the bulk TQFT after gauging is $U(1)_{2t_\star}$ Chern-Simons theory for some $t_\star$, and $V_{\Rep(G)}$ is its (essentially unique) free boson boundary condition, which we call $V_{2t_\star}$. (In fact, it turns out that $t_\star$ is a square-free integer.)

The last step of our argument is to invoke the fact that orbifolding a finite symmetry can always be ``undone'' by orbifolding a ``quantum symmetry'' \cite{Vafa:1989ih}. In the present setting, we have on general grounds that the bulk/boundary theory after gauging the $\Rep(G)$ one-form symmetry itself has a non-anomalous $G$ zero-form symmetry. Each element of $G$ is implemented by a topological surface which is able to end on a topological line junction on the boundary, and acts on the chiral algebra via automorphisms (see e.g.\ \cite{Gannon:2026ttf} for a recent discussion).  Gauging this $G$ zero-form symmetry, on the one hand, has the effect of throwing out the operators on the boundary that are charged under the $G$ symmetry. On the other hand, it should simply recover the chiral algebra $V$ we started with, which by assumption was an arbitrary nice $c=1$ chiral algebra. This is depicted in the passage from the middle to the right of Figure \ref{fig:GRep(G)gauging}.

Putting all this together, we learn that any nice $c=1$ VOA $V$ can be obtained as the $G$-invariant subalgebra of a free chiral boson theory $V_{2t_\star}$, for some square-free $t_\star$ and some finite subgroup $G$ of automorphisms of $V_{2t_\star}$. It is a straightforward exercise to enumerate the chiral algebras of this form, and the result is precisely the theories reported in the statement of Theorem \ref{thm:c=1classification}. 

\subsection{Relation to previous work}

We pause to emphasize that the classification of nice $c=1$ chiral algebras and conformal field theories is a problem to which countless authors have contributed (see e.g.\ \cite{Ginsparg:1987eb,Dijkgraaf:1987vp,Kiritsis:1988et,Kiritsis:1988es,Dijkgraaf:1989hb,Carpi:1999pq,Cappelli:2002wq,Carpi:2004bg,xu2005strong,dong2011characterization,dong2013characterization,10.1007/978-3-662-43831-2_3,lin2017characterizations,carpi2019classification,xu2026c2cofiniteness} for a partial list). Let us briefly compare our results to some of the existing literature. 

Kiritsis \cite{Kiritsis:1988et,Kiritsis:1988es} claimed a proof of the classification of full rational $c=1$ CFTs. However, the arguments carried out in op.\ cit.\ were sufficiently compressed that we have opted to structure our proofs so that they are logically independent. Even if one accepts that the contents of these papers are ultimately true, Kiritsis only attempted to classify consistent genus-1 partition functions, which are  not sufficient to characterize a CFT. Indeed, there exist ``isospectral'' theories which have the same torus partition functions but nevertheless are not isomorphic, such as the chiral $E_{8,1}\otimes E_{8,1}$ and $\textsl{Spin}(32)_1/\mathbb{Z}_2$ WZW models. In other words, one cannot hear the shape of a 2D CFT \cite{kac1966can}. The starting point of our approach is similar to that of Kiritsis in that we invoke the Serre-Stark theorem to constrain the genus-1 data of nice $c=1$ CFTs. However, we do not ever need to explicitly enumerate ``healthy-looking'' $q$-series, and therefore bypass some of the more case-by-case analyses of \cite{Kiritsis:1988et,Kiritsis:1988es}.

Xu \cite[Theorem 4.6]{xu2005strong} obtained an analog of Theorem \ref{thm:c=1classification} for conformal nets on the assumption of a spectrum condition, which requires roughly that a degenerate module of the $c=1$ Virasoro algebra other than the vacuum should appear in the conformal net. Conformal nets are expected to be equivalent to unitary VOAs \cite{Carpi:2015fga,Carpi:2023onx,henriques2025every} though, at the time of writing, a complete dictionary appears to be quite far away. Nevertheless, we can take for granted the expected equivalence and compare Xu's result to ours. After translating it to a statement about VOAs, it is in some ways stronger than our Theorem \ref{thm:c=1classification} and in other ways weaker. On the one hand, Xu does not require rationality. On the other hand, he imposes unitarity (compared to our requirement of pseudo-unitary) and, more importantly, the spectrum condition. Although one can in fact check that the nice $c=1$ VOAs appearing in Theorem \ref{thm:c=1classification} do satisfy the spectrum condition, it is unclear a priori why it has to hold, particularly in the absence of rationality.

Finally, in a series of papers \cite{dong2011characterization,dong2013characterization,10.1007/978-3-662-43831-2_3,lin2017characterizations}, most $c=1$ VOAs were characterized in terms of their low-lying spectrum (i.e.\ in terms of the dimensions of their weight spaces with small conformal dimension). However, as far as the authors are aware, the icosahedral VOA $V_I=V_2^{A_5}$ was never characterized in this manner. Even granting that this gap may soon be filled, completing the $c=1$ classification using these papers seems to require accepting the partition function characterization of \cite{Kiritsis:1988et,Kiritsis:1988es}, whose complete proof has not yet appeared in the literature. Finally, the arguments in all these cases are long, computationally involved, and distributed across many papers in the literature. For all these reasons, we felt that a conceptually and technically simpler approach to the classification of nice $c=1$ VOAs is desirable. Although the full paper clocks in at nearly 75 pages, most of this length comes from appendices, with the core argument requiring less than 20.

\subsection{A fermionic corollary:  classification of nice \texorpdfstring{$c=1$}{c=1} VOSAs}

Before concluding the introduction, we describe one further result. We obtain the classification of nice $c=1$ vertex operator \emph{super}algebras (VOSAs) as a straightforward corollary of Theorem \ref{thm:c=1classification}. See Definition \ref{def:niceVOSA} for our notion of a ``nice'' VOSA. A physicist should not understand the word ``super'' as referring to supersymmetry: rather, VOSAs are a mathematical axiomatization of the algebra of holomorphic local operators in the NS sector of a \emph{fermionic} 2D CFT which is not necessarily supersymmetric.

Let us denote the decomposition of a VOSA into even and odd subspaces with respect to the canonical involution $(-1)^F$ (i.e.\ fermion parity) as 
\begin{align}
    V=V_{\bar 0}\oplus V_{\bar 1}.
\end{align}
We use the notation $V_m$ to refer to the VOSA corresponding to the lattice $\sqrt{m}\mathbb{Z}$ and $V_m^+$ to refer to the fixed-points with respect to the $\mathbb{Z}_2^{\mathrm{C}}$ charge-conjugation automorphism. 
Physically, $V_m$ refers to the holomorphic boundary of $U(1)_m$ Chern-Simons theory, which is a spin TQFT when $m$ is odd \cite{Belov:2005ze}, and hence possesses a fermionic chiral algebra boundary. (For example, $V_1\cong \mathrm{Fer}(2)$ is the theory of 2 free chiral fermions.) Similarly, $V_m^+$ refers to the holomorphic boundary of $O(2)_m\cong U(1)_m/\mathbb{Z}_2^{\mathrm{C}}$ Chern-Simons theory.

Our second main result says that these are the only nice VOSAs at $c=1$. 

\begin{mainresult}
    \textbf{Theorem \ref{thm:c=1VOSAs}.} A vertex operator superalgebra $V$ is nice with $c=1$ and $V_{\bar 1}\neq 0$ if and only if it is isomorphic to one of the following, 
    \begin{align}
        V_m \ (m\geq 1, \ m\text{ odd}), \ \ \ \ \  V_m^+ \ (m\geq 1, \ m\text{ odd}).
    \end{align}
\end{mainresult}
Our proof relies on the observation that VOAs and VOSAs at a given central charge are related in a precise way by bosonization and fermionization maps. This was also the starting point for the recent classification of chiral fermionic CFTs with central charge $c\leq 24$ in \cite{BoyleSmith:2023xkd,Hohn:2023auw,Rayhaun:2023pgc}. 

The basic idea, which we state precisely in Proposition \ref{prop:chiralfermionization}, is that a VOSA $V=V_{\bar 0}\oplus V_{\bar 1}$ with $V_{\bar 1}\neq 0$ can always be obtained from a VOA $V_{\bar 0}$ equipped with a choice of simple module $V_{\bar 1}$ with $h(V_{\bar 1})\in \frac12+\mathbb{Z}_{\geq 0}$ and with fusion rule $V_{\bar 1}\boxtimes V_{\bar 1}\cong V_{\bar 0}$. We refer to a $V_{\bar 0}$-module $V_{\bar 1}$ with these two properties as a \emph{fermionic module}. 

Physically, one imagines starting with $V_{\bar 0}$, thought of as a gapless chiral boundary of a 3D TQFT. From the relationship between simple modules of a chiral algebra and anyons in the bulk TQFT, a choice of fermionic module $V_{\bar 1}$ determines a fermionic anyon $f$ in the bulk generating a $\mathbb{Z}_2^{(1)}$ one-form symmetry. The fact that $f$ is a fermion implies that this $\mathbb{Z}_2^{(1)}$ has a 't Hooft anomaly. 

Although the anomaly obstructs us from gauging $\mathbb{Z}_2^{(1)}$ to obtain a bosonic TQFT, one can instead apply the fermionization procedure to $f$. One first stacks $\mathrm{TQFT}_{V_{\bar 0}}$ with the unique invertible 3D spin TQFT with vanishing central charge, which we call $\mathrm{sVec}$. Since $\mathrm{TQFT}_{V_{\bar 0}}$ is placed on a spacetime with boundary, one needs to also choose a boundary condition for $\mathrm{sVec}$: there are two choices which are related by stacking with the Arf-invariant \cite{Karch:2019lnn}, i.e.\ the continuum limit of the Kitaev-Majorana chain \cite{Kitaev:2000nmw}. The precise choice will not matter in what follows.\footnote{The choice of boundary for $\mathrm{sVec}$ does not alter the VOSA $V$ we get from the fermionization procedure. Toggling from one choice of boundary to the other simply alters how $(-1)^F$ acts on the Ramond sector modules (i.e.\ $(-1)^F$-twisted modules) of $V$.} The theory $\mathrm{sVec}$ has a transparent fermion $\psi$ and, to fermionize, we simply gauge the $\mathbb{Z}_2$ one-form symmetry generated by $f\otimes \psi$.

The theory obtained after gauging is a spin TQFT in the bulk. The fermionic chiral algebra $V=V_{\bar 0}\oplus V_{\bar 1}$ lives on the boundary. This bulk/boundary system inherits a zero-form fermion parity symmetry $(-1)^F$ as the dual/quantum symmetry of the one-form gauging, which is implemented by a topological surface operator terminating on a line junction on the boundary. Gauging the fermion parity symmetry is sometimes called bosonization, and returns us back to the theory we started with. 

The fermionization and bosonization procedures are illustrated in Figure \ref{fig:chiralfermionization}. The upshot is that the classification of nice VOSAs at $c=1$ amounts to scanning over the nice $c=1$ VOAs classified by Theorem \ref{thm:c=1classification} and characterizing the inequivalent ways of fermionizing them. The result of this straightforward analysis is Theorem \ref{thm:c=1VOSAs}.

\begin{figure}
    \begin{center}

\begin{footnotesize}
\tikzset{every picture/.style={line width=0.75pt}}  
\begin{tikzpicture}[x=0.75pt,y=0.75pt,yscale=-1,xscale=1,scale=.8]
\draw  [draw opacity=0][fill={rgb, 255:red, 184; green, 233; blue, 134 }  ,fill opacity=.7 ][line width=1.5]  (71.2,116.9) -- (71.2,7) -- (167.2,54.1) -- (167.2,164) -- cycle ;
\draw [line width=1.5]  [dash pattern={on 1.69pt off 2.76pt}]  (11.5,7) -- (71.2,7) ;
\draw [line width=1.5]  [dash pattern={on 1.69pt off 2.76pt}]  (111.5,164) -- (167.2,164) ;
\draw [line width=1.5]  [dash pattern={on 1.69pt off 2.76pt}]  (11.5,116.9) -- (71.2,116.9) ;
\draw [line width=1.5]  [dash pattern={on 1.69pt off 2.76pt}]  (111.5,54.1) -- (167.2,54.1) ;
\draw    (202,152.5) -- (259.75,152.5) ;
\draw [shift={(261.75,152.5)}, rotate = 180] [color={rgb, 255:red, 0; green, 0; blue, 0 }  ][line width=0.75]    (10.93,-3.29) .. controls (6.95,-1.4) and (3.31,-0.3) .. (0,0) .. controls (3.31,0.3) and (6.95,1.4) .. (10.93,3.29)   ;
\draw  [draw opacity=0][fill={rgb, 255:red, 245; green, 166; blue, 35 }  ,fill opacity=.7 ][line width=1.5]  (377.2,115.9) -- (377.2,6) -- (473.2,53.1) -- (473.2,163) -- cycle ;
\draw [line width=1.5]  [dash pattern={on 1.69pt off 2.76pt}]  (317.5,6) -- (377.2,6) ;
\draw [line width=1.5]  [dash pattern={on 1.69pt off 2.76pt}]  (417.5,163) -- (473.2,163) ;
\draw [line width=1.5]  [dash pattern={on 1.69pt off 2.76pt}]  (317.5,115.9) -- (377.2,115.9) ;
\draw    (495,152) -- (552.75,152) ;
\draw [shift={(554.75,152)}, rotate = 180] [color={rgb, 255:red, 0; green, 0; blue, 0 }  ][line width=0.75]    (10.93,-3.29) .. controls (6.95,-1.4) and (3.31,-0.3) .. (0,0) .. controls (3.31,0.3) and (6.95,1.4) .. (10.93,3.29)   ;
\draw  [draw opacity=0][fill={rgb, 255:red, 184; green, 233; blue, 134 }  ,fill opacity=.7 ][line width=1.5]  (663.2,115.9) -- (663.2,6) -- (759.2,53.1) -- (759.2,163) -- cycle ;
\draw [line width=1.5]  [dash pattern={on 1.69pt off 2.76pt}]  (603.5,6) -- (663.2,6) ;
\draw [line width=1.5]  [dash pattern={on 1.69pt off 2.76pt}]  (703.5,163) -- (759.2,163) ;
\draw [line width=1.5]  [dash pattern={on 1.69pt off 2.76pt}]  (603.5,115.9) -- (663.2,115.9) ;
\draw [line width=1.5]  [dash pattern={on 1.69pt off 2.76pt}]  (703.5,53.1) -- (759.2,53.1) ;
\draw [color={rgb, 255:red, 74; green, 144; blue, 226 }  ,draw opacity=1 ][line width=1.5]    (17,51) -- (93.5,51) ;
\draw  [fill={rgb, 255:red, 0; green, 0; blue, 0 }  ,fill opacity=1 ] (90,51) .. controls (90,49.07) and (91.57,47.5) .. (93.5,47.5) .. controls (95.43,47.5) and (97,49.07) .. (97,51) .. controls (97,52.93) and (95.43,54.5) .. (93.5,54.5) .. controls (91.57,54.5) and (90,52.93) .. (90,51) -- cycle ;
\draw  [draw opacity=0][fill={rgb, 255:red, 80; green, 227; blue, 194 }  ,fill opacity=.7 ] (319.95,26.08) -- (375.81,26) -- (475.97,68) -- (420.1,68.08) -- cycle ;
\draw [color={rgb, 255:red, 74; green, 144; blue, 226 }  ,draw opacity=1 ][line width=1.5]    (377,26) -- (473.1,68) ;
\draw [line width=1.5]  [dash pattern={on 1.69pt off 2.76pt}]  (417.5,53.1) -- (473.2,53.1) ;
\draw (4,70.1) node [anchor=north west][inner sep=0.75pt]    {$\mathrm{TQFT}_{V_{\bar{0}}}$};
\draw (107.2,106.9) node [anchor=north west][inner sep=0.75pt]  [color={rgb, 255:red, 65; green, 117; blue, 5 }  ,opacity=1 ]  {$V_{\bar{0}}$};
\draw (191,124.4) node [anchor=north west][inner sep=0.75pt]    {$\text{fermionize}$};
\draw (195,46.1) node [anchor=north west][inner sep=0.75pt]    {$ \begin{array}{l}
\ \ \ \ \ \ \ \ \ \ \ \ \ \ \ \mathrm{TQFT}_{V}\\
\cong (\mathrm{TQFT}_{V_{\bar{0}}} \otimes \mathrm{sVec}) /\mathbb{Z}_{2}^{f\otimes \psi }
\end{array}$};
\draw (378.2,96.9) node [anchor=north west][inner sep=0.75pt]  [color={rgb, 255:red, 208; green, 2; blue, 27 }  ,opacity=1 ]  {$V\cong V_{\bar{0}} \oplus V_{\bar{1}}$};
\draw (527,52.1) node [anchor=north west][inner sep=0.75pt]    {$ \begin{array}{l}
\ \ \ \ \ \ \ \ \mathrm{TQFT}_{V_{\bar{0}}}\\
\cong \mathrm{TQFT}_{V} /\mathbb{Z}_{2}^{( -1)^{F}}
\end{array}$};
\draw (669.2,97.9) node [anchor=north west][inner sep=0.75pt]  [color={rgb, 255:red, 65; green, 117; blue, 5 }  ,opacity=1 ]  {$V_{\bar{0}} \cong V^{( -1)^{F}}$};
\draw (489,126) node [anchor=north west][inner sep=0.75pt]   [align=left] {bosonize};
\draw (33,30.4) node [anchor=north west][inner sep=0.75pt]  [color={rgb, 255:red, 74; green, 144; blue, 226 }  ,opacity=1 ]  {$f$};
\draw (354.8,36.4) node [anchor=north west][inner sep=0.75pt]  [color={rgb, 255:red, 74; green, 92; blue, 226 }  ,opacity=1 ]  {$( -1)^{F}$};
\draw (84.2,25.9) node [anchor=north west][inner sep=0.75pt]  [color={rgb, 255:red, 74; green, 144; blue, 226 }  ,opacity=1 ]  {$V_{\bar{1}}$};
\end{tikzpicture}
\end{footnotesize}
    \end{center}
\caption{The bulk TQFT supporting $V_{\bar 0}$ on its boundary is assumed to have a fermionic anyon $f$. The boundary local operators on which it terminates belong to the Hilbert space $V_{\bar 1}$. Fermionizing entails stacking with $\mathrm{sVec}$, which has a transparent fermion $\psi$, and gauging the $\mathbb{Z}_2$ one-form symmetry generated by $f\otimes \psi$. The bulk after fermionizing is a spin TQFT, and the boundary is a fermionic chiral algebra $V\cong V_{\bar 0}\oplus V_{\bar 1}$. The fermionized theory has a fermion parity zero-form symmetry denoted $(-1)^F$. Fermionization can be undone by bosonization, which entails gauging $(-1)^F$. The effect on the boundary is to project out the operators which are odd under $(-1)^F$.}\label{fig:chiralfermionization}
\end{figure}
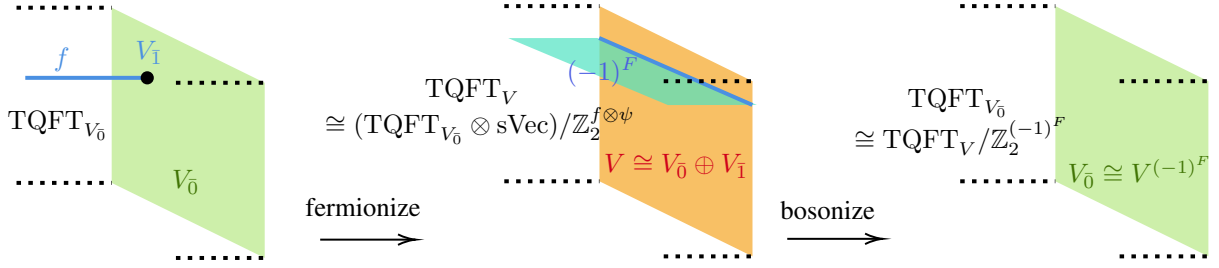

\subsection{Outline}

The remainder of this paper is organized as follows.

In Section \ref{sec:generalities}, we explain what we mean by a ``nice'' VOA. We also introduce Tannakian VOAs, and establish some basic theorems  about them (most notably, Theorem \ref{thm:TannakianVacChar}) which will be used in our proof of the classification of nice $c=1$ VOAs.

Section \ref{sec:classification} is dedicated to this classification.  In Section \ref{subsec:bosonicassump}, we carefully explain the assumptions behind and the statement of the main result (Theorem \ref{thm:c=1classification}), and also make some nearby observations. In Section \ref{subsec:mainproof}, we present the proof.

Section \ref{sec:VOSAclassification}  describes how to obtain the classification of nice $c=1$ vertex operator \emph{super}algebras, Theorem \ref{thm:c=1VOSAs}, as a corollary of the classification of nice $c=1$ vertex operator algebras.

Section \ref{sec:future} offers several directions for future research. It is followed by a number of technical appendices. Appendix \ref{app:bestiary} records the basic data of the nice $c=1$ VOAs: their modular data, their simple modules, their characters, etc. Appendix \ref{app:exceptional} presents a derivation of the modular data of the exceptional $c=1$ VOAs $V_T$, $V_O$, and $V_I$. In Appendix \ref{app:automorphisms}, we derive some technical results related to the automorphism groups of $c=1$ VOAs which are used in various places, both in this work and in our companion paper \cite{grCFTs}. Appendix \ref{app:QuOpTopLin} characterizes (at a physics level of rigor) the non-invertible symmetries of the $c=1$ VOAs, using ideas from our previous work \cite{Gannon:2026ttf}. (The results of Appendix \ref{app:QuOpTopLin} are not used anywhere in the main text, but are included to complement our discussion of the automorphism groups of $c=1$ VOAs in Appendix \ref{app:automorphisms}.)  \\

\noindent\textbf{Note added:} We sincerely thank Sebastiano Carpi, Tiziano Gaudio, and Luca Giorgetti for coordinating the submission of their related paper \cite{carpietal} to the arXiv on the same day. The authors obtain similar results to ours using complementary methods.

\section{Generalities}\label{sec:generalities}

We begin by establishing some generalities. In Section \ref{subsec:nice}, we explain what we mean by a ``nice'' vertex operator algebra, and also review some of their properties. In Section \ref{subsec:tannakian}, we introduce the notion of a \emph{Tannakian} vertex operator algebra. Tannakian technology will enter prominently in our proof of the classification of nice $c=1$ VOAs, and we expect that the methods we develop here will have applications beyond $c=1$ as well. 

\subsection{Nice vertex operator algebras}\label{subsec:nice}

First, we define the notion of strong rationality.

\begin{definition}
    A VOA is said to be strongly rational if it is simple, rational, $C_2$-cofinite, self-contragredient, and of CFT-type. 
\end{definition}
A physicist loses nothing by imagining a rational chiral algebra. For mathematicians, we recall that simple means that $V$ is simple when viewed as a module over itself. Rational means that the representation theory is semisimple: more precisely, every admissible $V$-module is a direct sum of simple admissible $V$-modules (see e.g.\ \cite{dong1997regularity} for the definition of admissible). The $C_2$ condition, first introduced in \cite{Zhu:1996gaq}, requires that $\dim V/C_2(V)<\infty$, where 
\begin{align}
    C_2(V)=\mathrm{Span}_{\mathbb{C}}\{u_{-2}v\mid u,v\in V\},
\end{align}
in conventions where the operator $Y(u,z)$ associated to a state $u\in V$ is moded as $Y(u,z)=\sum_{n\in\mathbb{Z}}u_nz^{-n-1}$. Self-contragredient means that $V$ is isomorphic to its contragredient module, defined e.g.\ in \cite{frenkel1993axiomatic}. Finally, CFT-type means that the $L_0$ grading on $V$ takes the form
\begin{align}
    V=\bigoplus_{n\geq 0}V_n, \ \ \ \ V_0=\mathbb{C}|\Omega\rangle,
\end{align}
i.e.\ there is a unique vacuum and all operators have non-negative conformal dimension.

A strongly rational VOA has a very nice representation theory. In fact, its representation category $\Rep(V)$ admits the structure of a modular fusion category \cite{Moore:1988qv,Huang:2005gs}.
In particular, there are finitely many simple modules, which we denote $V_a$. The vacuum corresponds to $a=0$, i.e.\ $V\equiv V_{0}$. We write $h_a$ for the conformal dimension of  $V_a$. The Hopf link invariant (un-normalized S-matrix) and quantum dimensions are defined as
\begin{align}\label{eqn:hopflink}
    \tilde{s}_{ab} = \mathrm{tr}_q(c_{V_b,V_a}\circ c_{V_a,V_b}), \ \ \ \ \ d_a=\tilde{s}_{a0},
\end{align}
and the chiral central charge as 
\begin{align}
    e^{\frac{2\pi \ri}{8} c(\Rep(V))}=\frac{1}{\sqrt{\mathrm{Dim}(\Rep(V))}}\sum_{a\in \Rep(V)}d_a^2 e^{2\pi \ri h_a}, \ \ \ \ \mathrm{Dim}(\Rep(V)) :=\sum_{a\in \Rep(V)}d_a^2.
\end{align}

We will make crucial use of the following well-known modular invariance properties of the graded-dimensions of simple modules of strongly rational VOAs \cite{Zhu:1996gaq,Coste:1999yc,Ng:2010win,dong2015congruence} (see also Proposition 3.5 and Lemma 3.6 of \cite{Moller:2024plb}).

\begin{theorem}\label{thm:modularity}
    Let $V$ be a strongly rational VOA of central charge $c$ with simple modules $V_a$. The characters 
    \begin{align}
        \mathrm{ch}_{V_a}(\tau) := \mathrm{Tr}_{V_a}q^{L_0-c/24}
    \end{align}
    transform in a representation of $\textsl{SL}_2(\mathbb{Z})$,
    \begin{align}\label{eqn:STtransform}
        \mathrm{ch}_{V_a}(-1/\tau) = \sum_{b\in \Rep(V)}S_{ab}\mathrm{ch}_{V_b}(\tau), \ \ \ \ \ \mathrm{ch}_{V_a}(\tau+1)=e^{2\pi \ri (h_a-c/24)}\mathrm{ch}_{V_a}(\tau).
    \end{align}
    The S-matrix $S_{ab}$ agrees with the Hopf link invariant of $\Rep(V)$ in \eqref{eqn:hopflink},
    \begin{align}
        S_{ab}/S_{00}=\tilde{s}_{ab},
    \end{align}
    and the normalization $S_{00}$ is given by
    \begin{align}
        S_{00}^{-1}=\epsilon \sqrt{\mathrm{Dim}(\Rep(V))}, \ \ \ \ \ \ \ \ \epsilon=\begin{cases}
            +1,& c-c(\Rep(V))=0~\mathrm{mod}~8\\
            -1, & c-c(\Rep(V))=4~\mathrm{mod}~8.
        \end{cases}
    \end{align}
    If $\Rep(V)$ is pseudo-unitary, then $\epsilon=+1$.
    \end{theorem}
    \begin{theorem}\label{thm:congruence}
    Each character $\mathrm{ch}_{V_a}(\tau)$ is individually invariant under the action of 
    \begin{align}
        \Gamma(C_V)=\left\{ \left(\begin{smallmatrix}\alpha & \beta \\ \gamma & \delta \end{smallmatrix}\right)\in \textsl{SL}_2(\mathbb{Z})\mid \left(\begin{smallmatrix}\alpha & \beta \\ \gamma & \delta \end{smallmatrix}\right)\equiv \left(\begin{smallmatrix}1 & 0 \\ 0 & 1\end{smallmatrix}\right)~\mathrm{mod}~C_V\right\},
    \end{align}
    where $C_V$, called the conductor of $V$, is defined to be the least common multiple of the denominators of $h_a-c/24$.
\end{theorem}

We make the further restriction in much of what follows that the conformal dimensions of all (non-vacuum) simple modules are positive.\footnote{In \cite{Moller:2024xtt}, pseudo-unitary VOAs were called \emph{positive}.}

\begin{definition}
A simple VOA is said to be pseudo-unitary if all simple modules not isomorphic to the vacuum have strictly positive conformal dimension. 
\end{definition}

The reason we have chosen this non-standard name is that, if $V$ is a pseudo-unitary and strongly rational VOA, then $\Rep(V)$ has only positive categorical dimensions, and hence is a pseudo-unitary modular fusion category (see Equation (2.10) of \cite{Dijkgraaf:1988tf} or Lemma 4.2 of \cite{dong2013quantum}). We also record the following partial converse for later use.

\begin{proposition}\label{prop:pseudounitarypartialconverse}
    If $V$ is strongly rational and $\Rep(V)$ is pseudo-unitary, then $h_a\geq 0$ for every simple $V$-module $V_a$.
\end{proposition}

\begin{proof}
    Let $h_{\mathrm{min}}$ be the minimum taken over the conformal dimensions $h_a$ of the simple modules $V_a$, and let $B=\{b\mid h(V_b)=h_{\mathrm{min}}\}$ index the $V$-modules with conformal dimension equal to $h_{\mathrm{min}}$.  Then, using Theorem \ref{thm:modularity}, we have that
    \begin{align}
        \lim_{t\to\infty}e^{-2\pi t\tilde{c}/24}\mathrm{ch}_{V_a}(\ri/ t)= v_a:= \sum_{b\in B}S_{ab}\dim((V_{b})_{h_{\mathrm{min}}}),
    \end{align}
    where $\tilde{c}=c-24 h_{\mathrm{min}}$ is the effective central charge. Because $\mathrm{ch}_{V_a}(\ri/t)>0$ for every $t>0$, we must have that $v_a\geq 0$. For $a=0$, pseudo-unitarity implies that the inequality is strict,
    \begin{align}
        v_0=S_{00}\sum_{b\in B}d_b \mathrm{dim}((V_b)_{h_{\mathrm{min}}})>0.
    \end{align}
    On the other hand, we can calculate that 
    \begin{align}
    \begin{split}
        0<\sum_a v_a S_{a0}&=\sum_{a}
        \sum_{b\in B}S_{ab}S_{a0}\dim((V_b)_{h_{\mathrm{min}}})=\sum_{b\in B}(S^2)_{b0}\dim((V_b)_{h_{\mathrm{min}}}) \\
        &=\begin{cases}
            \dim((V_0)_{h_{\mathrm{min}}}), & 0\in B \\
            0, & 0\notin B.
        \end{cases}
    \end{split}
    \end{align}
    In going from the first line to the second line, we have used that $S$ squares to the charge conjugation matrix, $(S^2)_{ab}=\delta_{ab^\ast}$. We reach a contradiction if $0\notin B$, hence $h_{\mathrm{min}}=0$.
\end{proof}

Another reason for the name is that unitary vertex operator algebras (with unitarizable modules) are automatically pseudo-unitary, so the prefix ``pseudo'' nicely reflects the fact that we are imposing a condition which is weaker than unitarity + unitarizable modules. 

We comment that we are not aware of any examples of pseudo-unitary strongly rational VOAs which are not also unitary (with unitarizable modules). We nevertheless make the milder assumption of pseudo-unitarity in what follows in case the class of pseudo-unitary VOAs turns out to be larger than the class of unitary ones (with unitarizable modules).

In order to avoid having to repeat all these hypotheses, we make the following definition. 

\begin{definition}\label{def:niceVOA}
    A VOA is said to be nice if it is strongly rational and pseudo-unitary.
\end{definition}

As we mentioned in the introduction, there are three known families of nice $c=1$ VOAs: lattice VOAs, charge conjugation orbifolds, and exceptional models. These are reviewed in detail in Appendix \ref{app:bestiary}. The proof that these are \emph{all} nice $c=1$ VOAs is the topic of Section \ref{sec:classification}.

Nice VOAs have an extremely well-controlled extension theory which can be formulated in the language of (tensor) categories. We will use the following synthesis of the results of \cite{Kirillov:2001ti,Huang:2014ixa,Creutzig:2024qvb} (summarized in Proposition 2.19 of \cite{Moller:2024xtt} in a form which will be more directly useful for our purposes).

\begin{proposition}\label{prop:ExtAlg}
    Let $V$ be a nice VOA. Then there is a bijection between:
    \begin{enumerate}[label=(\arabic*)]
        \item (equivalence classes of) conformal VOA extensions $W\supset V$ such that $W$ is simple (and hence nice),
        \item (equivalence classes of) condensable algebras $A$ in $\Rep(V)$.
    \end{enumerate}
    As an object of $\Rep(V)$, the algebra $A$ is the restriction of $W$, thought of as a $V$-module, to $V$. The representation category of a conformal extension $W\supset V$ is given by $\Rep(W)\cong \Rep(V)_A^{\mathrm{loc}}$, the category of local $A$-modules in $\Rep(V)$. 
\end{proposition}

\noindent We refer readers to \cite[Section 2.2]{Moller:2024xtt} and references therein for further details on the notion of a condensable algebra in a modular fusion category, and on the notion of a (local) $A$-module. 

There is also a useful interpretation of the category $\Rep(V)_A$ of $A$-modules in $\Rep(V)$ which are not necessarily local. We provide this interpretation in the case that $V=W^G$ for some finite group $G$ of automorphisms of $W$. (A physical perspective on the more general case can be found in \cite{Gannon:2026ttf}.)

Recall that, for any $g\in\mathrm{Aut}(W)$, there is a notion of a $g$-twisted module of $W$ (see e.g.\ \cite{dong1998twisted} for the definition). We use $\Rep_g(W)$ to denote the category of $g$-twisted modules of $W$, and write 
\begin{align}
    \Rep_G(W):=\bigoplus_{g\in G}\Rep_g(W), \ \ \ \ \ \ \Rep_1(W)=\Rep(W).
\end{align} 
If $W^G$ is strongly rational, then $\Rep_G(W)$ is a $G$-crossed ribbon tensor category in the sense of \cite{Turaev:2000ug}. We also have the following identification.
\begin{proposition}\label{prop:equivAmod}
Let $V$ be nice, and let $A$ be the condensable algebra corresponding to the conformal extension $W\supset W^G=V$. Then
\begin{align}\label{eqn:GtwistedAmodules}
    \Rep_G(W)\cong \Rep(V)_A,
\end{align}
i.e.\ $g$-twisted modules of $W$ correspond to not-necessarily-local $A$-modules in $\Rep(V)$. 
\end{proposition}
\noindent See e.g.\ \cite{Mcrae:2019pol} for further background on the twisted representation theory of VOAs. 

\bigskip

There is a slightly broader class of VOAs that will be useful to consider. We recall the following definition from \cite{Creutzig:2016fms}.

\begin{definition}
A VOA is said to be strongly finite if it is simple, $C_2$-cofinite, self-contragredient, and of CFT-type.
\end{definition}

That is, strongly finite VOAs satisfy all the same assumptions as strongly rational VOAs, except that rationality is dropped. Such VOAs are sometimes referred to as \emph{logarithmic}.  A strongly finite VOA retains the nice property of having finitely many simple modules but, if it is not strongly rational, often has uncountably many inequivalent indecomposable modules. The triplet algebra $W_p$ is an example. See e.g.\ \cite{Creutzig:2016fms} for a review.

We will use the following lemma in Section \ref{sec:classification}, which roughly says that properly logarithmic VOAs do not exist if one assumes rigidity and pseudo-unitarity of the representation category.

\begin{lemma}\label{lem:C2+pseudounitary}
    Suppose $V$ is strongly finite and that the category $\Rep(V)$ of grading-restricted generalized $V$-modules is rigid and pseudo-unitary. Then $V$ is strongly rational.
\end{lemma}

\begin{proof}
    By Main Theorem 1 of \cite{McRae:2021yyb}, $\Rep(V)$ will be a (not necessarily semisimple) modular tensor category, i.e.\ a factorizable finite ribbon category. Suppose $\Rep(V)$ is not semisimple. Then from the proof of Theorem 2.16 of \cite{Etingof:2004ftc}, every projective object in $\Rep(V)$ will have categorical dimension 0. However Proposition 4.7.5 in \cite{etingof2015tensor} implies that the categorical dimension of any object in $\Rep(V)$ is the sum of the categorical dimensions of the simple composition factors. This contradicts pseudo-unitarity. Thus $\Rep(V)$ is semisimple, so $V$ is strongly rational. 
\end{proof}

We note that strongly finite VOAs are expected to always have rigid representation categories, so it is likely that Lemma \ref{lem:C2+pseudounitary} can eventually be improved by dropping the assumption of rigidity. In any case, we will use it to relax the hypotheses of our main result, Theorem \ref{thm:c=1classification}, in Section \ref{sec:classification}.

\subsection{Tannakian VOAs}\label{subsec:tannakian}

Next, we introduce and study the notion of a Tannakian VOA.

We first briefly review some category-theoretic preliminaries. Recall that a braided fusion category $\mathcal{C}$ is said to be symmetric if $c_{Y,X}\circ c_{X,Y}=\mathrm{id}_{X\otimes Y}$, where $c_{X,Y}:X\otimes Y\to Y\otimes X$ is the braiding. 

Deligne's theorem \cite{deligne2002categories} says that any symmetric fusion category is braided tensor equivalent to $\Rep(G,z)$ where $G$ is a finite group and $z$ is a central element satisfying $z^2=1$. (Taking $z$ to be the identity element is allowed.) The braided fusion category $\Rep(G,z)$ is the same as $\Rep(G)$ as a fusion category, but its braiding is twisted by the central element $z$. To construct this braiding, note that a choice of $z$ endows $\Rep(G)$ with a $\mathbb{Z}_2$-grading: every representation decomposes as $R=R_{\bar 0}\oplus R_{\bar 1}$, where $R_{\bar a}$ is the subrepresentation of $R$ on which $z$ acts as $(-1)^a$. Then, the $z$-twisted braiding is defined as 
\begin{align}
    c^z_{R,R'}(v\otimes v') = (-1)^{ab}(v'\otimes v), \ \ v\in R_{\bar a}, \ v'\in R_{\bar b},
\end{align}
extended linearly to vectors $v,v'$ which are not homogeneous. A symmetric fusion category $\mathcal{C}$ is said to be \emph{Tannakian} if it is braided tensor equivalent to $\Rep(G):=\Rep(G,z=1)$, and it is said to be \emph{super Tannakian} if it is braided tensor equivalent to $\Rep(G,z)$ with $z\neq 1$.  

There is a useful criterion for testing whether a symmetric fusion category is Tannakian or super Tannakian. Suppose that $\mathcal{C}$ is a symmetric fusion category equipped with a choice of ribbon structure, which can be encoded in a complex number $\theta_X$ (called the twist) for every simple object $X$. In a braided fusion category, every ribbon structure defines a spherical structure (see e.g.\ \cite[Section 2.8.2]{drinfeld2010braided}) so we have a notion of categorical dimension $d_X$. By Remark 2.33 of \cite{drinfeld2010braided}, these two structures are related to the braiding as
\begin{align}
    \mathrm{Tr}[c_{X,X}]=\theta_X d_X.
\end{align}
In $\Rep(G,z)$, for an irreducible representation $R$, one has $\mathrm{Tr}[c_{R,R}]=\chi_R(z)$ where $\chi_R$ is the character of the representation $R$. Thus, we get the following.

\begin{proposition}\label{prop:Tannakiancriterion}
    If $\mathcal{C}$ is a symmetric fusion category equipped with a ribbon structure (and hence also a spherical structure), then $\mathcal{C}$ is Tannakian if and only if $\theta_X\cdot\mathrm{sign}(d_X)=1$ for every simple object $X\in\mathcal{C}$.
\end{proposition}
In what follows, the categories to which we will apply Proposition \ref{prop:Tannakiancriterion}  are subcategories of modular fusion categories, and hence inherit a ribbon structure.

With these categorical preliminaries out of the way, we are ready to define what we mean by a Tannakian VOA.

\begin{definition}
    A strongly rational vertex operator algebra $V$ is said to be Tannakian if the full subcategory $\Rep(V)_{h\in\mathbb{Z}}$ of $V$-modules with integer conformal dimension is a Tannakian symmetric fusion subcategory of $\Rep(V)$, i.e.\ if $\Rep(V)_{h\in\mathbb{Z}}$ is braided tensor equivalent to $\Rep(G)$ for some finite group $G$. 
\end{definition}
\noindent We remark that, thanks to Tannakian reconstruction, two Tannakian fusion categories $\Rep(G)$ and $\Rep(G')$ are braided tensor equivalent if and only if $G\cong G'$. Thus, if $V$ is Tannakian, one can unambiguously extract a finite group $G$ from it.

The main property of a Tannakian VOA $V$ that we will exploit is that it comes with a unique maximal conformal extension $W\supset V$ from which $V$ can be recovered by passing to the fixed points of a group $G$ of automorphisms, $V=W^G$. In fact, we have the following equivalent characterization of a Tannakian VOA.

\begin{proposition}\label{prop:Tannakianext}
    Suppose $V$ is a nice VOA. Then $V$ is Tannakian with $\Rep(V)_{h\in\mathbb{Z}}\cong \Rep(G)$ if and only if there is a nice VOA $W$ with the following properties. 
    \begin{enumerate}[label=(\arabic*)]
    \item There is a finite subgroup of $ \mathrm{Aut}(W)$ isomorphic to $G$ such that $W^G\cong V$. 
    \item If $M\in \Rep_G(W)$ is a simple, non-vacuum, twisted $W$-module, then there are no states in $M$ with integer conformal dimension. 
    \end{enumerate}
\end{proposition}

\begin{proof}
Suppose $V$ is Tannakian. Then $\Rep(V)_{h\in\mathbb{Z}}\cong \Rep(G)$, being a Tannakian subcategory of $\Rep(V)$, defines a (unique up to isomorphism) condensable algebra $A$ of $\Rep(V)$ of the form
\begin{align}
    A:=\bigoplus_{a\in \Rep(V)_{h\in\mathbb{Z}}}\mathbb{C}^{d_a}\otimes  V_a\in \Rep(V)_{h\in\mathbb{Z}}\subset \Rep(V),
\end{align}
called the regular algebra (see e.g.\ Definition 2.51 of \cite{drinfeld2010braided}). By Proposition \ref{prop:ExtAlg},  $A$ defines a VOA extension $W\supset V$ with $\Rep(W)=\Rep(V)_A^{\mathrm{loc}}$. 

Any VOA extension of $V$ is mediated by a condensable algebra supported in $\Rep(G)\cong \Rep(V)_{h\in\mathbb{Z}}\subset \Rep(V)$. Every condensable algebra in $\Rep(G)$ is a function algebra of the form $\mathrm{Fun}(G/H)$ for some subgroup $H\subset G$, and all of these occur as subalgebras of $A=\mathrm{Fun}(G)$. It follows that $W$ is the unique maximal conformal extension of $V$.

Note that every automorphism of the algebra $A$ in $\Rep(V)$ can be promoted to an automorphism of the VOA $W$ fixing $V\subset W$. It is known that $\mathrm{Aut}(A)\cong G$ (see \cite[Example 2.8]{davydov2013witt}) and that in fact $W^G\cong V$. By Proposition \ref{prop:equivAmod}, we then have that $\Rep_G(W)\cong \Rep(V)_A$. 

 Let $\mathrm{Res}(M)$ denote the restriction to $V$ of a $g$-twisted $W$-module $M$. Suppose that $M$ is simple and $\mathrm{Res}(M)$ contains a simple $V$-module $V_a$ with integer conformal weight, i.e.\ $V_a\in \Rep(V)_{h\in \mathbb{Z}}$. Then, by Frobenius reciprocity (see e.g.\ Theorem 1.6 of \cite{Kirillov:2001ti}), $M\subset  \mathrm{Ind}(V_a)$, where $\mathrm{Ind}:\Rep(V)\to\Rep(V)_A$ is the induction functor. But Proposition 3.2 (iii) of \cite{Muger2010BraidedCrossedG} says that $\mathrm{Ind}(V_a)=\mathbb{C}^{d_a}\otimes A$, and therefore $M=A=W$. This finishes the forward direction.

Now, suppose that there is a nice VOA $W$ satisfying properties (1) and (2) of the proposition. By assumption, $V=W^G$ and $V$ is strongly rational. By Frobenius reciprocity, every simple $V$-module occurs inside (the restriction of) some $g$-twisted module of $W$ for some $g\in G$. Since by assumption the vacuum is the only twisted module of $W$ involving integer conformal dimensions, it follows that every simple $V$-module in $\Rep(V)_{h\in\mathbb{Z}}$ occurs as a direct summand of $\mathrm{Res}(W)$. By invoking \cite{dong1996compact,McRae:2018wpu}, we learn that the simple $V$-modules arising in $\mathrm{Res}(W)$ (all of which have integer conformal dimension) are in one-to-one correspondence with irreducible representations of $G$, and in fact they span a symmetric fusion subcategory of $\Rep(V)$ which is equivalent to $\Rep(G)$. In other words, $\Rep(V)_{h\in\mathbb{Z}}\cong \Rep(G)$ and so $V$ is Tannakian.
\end{proof}

We henceforth write $V_{\Rep(G)}$ for the VOA $W$ attached to a Tannakian VOA $V$ by Proposition \ref{prop:Tannakianext}. We highlight that $V_{\Rep(G)}$ necessarily decomposes into $V$-modules as 
\begin{align}
    V_{\Rep(G)}\cong \bigoplus_{a\in \Rep(V)_{h\in\mathbb{Z}}}\mathbb{C}^{d_a}\otimes V_a,
\end{align}
where $d_a=S_{a0}/S_{00}$ is the quantum dimension of $V_a$, which is an integer thanks to the fact that $V$ is Tannakian. In particular, the vacuum character of $V_{\Rep(G)}$ is
    \begin{align}\label{eqn:chartannext}
\mathrm{ch}_{V_{\Rep(G)}}(\tau)=\sum_{a\in\Rep(V)_{h\in\mathbb{Z}}}d_a\mathrm{ch}_{V_a}(\tau).
    \end{align}

    We also remark that, since $V_{\Rep(G)}$ in particular does not have any non-vacuum module with integer conformal dimension, it follows that $V_{\Rep(G)}$ is a ``maximal'' conformal extension of $V$ in the sense that it cannot be extended further. In fact, it is the unique maximal conformal extension of $V$, as we demonstrated in the proof.

 All of the nice $c=1$ VOAs, cataloged in Appendix \ref{app:bestiary}, are Tannakian. Our ultimate goal is to show that this is a structural fact about nice $c=1$ VOAs.

 \begin{example}
     Uniquely decompose $m=t_\star n^2$, where $n$ is the largest integer such that $n^2$ divides $m$, and $t_\star=m/n^2$ is  square-free. Recall  that $V_{2m}$ denotes the VOA associated to the rank-1 even integral lattice $\sqrt{2m}\mathbb{Z}$, and that $V_{2m,r}$ denotes the simple $V_{2m}$-module associated to the coset $\sqrt{2m}\mathbb{Z}+r/\sqrt{2m}$ (cf.\ Appendix \ref{app:bestiary} for details). The VOA $V_{2m}$ is Tannakian with 
     \begin{align}
         \Rep(V_{2m})_{h\in\mathbb{Z}} \cong \Rep(\mathbb{Z}_n),
     \end{align}
     and the simples of $\Rep(V_{2m})_{h\in\mathbb{Z}}$ are given by $V_{2m,2t_\star n\ell}$, with  $\ell\in\mathbb{Z}_n$. The maximal conformal extension of $V_{2m}$ guaranteed by Proposition \ref{prop:Tannakianext} is the lattice VOA $V_{2t_\star}$. 
     \end{example}

     \begin{example}
     As before, decompose $m=t_\star n^2$. The charge conjugation orbifold $V_{2m}^+$ is Tannakian with 
     \begin{align}
         \Rep(V_{2m}^+)_{h\in\mathbb{Z}}\cong \Rep(D_{2n}),
     \end{align}
     where $D_{2n}$ is the dihedral group of order $2n$. Its maximal conformal extension is again $V_{2t_\star}$. 
     \end{example}
     \begin{example}
     The three exceptional VOAs $V_T=V_2^{A_4}$, $V_O=V_2^{S_4}$, and $V_I=V_2^{A_5}$ are Tannakian with 
     \begin{align}
         \Rep(V_T)_{h\in\mathbb{Z}}\cong \Rep(A_4), \ \ \ \Rep(V_O)_{h\in\mathbb{Z}}\cong \Rep(S_4), \ \ \ \ \Rep(V_I)_{h\in\mathbb{Z}}\cong \Rep(A_5),
     \end{align}
     and the maximal conformal extension of all three of these VOAs is $V_2\cong \widehat{\mathfrak{su}}(2)_1$. 
 \end{example}

We will give two ways of detecting whether a VOA is Tannakian: one using the modular S-matrix, and the other using the vacuum character. 

\begin{proposition}\label{prop:TannakianS}  A nice VOA $V$ is Tannakian if and only if 
\begin{align}
    S_{ab}=S_{00}d_ad_b, \ \ \ \ \text{for all } V_a,V_b\in \Rep(V)_{h\in\mathbb{Z}}
\end{align}
where $d_a=S_{0a}/S_{00}$ is the quantum dimension of $V_a$, and $S$ is the modular S-matrix of $V$.
\end{proposition}

\begin{proof} Suppose $S_{ab}=S_{00}d_ad_b$ for all the integer conformal dimension modules. By Proposition 2.5 of \cite{muger2003structure}, this means that 
\begin{align}
    c_{ba}\circ c_{ab} = \mathrm{id}_{a\otimes b}, \ \ \ \ \text{for all } V_a,V_b\in \Rep(V)_{h\in\mathbb{Z}},
\end{align}
where $c_{ab}$ is the braiding on $\Rep(V)$.
On the other hand, the balancing axiom of ribbon categories says that 
\begin{align}
    \theta_{a\otimes b}=c_{ba}\circ c_{ab}\circ (\theta_a\otimes \theta_b),
\end{align}
where $\theta_a=e^{2\pi \ri h_a}$ is the twist on $\Rep(V)$. 
In particular, when $V_a,V_b\in \Rep(V)_{h\in\mathbb{Z}}$, we have that 
\begin{align}
    \theta_{a\otimes b}=\theta_a\otimes \theta_b = 1.
\end{align}
That is, $\Rep(V)_{h\in\mathbb{Z}}$ closes under the tensor product, and in fact forms a symmetrically braided fusion subcategory of $\Rep(V)$. 

By Deligne's theorem \cite{deligne2002categories}, this means that $\Rep(V)_{h\in\mathbb{Z}}\cong \Rep(G)$ or $\Rep(V)_{h\in\mathbb{Z}}\cong \Rep(G,z)$.  Whether one is in the Tannakian case or the super Tannakian case can be determined by applying Proposition \ref{prop:Tannakiancriterion}, i.e.\ by checking whether $\mathrm{sign}(d_M)\theta_M=1$ for all simples $M\in\Rep(V)_{h\in\mathbb{Z}}$. And indeed, all twists are trivial since $\theta_M=e^{2\pi \ri h_M}=1$ for $M\in \Rep(V)_{h\in\mathbb{Z}}$. Furthermore, because $V$ is pseudo-unitary, all quantum dimensions are positive as well, $\mathrm{sign}(d_M)=1$. Thus, we conclude that $\Rep(V)_{h\in\mathbb{Z}}\cong \Rep(G)$.

In the reverse direction, $V$ being Tannakian implies that $c_{ba}\circ c_{ab}=\mathrm{id}_{a\otimes b}$ for all $V_a,V_b\in\Rep(V)_{h\in\mathbb{Z}}$ and hence, by Proposition 2.5 of \cite{muger2003structure}, we have that $S_{ab}=S_{00}d_ad_b$.
\end{proof}

More surprisingly, it turns out that one can determine whether a VOA is Tannakian directly from its vacuum character. To prepare for the statement, let $V$ be a nice VOA and let $\mathrm{ch}_V(\tau)$ be its vacuum character, defined in \eqref{eqn:vaccharV}. We note that we can extract $S_{00}$ just from the vacuum character as 
\begin{align}
    \mathrm{ch}_V(-1/\tau)=S_{00}q^{-c/24}+\cdots, \ \ \ \ S_{00}^{-2}=\sum_{a\in \Rep(V)} d_a^2.
\end{align}
We can also extract the following function directly from the vacuum character,
\begin{align}\label{eqn:PVdef}
    P_V(\tau) := \frac{1}{S_{00}}\Pr\left[\mathrm{ch}_V(-1/\tau)\right],
\end{align}
where, given $f(\tau)=\sum_{x\in\mathbb{Q}}C(x)q^x$, we define 
\begin{align}\label{eqn:PiZ}
    \Pr[f(\tau)]:= \sum_{n\in\mathbb{Z}}C({n-c/24})q^{n-c/24}
\end{align} 
as the $q$-expansion obtained by projecting onto the contributions of states with integer conformal dimension. The function $P_V(\tau)$ by construction detects the characters of all the modules with integer conformal dimension, 
\begin{align}
    P_V(\tau)=\sum_{a\in \Rep(V)_{h\in\mathbb{Z}}}d_a\mathrm{ch}_{V_a}(\tau),
\end{align}
though in general it is not possible to extract the individual characters $\mathrm{ch}_{V_a}(\tau)$ and quantum dimensions $d_a$ from $P_V(\tau)$.

\begin{theorem}\label{thm:TannakianVacChar} A nice VOA $V$ is Tannakian with $\Rep(V)_{h\in\mathbb{Z}}\cong \Rep(G)$ for some finite group $G$ if and only if 
\begin{align}\label{eqn:Tannakiancondition}
    \frac{1}{S_{00}}\Pr[P_V(-1/\tau)]=\DD P_V(\tau),
\end{align}
for some constant $\DD$, in which case 
\begin{align}
    \DD=\sum_{a\in\Rep(V)_{h\in\mathbb{Z}}}d_a^2=|G|.
\end{align} 
\end{theorem}

\begin{proof} 
Suppose that Equation \eqref{eqn:Tannakiancondition} holds. Using \eqref{eqn:STtransform}, we learn that
\begin{align}\label{eqn:unpacked}
    \sum_{b\in\Rep(V)_{h\in\mathbb{Z}}}\left(\DD S_{00}d_b-\sum_{a\in\Rep(V)_{h\in\mathbb{Z}}}d_aS_{ab}\right)\mathrm{ch}_{V_b}(\tau)=0.
\end{align}
By evaluating the coefficient of $q^{-c/24}$ in this equation, and using the fact that $S_{a0}/S_{00}=d_a$, we find that 
\begin{align}
    \DD S_{00}=\sum_{a\in\Rep(V)_{h\in\mathbb{Z}}}d_a S_{a0}\implies \DD = \sum_{a\in\Rep(V)_{h\in\mathbb{Z}}}d_a^2.
\end{align}
On the other hand, we have the following elementary inequality, 
\begin{align}\label{eqn:qdiminequality}
\begin{split}
    \mathrm{Re}\left(\sum_{a\in\Rep(V)_{h\in\mathbb{Z}}} d_a S_{ab}\right) &= \sum_{a\in\Rep(V)_{h\in\mathbb{Z}}}d_a\mathrm{Re}S_{ab} \leq \sum_{a\in\Rep(V)_{h\in\mathbb{Z}}}d_a|S_{ab}| \\
    &\leq S_{00}\sum_{a\in\Rep(V)_{h\in\mathbb{Z}}}d_a^2 d_b =\DD S_{00}d_b,
\end{split}
\end{align}
where in going from the first line to the second line, we have used that $|S_{ab}|\leq S_{00}d_ad_b$ (see Equation (4.2a) of \cite{gannon2005modular}).

Now, if we take the real part of Equation \eqref{eqn:unpacked} and evaluate at $\tau=\ri t$, then we find that 
\begin{align}
    \sum_{b\in\Rep(V)_{h\in\mathbb{Z}}}\left( \DD S_{00} d_b -\mathrm{Re}\sum_{a\in\Rep(V)_{h\in\mathbb{Z}}}d_a S_{ab}\right)\mathrm{ch}_{V_b}(\ri t) = 0.
\end{align}
Note that $\mathrm{ch}_{V_b}(it)$ is a strictly positive real number, since the coefficients in its $q$-expansion are dimensions of vector spaces. On the other hand, each term in the parentheses is non-negative by \eqref{eqn:qdiminequality}. The only way this equation is consistent is if the non-negative numbers are in fact identically zero, 
\begin{align}\label{eqn:nexteqn}
    \sum_{a\in\Rep(V)_{h\in\mathbb{Z}}}d_a\left( S_{00}d_ad_b-\mathrm{Re}S_{ab}\right)=0, \ \ \ \text{ for all }V_b\in \Rep(V)_{h\in\mathbb{Z}}.
\end{align}
Again, since $d_a>0$ and since $\mathrm{Re}S_{ab}\leq |S_{ab}|\leq S_{00}d_ad_b$, the only way \eqref{eqn:nexteqn} is consistent is if $\mathrm{Re} S_{ab}=S_{00}d_ad_b$. By squeezing, $\mathrm{Re}S_{ab}=|S_{ab}|$ (so that $S_{ab}$ is real and positive) and hence $S_{ab}=S_{00}d_ad_b$. That is, $V$ is Tannakian.

In the reverse direction, if $V$ is Tannakian, then $S_{ab}=S_{00}d_ad_b$ by Proposition \ref{prop:TannakianS} and hence
\begin{align}
\begin{split}
    \Pr\left[ P_V(-1/\tau)\right] &= \sum_{a\in\Rep(V)_{h\in\mathbb{Z}}}d_a \sum_{b\in \Rep(V)_{h\in\mathbb{Z}}} S_{ab}\mathrm{ch}_{V_b}(\tau)\\
    &=S_{00}\sum_{a\in\Rep(V)_{h\in\mathbb{Z}}}d_a^2\sum_{b\in\Rep(V)_{h\in\mathbb{Z}}}d_b\mathrm{ch}_{V_b}(\tau)= S_{00}\DD P_V(\tau).
\end{split}
\end{align}
This concludes the proof. 
\end{proof}

We remark that, by Proposition \ref{prop:Tannakianext}, when $V$ is Tannakian, $P_V(\tau)$ is the vacuum character of the unique maximal conformal extension $V_{\Rep(G)}$ of $V$. Theorem \ref{thm:TannakianVacChar} will feature prominently in our proof of the classification of nice $c=1$ VOAs in the next section.

\section{Classification of nice \texorpdfstring{$c=1$}{c=1} vertex operator algebras}\label{sec:classification}

With the generalities of the previous section out of the way, we move on to our main result. In Section \ref{subsec:bosonicassump}, we carefully state our main classification theorem and make some related observations. Section \ref{subsec:mainproof} carries the proof.

\subsection{Assumptions and statement}\label{subsec:bosonicassump}

\begin{theorem}[Classification of nice $c=1$ vertex operator algebras]\label{thm:c=1classification}
    If $V$ is a nice vertex operator algebra with central charge $c=1$, then $V$ is isomorphic to one of the following, 
    \begin{align}\label{eqn:c=1VOAlist}
        V_{2m} \ (m\geq 1), \ \ V_{2m}^+ \ (m\geq 2), \ \ V_T, \ \ V_O, \ \ V_I.
    \end{align}
    Moreover, the VOAs in this list are pairwise non-isomorphic. 
\end{theorem}

The reason we do not include $V_2^+$ in the list is because it is isomorphic to $V_8$. (For physicists, we note that this is a manifestation of the fact that the circle and orbifold branches of the $c=1$ conformal moduli space meet at the Kosterlitz-Thouless point.) The fact that the VOAs appearing in \eqref{eqn:c=1VOAlist} are pairwise distinct follows by inspecting their vacuum characters, recorded in Appendix \ref{app:bestiary}. Before we turn to the rest of the proof, a few remarks are in order. 

First, the converse to Theorem \ref{thm:c=1classification} is now settled using the recent preprint \cite{xu2026c2cofiniteness}, which claims to establish the strong rationality of $V_I$.

\begin{proposition}\label{prop:conversec=1classification}
     If $V$ is isomorphic to one of the VOAs in \eqref{eqn:c=1VOAlist},  then $V$ is nice.
\end{proposition}

\begin{proof}
    The strong rationality of $V_{2m}$ is a special case of the strong rationality of vertex operator algebras associated to even, positive-definite lattices \cite{dong1993vertex,Dong:1997ea}. Pseudo-unitarity follows by inspecting the simple modules of $V_{2m}$, reviewed in Appendix \ref{app:bestiary}. 
    
    The strong rationality of $V_{2m}^+:= V_{2m}^{\mathbb{Z}_2^{\mathrm{C}}}$, $V_T:=V_2^{A_4}$, and $V_O:=V_2^{S_4}$ follows from the fact that the $G$-fixed point subalgebra of a strongly rational VOA is again strongly rational if $G$ is a solvable finite group \cite{Carnahan:2016guf}. The strong rationality of $V_I$ was argued for recently in \cite{xu2026c2cofiniteness} by more elaborate means. The pseudo-unitarity of these VOAs can again be gleaned by inspecting their simple modules, reviewed in Appendix \ref{app:bestiary}. 
\end{proof}

Second, we note that every VOA appearing in Theorem \ref{thm:c=1classification} is in fact unitary.

\begin{proposition}\label{prop:unitary}
    Every nice vertex operator algebra with central charge $c=1$ is unitary. 
\end{proposition}

\begin{proof}
VOAs associated to positive definite even lattices, such as $V_{2m}$, are unitary \cite{dong2014unitary}. Moreover, if $G$ is a finite group of unitary automorphisms of a unitary VOA $V$, then the $G$-fixed point subalgebra $V^G$ is also unitary. This recovers the unitarity of the remaining VOAs in \eqref{eqn:c=1VOAlist}, since $V_{2m}^+:=V_{2m}^{\mathbb{Z}_2^{\mathrm{C}}}$, $V_T:=V_{2}^{A_4}$, $V_O:=V_2^{S_4}$, and $V_I:=V_2^{A_5}$. 
\end{proof}

Again, the authors are not aware of any pseudo-unitary strongly rational VOA which is not also unitary. Proposition \ref{prop:unitary} might be regarded as circumstantial evidence in favor of the speculation that no such VOAs exist. The VOAs appearing in Theorem \ref{thm:c=1classification} are presumably also completely unitary in the sense of \cite{gui2022q}, but we leave this interesting question to the future.

Third, we note that if pseudounitarity is dropped, then Theorem \ref{thm:c=1classification} is certainly false (cf.\ \cite[Example 3.26]{Moller:2024xtt}). For example, there are infinitely many strongly rational $c=1$ vertex operator algebras of the form 
\begin{align}
    \mathrm{LY}^{5n}\otimes V_L,
\end{align} 
where $\mathrm{LY}$ is the Lee-Yang VOA (i.e.\ the simple quotient of the Virasoro VOA at $c=-\sfrac{22}5$), and $V_L$ is a lattice VOA associated to an even, positive definite lattice of rank $22n+1$.

The theorem is also incorrect if we do not assume strong rationality. However, while the classification of strongly rational VOAs which are not pseudo-unitary is hopeless without further adjectives, we have the following conjectural generalization of Theorem \ref{thm:c=1classification} to irrational VOAs.

\begin{conjecture}\label{conj:c=1irrational}
    Let $V$ be a simple, CFT-type, self-contragredient vertex operator algebra with central charge $c=1$, and assume that $V$ decomposes into a direct sum of unitary simple $\mathrm{Vir}_1$-modules. Then $V$ is isomorphic to one of the VOAs appearing in \eqref{eqn:c=1VOAlist}, or it is one of 
    \begin{align}\label{eqn:irrationalunitaryVOAs}
        \mathrm{Vir}_1, \ \ \widehat{\mathfrak{u}}(1), \ \ \widehat{\mathfrak{u}}(1)^+,
    \end{align}
where $\mathrm{Vir}_1$ is the $c=1$ Virasoro VOA, $\widehat{\mathfrak{u}}(1)$ is the Heisenberg VOA (a.k.a.\ $\mathfrak{u}(1)$ Kac-Moody algebra), and $\widehat{\mathfrak{u}}(1)^+$ is the charge conjugation fixed-point subVOA of $\widehat{\mathfrak{u}}(1)$.
\end{conjecture}

\noindent The reason we have traded out pseudo-unitarity for the technical assumption that $V$ decomposes into unitary modules of $\mathrm{Vir}_1$ is because $\mathrm{Vir}_1$ itself is not pseudo-unitary. Indeed, it has a simple module for any $h\in\mathbb{C}$.

We note two existing results in the literature that are related to this conjecture. First, in Theorem 4.6 of \cite{xu2005strong}, Xu proved that any $c=1$ conformal net which satisfies what he called a ``spectrum condition'' is isomorphic to a conformal net obtained from a VOA appearing in Conjecture \ref{conj:c=1irrational}. (A $c=1$ conformal net $\mathcal{A}$ is said to satisfy the spectrum condition if a degenerate module of the Virasoro net, other than the vacuum, appears in $\mathcal{A}$ if $\mathcal{A}$ is not itself the Virasoro net.) 

Second, there is Theorem 3.15 of \cite{carpi2019classification}, which establishes that the VOAs appearing in \eqref{eqn:irrationalunitaryVOAs} are the only irrational unitary conformal subVOAs of any rank-1 lattice VOA. In particular, by our Theorem \ref{thm:c=1classification}, they are the only irrational unitary conformal subVOAs of \emph{any} strongly rational $c=1$ VOA, lattice or not. Said another way, any $c=1$ unitary VOA which does not appear in Conjecture \ref{conj:c=1irrational} evidently does not admit a conformal extension into a strongly rational VOA, and in particular cannot appear as a conformal subalgebra of the chiral algebra of a rational CFT.\medskip 

Finally, we conclude this subsection by noting that the hypotheses of Theorem \ref{thm:c=1classification} can be relaxed, roughly by trading rationality for rigidity. More precisely, we have the following corollary of Theorem \ref{thm:c=1classification}.

\begin{corollary}\label{cor:strengthened}
    Let $V$ be a strongly finite $c=1$ vertex operator algebra with a rigid and pseudo-unitary representation category. Suppose further that there is a unique simple module $M$ with $h_M=h_{\mathrm{min}}$, where $h_{\mathrm{min}}$ is the minimum taken over the conformal dimensions of the simple modules of $V$. Then $V$ is isomorphic to one of the VOAs in \eqref{eqn:c=1VOAlist}.
\end{corollary}

\begin{proof}
    By Lemma \ref{lem:C2+pseudounitary}, $V$ must be strongly rational, and by Proposition \ref{prop:pseudounitarypartialconverse}, $V$ must be pseudo-unitary. Hence, the hypotheses of Theorem \ref{thm:c=1classification} hold. 
\end{proof}

\subsection{Proof}\label{subsec:mainproof}

The remainder of this section will be dedicated to proving the following.

\begin{proposition}\label{prop:c=1Tannakian}
    Every nice $c=1$ VOA $V$ is Tannakian. Moreover, the maximal conformal extension $V_{\Rep(G)}\supset V$ guaranteed by Proposition \ref{prop:Tannakianext} is isomorphic to a lattice VOA, $V_{\Rep(G)}\cong V_{2t_\star}$, for some square-free positive integer $t_\star$ which depends on $V$.
\end{proposition}

Theorem \ref{thm:c=1classification} can be recovered as a corollary of Proposition \ref{prop:c=1Tannakian}. Indeed, assuming Proposition \ref{prop:c=1Tannakian}, we learn by applying Proposition \ref{prop:Tannakianext} that any nice $c=1$ VOA can be obtained as $V_{2t_\star}^G$ for some square-free $t_\star$ and some finite subgroup $G\subset \mathrm{Aut}(V_{2t_\star})$. The VOAs of this form are enumerated in Proposition \ref{prop:Gfixedpts} (see also Example 3.26 of \cite{Moller:2024xtt}) and recover precisely the list appearing in Theorem \ref{thm:c=1classification}.

Our approach is to use modular covariance and the Galois symmetry of modular fusion categories to control the characters of any nice $c=1$ VOA enough so that we can invoke Theorem \ref{thm:TannakianVacChar} and conclude that $V$ is Tannakian. 

 Recall that the $h\geq 0$ simple highest weight modules of the $c=1$ Virasoro VOA have characters given by 
\begin{align}
    \chi^{\mathrm{Vir}}_h(\tau) = \begin{cases}
       \frac{1}{\eta(\tau)}(q^{\mu^2}-q^{(\mu+1)^2}), & h=\mu^2  \text{ with }\mu \in \frac12 \mathbb{Z}_{\geq 0}, \\
       \frac{1}{\eta(\tau)}q^h, & \text{otherwise,}
    \end{cases}
\end{align}
where $\eta(\tau)=q^{1/24}\prod_{n=1}^\infty (1-q^n)$ is the Dedekind eta function. Even though $V$ may not be completely reducible as a Virasoro module, its vacuum character $\mathrm{ch}_V(\tau)$ decomposes into a sum of its simple composition factors,
\begin{align}\label{eqn:chVdecomp}
    \mathrm{ch}_V(\tau) = \chi_0^{\mathrm{Vir}}(\tau)+\sum_{k>0}A_k \chi_k^{\mathrm{Vir}}(\tau), \ \ \ \ A_{k}\in\mathbb{Z}_{\geq 0}.
\end{align}
It turns out it will be easier to impose the constraints of modularity on the auxiliary function 
\begin{align}\label{eqn:FV}
    F_V(\tau)=\eta(\tau)\mathrm{ch}_V(\tau).
\end{align}
Indeed, $F_V(\tau)$ serves almost as a generating function for the Jordan-Hölder multiplicities $A_k$ in the sense that
\begin{align}\label{eqn:virgenfunc}
    F_V(\tau) = 1+(A_1-1)q+\sum_{n> 1}(A_{n^2}-A_{(n-1)^2})q^{n^2}+ \sum_{\substack{\ell>0 \\ \ell\neq n^2}}A_\ell q^\ell.
\end{align}
The word ``almost'' here refers to the fact that the coefficient of $q^{n^2}$ is $A_{n^2}-A_{(n-1)^2}$ rather than just $A_{n^2}$.

Another reason it is useful to work with $F_V(\tau)$ is that, for a nice $c=1$ VOA, $\eta(\tau)$ clears the poles of $\mathrm{ch}_V(\tau)$ at every cusp and hence makes $F_V(\tau)$ a \emph{holomorphic} (rather than weakly-holomorphic) modular form. Holomorphic modular forms belong to finite-dimensional vector spaces which are often computable. 

To make this more precise, recall the definition of the congruence subgroups,
\begin{align}
\begin{split}
    \Gamma_1(N)&=\{\left(\begin{smallmatrix}\alpha & \beta \\ \gamma & \delta \end{smallmatrix}\right)\in \textsl{SL}_2(\mathbb{Z}) \mid \alpha,\delta\equiv 1~\mathrm{mod}~N, \ \gamma\equiv 0~\mathrm{mod}~N\} \\
    \Gamma_0(N) &=\{\left(\begin{smallmatrix}\alpha & \beta \\ \gamma & \delta \end{smallmatrix}\right) \in \textsl{SL}_2(\mathbb{Z}) \mid   \gamma\equiv 0~\mathrm{mod}~N\}  ,
\end{split}
\end{align}
and that $\textsl{SL}_2(\mathbb{Z})$ acts on the upper half-plane by fractional linear transformations,
\begin{align}
    \left(\begin{smallmatrix}\alpha & \beta \\ \gamma & \delta \end{smallmatrix}\right) \cdot  \tau = \frac{\alpha\tau +\beta}{\gamma\tau+\delta}.
\end{align}
A function $f(\tau)$ from the upper half-plane to the complex numbers is said to be a $\theta$-type holomorphic modular form of weight $\sfrac12$ for a congruence subgroup $\Gamma\subset \Gamma_0(4)$ if it is holomorphic in the interior of the upper half-plane and at the cusps, and if the multiplier system of $f(\tau)$ agrees with that of $\theta(\tau)$ on $\Gamma$, i.e.\ 
\begin{align}
    f(\gamma\tau)/f(\tau) = \theta(\gamma\tau)/\theta(\tau), \ \ \ \ \text{ for all }\gamma \in \Gamma,
\end{align}
where $\theta(\tau)=\sum_{n\in\mathbb{Z}} q^{n^2}$.

\begin{lemma}\label{lem:FVmodularity}
    If $V$ is a nice $c=1$ VOA, then the function $F_V(\tau)$ defined in \eqref{eqn:FV} is a holomorphic $\theta$-type modular form of weight $\sfrac12$ for $\Gamma_1(C_V)$, where $C_V$ is the conductor of $V$.
\end{lemma}

\begin{proof}
    First, we note that  
    \begin{align}\label{eqn:etamultiplier}
        \eta(\gamma\tau)/\eta(\tau) = \theta(\gamma\tau)/\theta(\tau), \ \ \ \ \text{for all } \gamma \in\Gamma(24).
    \end{align}
    A quick way to see this is to observe that $\theta(\tau)/\eta(\tau)$ is the vacuum character of $V_2\cong \widehat{\mathfrak{su}}(2)_1$ which has conductor $C_{V_2}=24$, so by Theorem \ref{thm:congruence}, $\theta(\tau)/\eta(\tau)$ is invariant under the action of $\Gamma(24)$.

    Combining \eqref{eqn:etamultiplier} with the fact that $\mathrm{ch}_V(\tau)$ is invariant under $\Gamma(C_V)$ by Theorem \ref{thm:congruence}, we learn that $F_V(\tau)$ transforms like a weight $\sfrac12$ modular form for $\Gamma(C_V) \cap \Gamma(24)$. But since the denominator of $-c/24=-1/24$ divides the conductor, we have that $\Gamma(C_V)\cap \Gamma(24)=\Gamma(C_V)$. Furthermore, by inspecting the $q$-expansion of $F_V(\tau)$ in \eqref{eqn:virgenfunc}, we see that it obeys $F_V(\tau+1)=F_V(\tau)$. Thus, in fact, $F_V(\tau)$ is a weight $\sfrac12$ modular form for $\langle \Gamma (C_V),\left(\begin{smallmatrix} 1 & 1 \\ 0 & 1 \end{smallmatrix}\right)\rangle =\Gamma_1(C_V)$.

    Finally, since $\textsl{SL}_2(\mathbb{Z})$ acts transitively on the cusps, $F_V(\tau)$ is holomorphic at all cusps if and only if $F_V(\gamma\tau)$ is holomorphic at $q=0$ for all $\gamma\in \textsl{SL}_2(\mathbb{Z})$. Theorem \ref{thm:modularity} implies that the $q$-expansion of $\mathrm{ch}_V(\gamma\tau)$  has singular terms which are no worse than $q^\nu$ for $\nu=\mathrm{min}_a (h_a-c/24)$. For a pseudo-unitary $c=1$ VOA, $\nu=-\sfrac{1}{24}$, and since $\eta(\gamma\tau)(c\tau+d)^{-\sfrac12}$ has a $q$-expansion which starts as $\propto q^{+\sfrac{1}{24}}+\cdots $, it follows that $F_V(\gamma\tau) = \eta(\gamma\tau) \mathrm{ch}_V(\gamma\tau)$ is holomorphic. 
\end{proof}

Kiritsis \cite{Kiritsis:1988et} recognized that one could apply the Serre-Stark classification of weight $\sfrac12$ holomorphic modular forms of $\theta$-type \cite{serre2006modular} to constrain the partition functions of $c=1$ rational conformal field theories. We will follow a similar path, though will depart from \cite{Kiritsis:1988et} in the details. First, we will only work with characters of $V$ (rather than partition functions of full CFTs). Second, we will not attempt to explicitly enumerate the possible $c=1$ character vectors directly, but rather just constrain their form enough so that we can conclude that $V$ is Tannakian.

We will work with theta functions of the shape
\begin{align}
    \theta_{n_0,r,t}(\tau) := \sum_{\substack{n\equiv n_0 ~\mathrm{mod}~r\\n\in\mathbb{Z}}}q^{tn^2},
\end{align}
with $t,r\in\mathbb{Z}_{> 0}$.
The function $\theta_{n_0,r,t}(\tau)$ is a weight $\sfrac12$ holomorphic modular form of $\theta$-type for $\Gamma_1(4tr^2)$. It is elementary to see that the $\theta_{n_0,r,t}$ satisfy the relations 
\begin{align}\label{eqn:relations}
    \theta_{n_0,r,t}=\theta_{-n_0,r,t}, \ \ \ \ \ \theta_{n_0,r,t\ell^2}=\theta_{n_0\ell,r\ell,t}, \ \ \ \ \ \theta_{n_0,r,t}=\sum_{\substack{n_0'\in\mathbb{Z}_{r\ell} \\ n_0'\equiv n_0~\mathrm{mod}~r}}\theta_{n_0',r\ell,t}.
\end{align}
They are related to the rank-1 lattice theta functions defined in Equation \eqref{eqn:freebosoncharacters} as 
\begin{align}\label{eqn:thTh}
    \theta_{n_0,r,t}(\tau)=\Theta_{tr^2,2trn_0}(\tau).
\end{align}
Thus, one can use the well-known modular S-matrix of the lattice VOA $V_{2tr^2}$, Equation \eqref{eqn:latticeSmatrix}, to conclude that 
\begin{align}\label{eqn:thetaS}
    \theta_{n_0,r,t}(-1/\tau) = \frac{\sqrt{-\ri\tau}}{r\sqrt{2t}}\sum_{\nu=0}^{2tr^2-1}e^{-2\pi \ri n_0 \nu/r}\Theta_{tr^2,\nu}(\tau),
\end{align}
where we have used the standard fact that $\eta(-1/\tau)=\sqrt{-\ri\tau}\eta(\tau)$.

It will also be useful to record how these theta functions transform under Galois symmetries. We first record the following standard number theory result. 
Let $D\neq 0$, and write $D=m^2d$ with $d$ square-free. Define
\begin{align}
    \Delta = \begin{cases}
        d, & d\equiv 1~\mathrm{mod}~4 \\
        4d, & \text{otherwise,}
    \end{cases}
\end{align}
so that $\Delta$ is the fundamental discriminant of $\mathbb{Q}(\sqrt{D})$. 
Suppose that $N$ is a multiple of $|\Delta|$ and $\ell$ an integer with $\gcd(\ell,N)=\gcd(\ell,D)=1$. Then $\sqrt{D}\in \mathbb{Q}(\zeta_N)$ and 
    \begin{align}\label{eqn:galoisaut}
        \sigma_\ell(\sqrt{D})= \left(\frac{D}{\ell}\right) \sqrt{D},
    \end{align}
    where $\left(\frac\cdot\cdot\right)$ is the Kronecker symbol and $\sigma_\ell$ is the Galois automorphism of $\mathbb{Q}(\zeta_N)$ sending $\zeta_N\mapsto \zeta_N^\ell$.

\begin{proposition}
Let $M=\mathrm{lcm}(4tr^2,24)$ and $\ell$ an integer coprime to $M$.  For any $\gamma_\ell\in \textsl{SL}_2(\mathbb{Z})$ which is congruent to $\mathrm{diag}(\ell,\ell^{-1})~\mathrm{mod}~M$,
\begin{align}\label{eqn:Galoistheta}
   \frac{ \theta_{n_0,r,t}(\gamma_\ell \tau)}{\eta(\gamma_\ell\tau)}=\left(\frac{2t}{\ell}\right)\frac{\theta_{\ell n_0,r,t}(\tau)}{\eta(\tau)},
\end{align}
where $\left(\frac ab\right)$ is the Kronecker symbol.
    
\end{proposition} 

\begin{proof} A slick way to prove this is to use the fact that $\theta_{n_0,r,t}/\eta$ is a character of the lattice VOA $V_{2tr^2}$ à la \eqref{eqn:thTh}. Via the Galois symmetry of modular fusion categories \cite{Coste:1999yc,Bantay:2001ni,dong2015congruence}, we know that, for a strongly rational VOA $V$ with modular representation $\rho:\textsl{SL}_2(\mathbb{Z})\to \textsl{GL}_N(\mathbb{C})$, one has that
\begin{align}
    \rho(\gamma_\ell)=\sigma_\ell(S)S^{-1}, \ \ \ \ \ \ \ \ S=\rho\left(\begin{smallmatrix}0 & -1 \\ 1 & 0\end{smallmatrix}\right),
\end{align}
where $\sigma_\ell$ is the Galois automorphism of $\mathbb{Q}(e^{2\pi \ri/C_V})$ given by $e^{2\pi \ri/C_V}\mapsto e^{2\pi \ri \ell/C_V}$, with $C_V$ the conductor of $V$.

For the lattice VOA $V=V_{2tr^2}$, the conductor is precisely $C_V=\mathrm{lcm}(4tr^2,24)$. Furthermore, 
\begin{align}
    \sigma_\ell(S_{\mu\nu})=\sigma_\ell \left(\frac{1}{\sqrt{2t}r}e^{-\pi \ri \mu\nu/tr^2}       \right)=\left(\frac{2t}{\ell}\right)\frac{1}{\sqrt{2t}r}e^{-\pi \ri \ell \mu \nu/tr^2}=\left(\frac{2t}{\ell}\right)S_{\ell\mu,\nu},
\end{align}
where we have used \eqref{eqn:galoisaut} to obtain $\sigma_\ell(\frac{1}{\sqrt{2t}})=\frac{1}{2t}\sigma_\ell(\sqrt{2t})=\frac{1}{\sqrt{2t}}\left(\frac{2t}{\ell}\right)$.

Thus, $\rho(\gamma_\ell)$ evaluates to 
\begin{align}
    \rho(\gamma_\ell)_{\mu\nu}=\left(\frac{2t}{\ell}\right)\sum_\alpha S_{\ell\mu,\alpha}(S^{-1})_{\alpha,\nu} = \left(\frac{2t}{\ell}\right)\delta_{\ell\mu,\nu}.
\end{align}
Specializing to $\mu=2trn_0$, the statement of the proposition follows.
\end{proof}

Now, we apply the Serre-Stark theorem \cite{serre2006modular} to constrain the form of the vacuum character of any nice $c=1$ VOA.

\begin{lemma}\label{lem:c=1VacChar}
    Suppose $V$ is a nice $c=1$ VOA. Given an integer $t$ such that $4t$ divides $C_V$, let $r_t$ be the largest integer such that $4t r_t^2$ divides $C_V$. Then there is a unique decomposition of the form 
    \begin{align}
        \mathrm{ch}_V(\tau) = \frac{1}{\eta(\tau)}\sum_{t}^\star\sum_{n_0~\mathrm{mod}~r_t}a_t(n_0) \theta_{n_0,r_t,t}(\tau), \ \ \ \ \ a_t(n_0)=a_t(-n_0),
    \end{align}
    where the sum $\sum_t^\star$ is over square-free integers $t$ dividing $C_V/4$.
\end{lemma}

\begin{proof}
    Corollary 1 and Corollary 3 of \cite{serre2006modular} assert that any function which is a $\theta$-type weight $\sfrac12$ holomorphic modular form for $\Gamma_1(N)$ can be expanded as a finite linear combination of theta series $\theta_{n_0,r,t}$ with $4tr^2\vert N$. By Lemma \ref{lem:FVmodularity}, it follows that $F_V(\tau)=\eta(\tau)\mathrm{ch}_V(\tau)$ can be expanded in these theta series as well. 

    Because $F_V(\tau)$ is modular with respect to $\Gamma_1(C_V)$, we can restrict the sum over theta functions to those for which $4tr^2\vert C_V$. The first identity in \eqref{eqn:relations} tells us that we can impose the symmetry condition $a_t(n_0)=a_t(-n_0)$ on the coefficients. The second identity tells us that we can restrict to square-free $t$. The third identity tells us that we can restrict to the functions $\theta_{n_0,r_t,t}$, as any other choice of $\theta_{n_0,r,t}$ which is compatible with the level being $C_V$ can be expanded in the functions $\theta_{n_0,r_t,t}$. 

    To show that the decomposition is unique, we argue that the $a_t(n_0)$ are completely fixed by the $q$-expansion of $F_V(\tau)$. Writing 
    \begin{align}
        F_V(\tau) = \sum_{\ell \in\mathbb{Z}_{\geq 0}}c_\ell q^\ell,
    \end{align}
    we observe that, for $t$ a square-free integer dividing $C_V/4$, we have $c_{tn^2}=2a_t(\bar n)=2a_t(-\bar n)$ for any $n>0$, where $\bar n$ is the reduction of $n$ modulo $r_t$. Indeed, among the theta functions appearing in the decomposition, only $\theta_{\bar n,r_t,t}$ and $\theta_{-\bar n,r_t,t}$ possess a $q^{tn^2}$ term due to the fact that $tn^2=t'm^2$ for $t,t'>0$ square-free and $n,m$ integers implies that $t=t'$ and $|n|=|m|$.
\end{proof}

Lemma \ref{lem:c=1VacChar} strongly constrains the shape of the vacuum character of any nice $c=1$ VOA $V$. We also get precious information about the $V$-modules $V_a$ with integer conformal dimension through the function $P_V(\tau)$ defined in \eqref{eqn:PVdef}.

\begin{lemma}
    Suppose $V$ is a nice $c=1$ VOA. Then  
    \begin{align}\label{eqn:PVform}
        P_V(\tau) =  \sum_{a\in \Rep(V)_{h\in\mathbb{Z}}}d_a\mathrm{ch}_{V_a}(\tau)= \frac{1}{S_{00}\eta(\tau)}\sum_{t}^\star w_t\theta_{0,1,t}(\tau),
    \end{align}
    where the sum is over square-free $t$ dividing $C_V/4$, and the coefficients $w_t$ are given by
    \begin{align}
        w_t:=\frac{1}{r_t\sqrt{2t}}\sum_{n_0~\mathrm{mod}~r_t}a_t(n_0), \ \ \ \ \sum_t^\star w_t=S_{00}.
    \end{align}
\end{lemma}

\begin{proof}
As described in the discussion around Equation \eqref{eqn:PVdef}, we can extract $P_V(\tau)$ from the vacuum character by taking the modular $S$ transformation and projecting onto the contributions of states with integer conformal dimension. Using the $\tau\mapsto -1/\tau$ transformation equation of the theta functions given in \eqref{eqn:thetaS}, we see that
\begin{align}
    P_V(\tau)&=\frac{1}{S_{00}}\Pr\left[\frac{1}{\eta(\tau)}\sum_{t}^\star\frac{1}{r_t\sqrt{2t}}\sum_{n_0~\mathrm{mod}~r_t}a_{t}(n_0)\sum_{\nu \in \mathbb{Z}_{2tr_t^2}}e^{-2\pi \ri n_0\nu/r_t}\Theta_{tr_t^2,\nu}(\tau)\right],
\end{align}
where the projector $\Pr$ was defined in \eqref{eqn:PVdef}.  
Now, note that $\Theta_{tr_t^2,\nu}$ features only terms proportional to $q^{\nu^2/4tr_t^2+n}$ in its $q$-expansion, with $n$ an integer. These exponents are integers if and only if $\nu=2tr_t\ell$ so, using \eqref{eqn:thTh} to revert back to an expression in terms of the $\theta_{n_0,r,t}$, we find 
\begin{align}
\begin{split}
    P_V(\tau) 
    &=\frac{1}{S_{00}\eta(\tau)}\sum_t^\star \frac{1}{r_t\sqrt{2t}}\sum_{n_0=0}^{r_t-1}a_t(n_0)\sum_{\ell=0}^{r_t-1}\theta_{\ell,r_t,t}(\tau)\\
    &=\frac{1}{S_{00}\eta(\tau)}\sum_t^\star \frac{1}{r_t\sqrt{2t}}\sum_{n_0=0}^{r_t-1}a_t(n_0)\theta_{0,1,t}(\tau),
\end{split}
\end{align}
which we recognize as the right-hand side of \eqref{eqn:PVform}. We can conclude that $S_{00}=\sum_t^\star w_t$ by comparing the leading terms in the $q$-expansion of $P_V(\tau)$ using both expressions appearing in Equation \eqref{eqn:PVform}.
\end{proof}

Ultimately, we would like to show that only a single square-free $t$ contributes to $P_V(\tau)$ in \eqref{eqn:PVform}, which will put us in position to successfully invoke Theorem \ref{thm:TannakianVacChar}. A useful intermediate step is to show that the $w_t$ are all non-negative. 

\begin{lemma}
    If $V$ is a nice $c=1$ VOA, then the coefficients $w_t$ appearing in \eqref{eqn:PVform} are all non-negative.
\end{lemma}

\begin{proof} It suffices to show that the quantities $s_t:=\sum_{n_0~\mathrm{mod}~r_t}a_t(n_0)$ are all non-negative. Let $A_n\in\mathbb{Z}_{\geq0}$ be the Jordan-Hölder multiplicity inside $V$ of the simple highest-weight module of the $c=1$ Virasoro VOA with conformal dimension $h=n$. The $A_n$ appear in the $q$-expansion of $F_V(\tau)$ as in \eqref{eqn:virgenfunc}.  If $t>1$, then 
\begin{align}
    a_t(\bar n)=\frac12A_{tn^2}\geq 0, \ \ \ \ \text{ for all } n>0,
\end{align} 
where $\bar n$ is the reduction of $n$ modulo $r_t$. This gives the non-negativity of $s_t$ for all $t>1$. For $t=1$, we have that e.g.\ 
\begin{align}
    A_{(kr_1)^2}=1+2ks_1
\end{align}
for all $k\geq 0$, so if $s_1$ was negative it would eventually make $A_{(kr_1)^2}$ negative for some $k>0$, a contradiction.
\end{proof}

\begin{proposition}\label{prop:PVfinal}
If $V$ is a nice $c=1$ VOA, then there is a single square-free $t_\star$ such that $w_{t_\star}\neq 0$. In particular, it follows that 
\begin{align}
    P_V(\tau) = \sum_{a\in \Rep(V)_{h\in\mathbb{Z}}}d_a\mathrm{ch}_{V_a}(\tau) = \frac{\theta_{0,1,t_\star}(\tau)}{\eta(\tau)}.
\end{align}
\end{proposition}

\begin{proof} 

Let $C_V$ be the conductor of $V$, and let $\rho:\textsl{SL}_2(\mathbb{Z})\to \textsl{GL}_N(\mathbb{C})$ be the modular representation associated to $V$ which in particular dictates the transformation properties of its characters, as in Theorem \ref{thm:modularity}. Here, $N=\mathrm{rk}(\Rep(V))$ is the number of simple modules of $V$.

For any $\ell$ coprime to $C_V$, pick an $\textsl{SL}_2(\mathbb{Z})$ element $\gamma_\ell = \mathrm{diag}(\ell,\ell^{-1})~\mathrm{mod}~C_V$, where $\ell^{-1}$ denotes the multiplicative inverse modulo $C_V$. A standard result in the Galois theory of modular fusion categories \cite{Coste:1999yc,Bantay:2001ni,dong2015congruence} says that $\rho(\gamma_\ell)$ is a signed permutation matrix. In particular, this implies that 
\begin{align}
    \mathrm{ch}_V(\gamma_\ell\tau) = \delta(\ell)\mathrm{ch}_{V^{(\ell)}}(\tau),
\end{align}
where $\delta(\ell)\in \{\pm 1\}$ is an $\ell$-dependent sign, and $V^{(\ell)}$ is some simple $V$-module (the Galois image with respect to $\ell$ of the vacuum). 

On the other hand, given the constrained form of the vacuum character of a nice $c=1$ VOA coming from Lemma \ref{lem:c=1VacChar}, we can explicitly compute that 
\begin{align}
    \mathrm{ch}_V(\gamma_\ell\tau)=\sum_t^\star\left(\frac{2t}\ell\right) \sum_{n_0\in \mathbb{Z}_{r_t}} a_t(n_0)\frac{\theta_{\ell n_0,r_t,t}(\tau)}{\eta(\tau)},
\end{align}
where we have made use of Equation \eqref{eqn:Galoistheta}.

Let us take the modular $S$-transform of both of these expressions for $\mathrm{ch}_V(\gamma_\ell\tau)$ and compare the leading term in the $q$-expansion (i.e.\ the coefficient of $q^{-\sfrac1{24}}$). On one side, we have that
\begin{align}
    \delta(\ell)\sum_b S_{V^{(\ell)},b}\mathrm{ch}_{V_b}(\tau)\Bigg\vert_{q^{-\sfrac1{24}}} = \delta(\ell)S_{V^{(\ell)},V},
\end{align}
where $S_{V^{(\ell)},V}$ is the S-matrix element between the vacuum and its $\ell$th Galois image. Making use of \eqref{eqn:thetaS}, the other expression for $\mathrm{ch}_{V}(\gamma_\ell\tau)$ gives us
\begin{align}
    &\sum_t^\star \left(\frac{2t}{\ell}\right)\sum_{n_0\in\mathbb{Z}_{r_t}}a_t(n_0)\frac{1}{r_t\sqrt{2t}}\sum_{\nu\in\mathbb{Z}_{2tr_t^2}}e^{-2\pi i n_0\nu/r_t}\frac{\Theta_{tr_t^2,\nu}(\tau)}{\eta(\tau)}\Bigg\vert_{q^{-1/24}} = \sum_t^\star \left(\frac{2t}{\ell}\right)w_t.
\end{align}
Equating these two and taking absolute values, we deduce the inequality
\begin{align}\label{eqn:S00lowerbound}
    \left|\sum_t^\star \left(\frac{2t}\ell\right) w_t\right| = S_{V^{(\ell)},V}=d_{V^{(\ell)}}S_{00}\geq S_{00},
\end{align}
where we use that $d_{V^{(\ell)}}\geq 1$ for a nice VOA since the representation category is pseudo-unitary.

Now, suppose that $t_1$ and $t_2$ are distinct square-free integers dividing $C_V/4$ for which $w_{t_1}\neq 0$ and $w_{t_2}\neq 0$. Define $g=\gcd(t_1,t_2)$. Then $t_1t_2=g^2s$ where $s=(t_1/g)(t_2/g)$. It is clear that $s>1$, otherwise $t_1$ would equal $t_2$. Also, $s$ is square-free because $t_1/g$ and $t_2/g$ are both coprime and square-free. Moreover, every prime of $s$ clearly divides $C_V$ because $t_1$ and $t_2$ both divide $C_V$, so $s$ divides $C_V$. In fact, $4s$ divides $C_V$ because the conductor of a nice $c=1$ VOA is always divisible by $24$.

Now, let $\Delta$ be the fundamental discriminant of $\mathbb{Q}(\sqrt{s})$, i.e.\ $\Delta=s$ if $s\equiv 1~\mathrm{mod}~4$ and $\Delta=4s$ otherwise. Then,
\begin{align}
    \ell\mapsto \left(\frac{2t_1}{\ell}\right)\left(\frac{2t_2}{\ell}\right)=\left(\frac{t_1t_2}{\ell}\right)=\left(\frac{s}{\ell}\right)=\left(\frac{\Delta}{\ell}\right)=:\chi_\Delta(\ell),
\end{align}
is the primitive quadratic Dirichlet character with conductor $|\Delta|$. Since $s>1$, it follows that $\Delta\neq 1$ and hence $\chi_\Delta$ is non-trivial, i.e.\ $\chi_\Delta(\ell_\star)=-1$ for some $\ell_\star$ in $(\mathbb{Z}/|\Delta|)^\times$. By the Chinese remainder theorem, the natural map $\mathbb{Z}_{C_V}^\times \to \mathbb{Z}_{|\Delta|}^\times$ is surjective, so in fact we can take $\ell_\star$ to be coprime to $C_V$. Thus, there is some $\ell_\star$ such that $\chi_\Delta(\ell_\star)=-1$, and hence 
\begin{align}
    \left(\frac{2t_1}{\ell_\star}\right)=-\left(\frac{2t_2}{\ell_\star}\right).
\end{align}
For this value of $\ell_\star$, one has that
\begin{align}
    \left|\sum_t^\star \left(\frac{2t}{\ell_\star}\right)w_t\right|<\sum_t^\star w_t=S_{00}
\end{align}
which contradicts the bound obtained earlier in \eqref{eqn:S00lowerbound}. Hence, it must be the case that precisely a single $t_\star$ has $w_{t_\star}>0$ (and hence $w_{t_\star}=S_{00}$), and therefore we learn that 
\begin{align}
    P_V(\tau)=\sum_{a\in \Rep(V)_{h\in\mathbb{Z}}} d_a\mathrm{ch}_{V_a}(\tau) = \frac{\theta_{0,1,t_\star}(\tau)}{\eta(\tau)}.
\end{align}
\end{proof}

After this long chain of constraints coming from modularity of the characters, we can finally provide the proof of Proposition \ref{prop:c=1Tannakian}.

\begin{proof}[Proof of Proposition \ref{prop:c=1Tannakian}]

Suppose that $V$ is a nice $c=1$ VOA. A straightforward calculation shows that, for the form of $P_V(\tau)$ guaranteed by Proposition \ref{prop:PVfinal}, 
\begin{align}
    \Pr\left[P_V(-1/\tau)\right] =  \frac{1}{\eta(\tau)\sqrt{2t_\star}}\theta_{0,1,t_\star}(\tau) = \frac{1}{\sqrt{2t_\star}}P_V(\tau).
\end{align}
By Theorem \ref{thm:TannakianVacChar}, this means that $V$ is Tannakian. 

Furthermore, the maximal conformal extension $V_{\Rep(G)}\supset V$ has vacuum character given by $P_V(\tau)$. Since $P_V(\tau)$ has a non-vanishing $q^{23/24}$ term, this implies that  the rank $l$ of the reductive Lie algebra defined
by the weight-1 space $(V_{\Rep(G)})_1$ is non-zero. By Theorem 2 of \cite{Dong:2002fs}, $l\leq c$ for a pseudo-unitary VOA, and since $c=1$, we must have that $l=c=1$. By Theorem 3 of \cite{Dong:2002fs}, this means that $V_{\Rep(G)}$ is a VOA associated to a rank-1 even integral lattice. In fact, $V_{\Rep(G)}\cong V_{2t_\star}$ because $V_{2t_\star}$ is the only rank-1 lattice VOA with vacuum character given by $\theta_{0,1,t_\star}(\tau)/\eta(\tau)$.
\end{proof}

\section{Classification of nice \texorpdfstring{$c=1$}{c=1} vertex operator superalgebras}\label{sec:VOSAclassification}

In the previous section, we classified nice $c=1$ vertex operator algebras, Theorem \ref{thm:c=1classification}. In this short section, we show that the classification of nice $c=1$ vertex operator \emph{super}algebras (VOSAs, i.e.\ fermionic chiral algebras) can be recovered as a corollary.\footnote{We emphasize that the prefix ``super'' in vertex operator superalgebras does not refer to supersymmetry, but rather to the fact that the theory is fermionic.}

\subsection{Assumptions and statement}

Let us say more precisely what we mean by a ``nice'' vertex operator superalgebra. We closely follow the conventions and discussion of \cite{Hohn:2023auw}.

We refer to Definition 2.1 of \cite{Hohn:2023auw} for the formal definition of a vertex operator superalgebra, which physically describes the chiral algebra of a fermionic CFT in the NS sector. One important aspect of the definition we emphasize is that it imposes spin-statistics. That is, a vertex operator superalgebra is $\frac12\mathbb{Z}$-graded by conformal dimension,
\begin{align}
    V=\bigoplus_{h\in\frac12\mathbb{Z}}V_h,
\end{align}
such that the bosonic states (i.e.\ states even under $(-1)^F$) have integer conformal dimension, while the fermionic states (odd under $(-1)^F$) have half-integer conformal dimension, 
\begin{align}
    V_{\bar 0}=\bigoplus_{h\in\mathbb{Z}}V_h, \ \ \ \ V_{\bar 1}=\bigoplus_{h\in \frac12 + \mathbb{Z}}V_h,
\end{align}
where $V_{\bar 0} := V^{(-1)^F}$ consists of the states even under $(-1)^F$. In the physics literature, spin-statistics is a consequence of unitarity, but we emphasize that there are interesting classes of (non-unitary) VOSAs for which spin and statistics are decoupled, such as those coming from 4d $\mathcal{N}=2$ superconformal field theories \cite{Beem:2013sza}. We will not have anything to say about such vertex operator superalgebras in this paper.

Just as in the bosonic case, we will impose a number of niceness conditions on our VOSAs. The notions of simple, $C_2$-cofinite, CFT-type, and self-contragredient are defined mutatis mutandis as for VOAs. 

A priori, there are two independent semisimplicity axioms: rationality and $(-1)^F$-rationality. The former imposes that all admissible modules (i.e.\ NS-sector modules) decompose into a direct sum of simple admissible modules. The latter requires that all admissible $(-1)^F$-twisted modules (i.e.\ R-sector modules) decompose into a direct sum of simple admissible $(-1)^F$-twisted modules. It turns out that we will only need to impose rationality in what follows, and not $(-1)^F$-rationality. When the rest of the adjectives of the previous paragraph are imposed, a VOSA which is rational is automatically $(-1)^F$-rational as well.

There are also two independent positivity axioms for a VOSA. We can say that a VOSA is pseudo-unitary if all of its simple modules have strictly positive conformal dimension, except for $V$ itself, which has $h=0$. We say that a VOSA is $(-1)^F$-pseudo-unitary if in addition, all simple $(-1)^F$-twisted modules have strictly positive conformal dimension. At present, there is no theorem which relates the two, so we will explicitly postulate both. 

\begin{definition}\label{def:niceVOSA}
    We say that a VOSA is nice if it is simple, rational, $C_2$-cofinite, self-contragredient, CFT-type, pseudo-unitary, and $(-1)^F$-pseudo-unitary.
\end{definition}
With these preliminaries out of the way, we can state the main theorem of this section.

\begin{theorem}[Classification of nice VOSAs with $c=1$]\label{thm:c=1VOSAs} If $V$ is a nice $c=1$ vertex operator superalgebra with $V_{\bar 1}\neq 0$, then $V$ is isomorphic to an odd lattice vertex operator superalgebra or the charge conjugation orbifold thereof, 
\begin{align}\label{eqn:c=1VOSAs}
    V_{m} \ \ (m\geq 1, \ m\text{ odd}), \ \ \ \ \ \  V_{m}^+ \ \ (m\geq 1, \ m\text{ odd}).
\end{align}
Conversely, every VOSA in this list is nice with $c=1$ and $V_{\bar 1}\neq 0$. Moreover, the VOSAs in this list are pairwise distinct.
\end{theorem}

We note in passing the following two isomorphisms,
\begin{align}
    V_1\cong F^{\otimes 2}, \ \ \ \ V_1^+\cong L(\sfrac12)\otimes F,
\end{align}
where $F\cong L(\sfrac12)\oplus L(\sfrac12,\sfrac12)$ is the free fermion VOSA of central charge $c=\sfrac12$, and $L(\sfrac12)$ is the (simple quotient of the) $c=\sfrac12$ Virasoro VOA.

We also remark that, even though there are no exceptional nice $c=1$ VOSAs (i.e.\ fermionic analogs of the VOAs $V_T$, $V_O$, and $V_I$), it does not follow that there are no exceptional fermionic full CFTs. Indeed, the chiral algebra of a fermionic CFT may happen to have no fermionic local operators at all, even though the full theory does. So in particular, there could possibly be a fermionic CFT with chiral algebra given by e.g.\ the VOA $V_O$. We will see that a theory of this kind is in fact realized in our companion paper \cite{grCFTs}.

\subsection{Proof}

The main idea of our proof is to use the fact that the specification of a nice VOSA $V$ is essentially equivalent to the specification of a nice VOA $V_{\bar 0}$ --- to serve as the bosonic states in $V$ --- equipped with the choice of a $V_{\bar 0}$-module $V_{\bar 1}$ with suitable properties so that it may serve as the space of fermionic states in $V$. This approach was used in slightly different ways in \cite{Rayhaun:2023pgc} and in \cite{Hohn:2023auw,BoyleSmith:2023xkd} to classify \emph{self-dual} nice VOSAs (i.e.\ chiral fermionic CFTs). In our theorem, we drop the assumption of self-duality (i.e.\ we allow our VOSAs to live on the boundary of potentially \emph{non-invertible} fermionic TQFTs in the bulk). However, in order for the problem to remain tractable, we restrict to central charge $c=1$. (By comparison, \cite{Hohn:2023auw} was able to classify nice self-dual VOSAs up through $c\leq 24$, assuming the uniqueness of the moonshine module.) 

More physically, our approach rests on the observation that any nice fermionic chiral algebra $V$ can be obtained from a nice bosonic chiral algebra $V_{\bar 0}$ by applying a ``chiral fermionization'' procedure. See the introduction for further details on this perspective.

Let us give a mathematically precise statement of chiral fermionization (see \cite{Hohn:2023auw} and references therein). We call a module $M$ of a nice VOA $V_{\bar 0}$ \emph{fermionic} if it  is a simple current with fusion rule $M\boxtimes M \cong V_{\bar 0}$ and has conformal dimension $h(M)\in\frac12+\mathbb{Z}_{\geq 0}$. 

\begin{proposition}\label{prop:chiralfermionization}
    If $V=V_{\bar 0}\oplus V_{\bar 1}$ is a nice VOSA, then $V_{\bar 0}$ is a nice VOA and $V_{\bar 1}$ is a fermionic $V_{\bar 0}$-module. Conversely, if $V_{\bar 0}$ is a nice VOA and $V_{\bar 1}$ is a fermionic $V_{\bar 0}$-module, then $V=V_{\bar 0}\oplus V_{\bar 1}$ uniquely admits the structure of a nice VOSA extension of $V_{\bar 0}$.
\end{proposition}

We say that $V$ is obtained from the pair $(V_{\bar 0},V_{\bar 1})$ via chiral fermionization. Proposition \ref{prop:chiralfermionization} says that every nice VOSA $V$ can be obtained from a pair $(V_{\bar 0},V_{\bar 1})$ consisting of a nice VOA $V_{\bar 0}$ and a choice $V_{\bar 1}$ of fermionic $V_{\bar 0}$-module. In particular, we can obtain a classification of nice $c=1$ VOSAs simply by enumerating all fermionic modules of the VOAs appearing in Theorem \ref{thm:c=1classification} and applying chiral fermionization.

Before we do this, there is a subtlety we must take into account. For a given nice VOA $V_{\bar 0}$, there may be non-isomorphic fermionic modules $V_{\bar 1}$ and $V_{\bar 1}'$ which nevertheless lead to isomorphic VOSAs, i.e.\ $V_{\bar 0}\oplus V_{\bar 1}\cong V_{\bar 0}\oplus V_{\bar 1}'$. 

The following proposition provides a criterion for determining when this happens. Recall that, for $g$ an automorphism of a VOA $W$ and $M$ a $W$-module, one can obtain another $W$-module $M^g$ which has the same underlying vector space as $M$, but with the action of $W$ modified as 
\begin{align}
    Y_{M^g}(v,z) = Y_{M}(g^{-1}v,z), 
\end{align}
where $v\in W$ and $Y_M(v,z)=\sum_n v_nz^{-n-1}$ with $v_n\in\mathrm{End}(M)$. (See e.g.\ \cite{dong1998twisted} for further background.)

\begin{proposition}\label{prop:equivferm}
    Suppose $V_{\bar 0}$ is a nice VOA and $V_{\bar 1}$, $V_{\bar 1}'$ are fermionic modules. Then the VOSAs $V\cong V_{\bar 0} \oplus V_{\bar 1}$ and $V'\cong V_{\bar 0}\oplus V_{\bar 1}'$ are isomorphic if and only if there is an automorphism $g\in \mathrm{Aut}(V_{\bar 0})$ such that $V_{\bar 1}^g\cong V_{\bar 1}'$ as $V_{\bar 0}$-modules.
\end{proposition}
\begin{proof}
     Suppose $\Phi:V\to V'$ is an isomorphism and let $g:=\Phi\vert_{V_{\bar 0}}$. Since isomorphisms preserve the $L_0$-grading, $g(V_{\bar 0})=V_{\bar 0}$ and hence $g$ defines an automorphism of $V_{\bar 0}$. Furthermore, $\psi:=\Phi\vert_{V_{\bar 1}}:V_{\bar 1}\to V_{\bar 1}'$ is a linear isomorphism satisfying 
    \begin{align}
        \psi Y_{V_{\bar 1}}(a,z)=Y_{V_{\bar 1}'}(ga,z)\psi,
    \end{align}
    and hence defines an isomorphism $\psi:V_{\bar 1}^g\to V_{\bar 1}'$ of $V_{\bar 0}$-modules.

    In the reverse direction, suppose $g$ is an automorphism of $V_{\bar 0}$ for which $V_{\bar 1}^g\cong V_{\bar 1}'$, and fix a choice of isomorphism $\psi:V_{\bar 1}^g\to V_{\bar 1}'$ of $V_{\bar 0}$-modules. Defining the linear map $\Phi:=g\oplus \psi: V_{\bar 0}\oplus V_{\bar 1}\to V_{\bar 0}\oplus V_{\bar 1}'$ allows us to pull back the VOSA structure on $V'=V_{\bar 0}\oplus V_{\bar 1}'$ to a VOSA structure $\widetilde{Y}_V(u,z)$ on $V=V_{\bar 0}\oplus V_{\bar 1}$ (a priori different from the given VOSA structure on $V$) as 
    \begin{align}
        \widetilde{Y}_V(u,z)v=\Phi^{-1}Y_{V'}(\Phi u,z)\Phi v.
    \end{align}
    So $V'$ is isomorphic to some VOSA structure on $V$, but by Proposition \ref{prop:chiralfermionization}, the structure of $V_{\bar 0}\oplus V_{\bar 1}$ as a VOSA extension of $V_{\bar 0}$ is unique up to isomorphism. Hence, $V'\cong V$.
\end{proof}

With Proposition \ref{prop:equivferm} in hand, we can move on to the proof of Theorem \ref{thm:c=1VOSAs}.

\begin{proof}[Proof of Theorem \ref{thm:c=1VOSAs}] 

Let us treat the forward direction. If $V$ is a nice fermionic VOSA with $c=1$, then by Proposition \ref{prop:chiralfermionization}, its even subalgebra $V_{\bar 0}$, being a nice VOA, must be one of the VOAs appearing in Theorem \ref{thm:c=1classification}. Then $V$ can be reconstructed as $V_{\bar 0}\oplus V_{\bar 1}$ for some fermionic $V_{\bar 0}$-module. So let us treat the different possibilities for $V_{\bar 0}$ case-by-case.

\paragraph{Case: $V_{\bar 0}\cong V_T$, $V_O,$ or $V_I$.} We can immediately conclude that there are no nice $c=1$ VOSAs $V$ with $V_{\bar 0}$ given by $V_T$, $V_O$, or $V_I$. Indeed, by inspecting the simple modules of these VOAs recorded in Appendix \ref{app:bestiary}, one sees that there are none with conformal dimension $h\in \frac12 + \mathbb{Z}_{\geq 0}$, and hence they cannot be chirally fermionized. This conclusion does not rely on $V_I$ being nice: if it happens not to be, then it could not arise as the even subalgebra of a nice $c=1$ VOSA anyways.

\paragraph{Case: $V_{\bar 0}\cong V_{2k}$.} Next, suppose that $V_{\bar 0}$ is isomorphic to a rank-1 lattice VOA. The representation category of $V_{2k}$ is pointed with fusion rules described by the group $\mathbb{Z}_{2k}$. Thus, the only non-trivial simple module $M$ with fusion rule $M\boxtimes M\cong V_{2k}$ is $M=V_{2k,k}$. This module has conformal dimension $k/4$, which resides in $\frac12+\mathbb{Z}_{\geq 0}$ if and only if $k\equiv 2~\mathrm{mod}~4$, i.e.\ if and only if $k=2m$ with $m$ an odd integer. Thus, the VOSAs with even part given by a lattice VOA are precisely of the form
\begin{align}\label{eqn:latticeVOSA}
    V_{4m}\oplus V_{4m,2m}\cong V_{m}, \ \ \ \ (m\text{ odd}).
\end{align}
We have identified $V_{4m}\oplus V_{4m,2m}$ with the lattice VOA $V_m$ using the fact that the lattice cosets underlying the left-hand side of \eqref{eqn:latticeVOSA} combine as
\begin{align}
    \sqrt{4m}\mathbb{Z}\oplus \left(\frac{2m}{\sqrt{4m}}+\sqrt{4m}\mathbb{Z}\right)= 2\sqrt{m}\mathbb{Z}\oplus \left(\sqrt{m}+2\sqrt{m}\mathbb{Z} \right)=\sqrt{m}\mathbb{Z}. 
\end{align}
\paragraph{Case: $V_{\bar 0}\cong V_{2k}^+$.} Finally, suppose that $V_{\bar 0}$ is isomorphic to $V_{2k}^+$ for some $k$. By inspecting the simple modules reported in Table \ref{tab:chargeconjugationmodules}, we see that $V_{2k}^+$ only admits fermionic modules when $k=2m$ for $m$ an odd integer. Moreover, there are precisely 2 fermionic modules, $V_{4m,2m}^+$ and $V_{4m,2m}^-$. Thus, for each odd integer $m$, we have two (possibly isomorphic) VOSAs with even subalgebra $V_{4m}^+$, namely, 
\begin{align}\label{eqn:twoCCVOSAs}
    V_{4m}^+\oplus V_{4m,2m}^+, \ \ \ \ V_{4m}^+\oplus V_{4m,2m}^-.
\end{align}
However, the automorphism group of $V_{4m}^+$ with $m$ odd is always $\mathrm{Aut}(V_{4m}^+)\cong \mathbb{Z}_2$ (see Table \ref{tab:aut}) and, from Proposition \ref{prop:V2m+perm}, we see that the non-trivial involution exchanges $V_{4m,2m}^+$ and $V_{4m,2m}^-$. Thus, by Proposition \ref{prop:equivferm}, these two VOSAs built on $V_{4m}^+$ are isomorphic. 

On the other hand, $V_m^+$ is also a nice VOSA with even subalgebra $V_{4m}^+$ so, from the fact that $V_{4m}^+\oplus V_{4m,2m}^+$ is the only chiral fermionization of $V_{4m}^+$, we learn that we can identify
\begin{align}
    V_{m}^+\cong V_{4m}^+\oplus V_{4m,2m}^+.
\end{align}
This concludes the forward direction.

The niceness of the VOSAs in \eqref{eqn:c=1VOSAs} follows from Proposition \ref{prop:chiralfermionization}, using the fact that $V_{\bar 0}$ is nice and $V_{\bar 1}$ is a fermionic $V_{\bar 0}$-module. The fact that they are pairwise distinct follows from the fact that their even subalgebras are all distinct. 
\end{proof}

\section{Future directions}\label{sec:future}

In this paper, we have given a mathematically rigorous proof of the classification of nice $c=1$ vertex operator algebras and vertex operator superalgebras.  In our companion paper \cite{grCFTs}, we will extend this to a classification of bosonic and fermionic rational $c=1$ conformal field theories with both left- and right-movers. There are a number of directions for future research.

\begin{enumerate}[label=\arabic*)]
    \item It would be interesting to prove a classification theorem at $c=1$ where rationality is dropped. Conjecture \ref{conj:c=1irrational} records one potential target which may be within reach of existing methods.
    \item The next frontier in the classification of nice VOAs, organized by increasing central charge, is $c>1$. Physical intuition suggests that there are ``deserts'' in the space of conformal field theories with compact spectrum. One concrete corollary of this lore is the expectation that there exists an $\epsilon>0$ such that there are \emph{no} nice VOAs with central charge $c$ in the interval $(1,1+\epsilon)$. It would be remarkable to prove this expectation. See \cite{Benjamin:2026lbj} for recent numerical work in this direction.
    \item It would be interesting to put our discussion in Appendix \ref{app:QuOpTopLin} of the quantum operations of the nice $c=1$ VOAs on firmer mathematical footing. 
    \item One expects, on the power of the growing dictionary between (unitary) VOAs and conformal nets \cite{Carpi:2015fga,Carpi:2023onx,henriques2025every}, that a theorem similar to our Theorem \ref{thm:c=1classification} should hold for conformal nets. Xu \cite{xu2005strong} obtained such a theorem by assuming a ``spectrum condition'' which one would like to remove.
    \item As mentioned in the introduction, we will show elsewhere that a similar strategy to the one employed in this paper will yield the classification of nice $c=\sfrac32$ vertex operator superalgebras with $N=1$ supersymmetry (see \cite{Dixon:1988ac,Wendland:2004pp} for expectations coming from physics). It would be interesting to find other classes of VO(S)As which can be classified using Tannakian methods. 
\end{enumerate}

\section*{Acknowledgements}
B.R.~is grateful to Yichul Choi and Ho Tat Lam for an inspiring collaboration on a related project. We also thank Sebastiano Carpi, André Henriques, Adrià Marín-Salvador, Robert McRae, Sven Möller, Nivedita, Sylvain Ribault, Nathan Seiberg, Sahand Seifnashri, Yifan Wang, and Edward Witten for helpful discussions. B.R.~acknowledges support from the Leinweber Foundation, the Sivian Fund, the U.S.\ National Science Foundation (NSF) under grant PHY-2210533, and the Department of Energy (DOE) under grant DE-SC0009988. B.R.~also thanks the Yang Institute for Theoretical Physics and the Simons Center for Geometry and Physics where much of this work was carried out. The research of TG was partially supported by NSERC (Canada).

\appendix

\section{Bestiary of nice \texorpdfstring{$c=1$}{c=1} vertex operator algebras}\label{app:bestiary}

In this appendix, we record some of the basic data of the  ``nice'' $c=1$ VOAs, such as their simple modules, their characters, and their modular data. 

Much of this data is presented for completeness rather than out of necessity, and to collect all the information in one place. We also take this appendix as an opportunity to correct several typos in the literature, most notably in the modular S-matrices for the nice $c=1$ VOAs. Presumably, many of these errors were incurred because the derivations used the characters of the VOA, which are generally not linearly independent and therefore insufficient for pinning down the modular data. The more careful categorical analysis of \cite{Galindo:2024law} resolves these ambiguities for the charge conjugation orbifolds, and our Appendix \ref{app:exceptional} resolves this for the three exceptional VOAs.

We note right away that we work in conventions where $S$ is what the modular representation associated to $V$ assigns to $\left(\begin{smallmatrix} 0 & -1 \\ 1 & 0 \end{smallmatrix}\right)$, rather than to $\left(\begin{smallmatrix} 0 & 1 \\ -1 & 0 \end{smallmatrix}\right)$. In particular, $(ST)^3$ is the charge conjugation matrix, rather than the identity.

\subsection{Lattice vertex operator algebras}\label{subsec:lattice}

We use the notation $V_{2m}$ to denote the VOA associated to the rank-1 even integral lattice $\sqrt{2m}\mathbb{Z}$. For a physicist, the theory $V_{2m}$ can be thought of as the maximal chiral algebra of the 1+1d compact boson of radius $R=\sqrt{2p/q}$ for any pair of coprime integers $p,q$ with $pq=m$. Our conventions are such that $R^2=2$ is the self T-dual radius, so that in particular $V_2\cong \widehat{\mathfrak{su}}(2)_1$. 

The simple modules $V_{2m,r}$ of $V_{2m}$ are labeled by elements $r\in\mathbb{Z}_{2m}\cong L'/L$ of the discriminant group of the lattice $L=\sqrt{2m}\mathbb{Z}$ on which $V_{2m}$ is based,
\begin{align}
\begin{split}
\mathbb{Z}_{2m} &\xrightarrow{\sim} L'/L \\
    r&\mapsto L+\frac{r}{\sqrt{2m}}.
\end{split}
\end{align}
Physically, we think of the simple module $V_{2m,r}$ as being spanned by vertex operators $e^{i \frac{r+2mn}{\sqrt{2m}} \phi(z)}$ with $n\in\mathbb{Z}$, as well as their descendants under the $\mathfrak{u}(1)$ Kac--Moody algebra. Accordingly, the conformal dimension of the module $V_{2m,r}$ is 
\begin{align}
    h(V_{2m,r})=\frac{r^2}{4m}, \ \ \ \ \ \text{ for } r=-m,-m+1,\dots,m-1,
\end{align}
and its character/graded-dimension is given by
\begin{align}\label{eqn:freebosoncharacters}
    \mathrm{ch}_{V_{2m,r}}(\tau) := \mathrm{Tr}_{V_{2m,r}}q^{L_0-\frac{1}{24}}= \frac{1}{\eta(\tau)}\sum_{n\in\mathbb{Z}}q^{\frac{(r+2mn)^2}{4m}}=:\frac{\Theta_{m,r}(\tau)}{\eta(\tau)},
\end{align}
where $\eta(\tau)$ is the Dedekind eta function and $\Theta_{m,r}(\tau)$ is the theta function associated to the lattice coset $\sqrt{2m}\mathbb{Z}+\frac{r}{\sqrt{2m}}$. The modular S-matrix is given by the Weil representation associated to $L=\sqrt{2m}\mathbb{Z}$, up to an overall factor coming from the multiplier system of the Dedekind eta function: 
\begin{align}\label{eqn:latticeSmatrix}
    S_{rs}= \frac{1}{\sqrt{2m}}e^{-\pi \ri rs/m}.
\end{align}
The fusion rules coincide with the group multiplication on $\mathbb{Z}_{2m}$, i.e.\ 
\begin{align}
    V_{2m,r}\boxtimes V_{2m,s} \cong V_{2m,r+s}.
\end{align}
The modular fusion category $\mathrm{Rep}(V_{2m})$ defined by the representation category of $V_{2m}$ is equivalent to the category of anyons in $U(1)_{2m}$ Abelian Chern-Simons theory, which can be thought of as the bulk topological field theory supporting $V_{2m}$ on its boundary.

\subsection{Charge conjugation orbifolds}

Next, we move on to the charge conjugation orbifold chiral algebras $V_{2m}^+$. The VOA $V_{2m}^+$ is defined as the fixed-point subalgebra of $V_{2m}$ with respect to its $\mathbb{Z}_2^{\mathrm{C}}$ charge conjugation symmetry. Physically, the $V_{2m}^+$ arise as chiral algebras of 1+1d CFTs on the orbifold branch. We assemble the data here from various references in the literature, including \cite{Dijkgraaf:1989hb,abe2001fusion,dong1999representations,Galindo:2024law}.

We note the exceptional isomorphism
\begin{align}
    V_2^+\cong V_8,
\end{align}
which allows us to restrict in much of what follows to $m>1$. Physically, this exceptional isomorphism is related to the fact that the circle and orbifold branches of the $c=1$ conformal manifold meet at the Kosterlitz-Thouless point. There is also the isomorphism 
\begin{align}
    V_4^+\cong L(\sfrac12)\otimes L(\sfrac12),
\end{align}
where $L(\sfrac12)$ is the Ising VOA, i.e.\ the chiral algebra of the Ising CFT, known to mathematicians as the simple quotient of the $c=\sfrac12$ Virasoro VOA.

\begin{table}
\begin{center}
    \begin{tabular}{c|c|c}
    $M$ & $h(M)$ & $\mathrm{ch}_M(\tau)$ \\\toprule
    $V_{2m}^+$ & $0$ & $\frac12\mathrm{ch}_{V_{2m}}+\frac12\vartheta[\substack{0\\1/2}]$ \\
    $V_{2m}^-$ & $1$ & $\frac12\mathrm{ch}_{V_{2m}}-\frac12 \vartheta[\substack{0\\1/2}]$ \\
    $V_{2m,m}^\pm$ & $\sfrac{m}{4}$ & $\frac12 \mathrm{ch}_{V_{2m,m}}$ \\
    $V_{2m,r}$ & $\sfrac{r^2}{4m}$ & $\mathrm{ch}_{V_{2m,r}}$ \\
    $V_{2m}^{T_i,+}$ & $\sfrac{1}{16}$ & $\frac12\vartheta[\substack{1/4\\0}]+\frac12\vartheta[\substack{1/4\\1/2}]$ \\
    $V_{2m}^{T_i,-}$ & $\sfrac{9}{16}$ & $\frac12\vartheta[\substack{1/4\\0}]-\frac12\vartheta[\substack{1/4\\1/2}]$
\end{tabular}\caption{The simple modules of $V_{2m}^+$. Here, $r=1,\dots,m-1$ and $i=1,2$.}\label{tab:chargeconjugationmodules}
\end{center}
\end{table}

On general grounds, given a strongly rational vertex operator algebra $V$ and a finite group $G$ of automorphisms, it is known that, if $V^G$ is also strongly rational, then the simple modules of $V^G$ all occur inside of ordinary and $g$-twisted modules of $V$ as $g$ runs over elements of $G$ \cite{dong2017orbifold}. It turns out that if one applies this principle to the charge conjugation orbifolds $V_{2m}^+$, one finds that there are $m+7$ simple modules.

For example,  the simple $V_{2m}$-module $V_{2m,r}$ remains simple when restricted to a $V_{2m}^+$-module for $r\neq 0,m$ (though $V_{2m,r}$ and $V_{2m,-r}$ become isomorphic as $V_{2m}^+$-modules). On the other hand, the $V_{2m}$-modules $V_{2m}$ and $V_{2m,m}$ can be split up into their even/odd parts $V_{2m}^\pm$ and $V_{2m,m}^\pm$ with respect to $\mathbb{Z}_2^{\mathrm{C}}$ to obtain simple $V_{2m}^+$-modules. Similarly, there are two $\mathbb{Z}_2^{\mathrm{C}}$-twisted $V_{2m}$-modules $V_{2m}^{T_i}$, and they can also be split into their even/odd parts $V_{2m}^{T_i,\pm}$ to obtain four more simple $V_{2m}^+$-modules. 

We summarize the representation theory of $V_{2m}^+$ in Table \ref{tab:chargeconjugationmodules}. The conformal dimension of each module is given in the second column, and the characters are given in the third column, where we have defined 
\begin{align}
    \vartheta[\substack{\alpha\\\beta}](\tau) = \frac{1}{\eta(\tau)} \sum_{n\in\mathbb{Z}}q^{(n+\alpha)^2}e^{2\pi \ri n \beta}.
\end{align}
We also provide the modular S-matrix for $V_{2m}^+$ in Table \ref{tab:CCSmatrix}, which can be computed using the results of \cite{Galindo:2024law}. The fusion rules can be read off using Verlinde's formula \cite{Verlinde:1988sn}. (See \cite{abe2001fusion} for explicit expressions for the fusion rules.) The modular fusion category $\Rep(V_{2m}^+)$ agrees with the category of anyons of $O(2)_{2m}$ Chern-Simons theory, which is the bulk TQFT supporting $V_{2m}^+$ on its boundary.

\begin{table}
\begin{center}
    \begin{tabular}{c|cccc}
        & $V_{2m}^\epsilon$ & $V_{2m,m}^\epsilon$ & $V_{2m,r}$ & $V_{2m}^{T_i,\epsilon}$  \\\midrule
        $V_{2m}^{\epsilon'}$ & $1$ & $1$ & $2$ & $\epsilon' \sqrt{m}$ \\
        $V_{2m,m}^{\epsilon'}$ &  & $(-1)^m$ & $2(-1)^r$  & $(-\sqrt{-1})^m\,\sigma_{\epsilon',i}\sqrt{m}$\\
        $V_{2m,s}$ & & & $4\cos(\pi rs/m)$ & $0$ \\
        $V_{2m}^{T_j,\epsilon'}$ & & & & $\epsilon\epsilon'\,
        \frac{1+\sigma_{i,j}(-\sqrt{-1})^m}{2}\sqrt{2m}$
    \end{tabular}
    \caption{The modular S-matrix of $V_{2m}^+$ up to the overall
    normalization $1/\sqrt{8m}$. Here $\epsilon,\epsilon'\in \{\pm\}$ and $\sigma_{a,b}:=2\delta_{a,b}-1$, where $\delta_{a,b}$ is the Kronecker delta. (We use the convention that $\delta_{+,1}=\delta_{-,2}=1$ and $\delta_{+,2}=\delta_{-,1}=0$.)}\label{tab:CCSmatrix}
\end{center}
\end{table}

\subsection{Exceptional models}

We move on to the three exceptional $c=1$ VOAs. See \cite{Ginsparg:1987eb,Dijkgraaf:1989hb,Cappelli:2002wq,dong2013representations,wu2016representations,dong2015fusion,liao2026fusion,xu2026c2cofiniteness} for further details omitted here.

These are obtained from $V_2\cong \widehat{\mathfrak{su}}(2)_1$ by passing to the fixed points of the action of the three exceptional finite subgroups of $SO(3)\subset \textsl{PSL}(2,\mathbb{C})\cong \mathrm{Aut}(V_2)$. We write these as 
\begin{align}
    V_T:=V_2^{A_4}, \ \ \ \ \ \ V_O:= V_2^{S_4}, \ \ \ \ \ \ V_I:= V_2^{A_5},
\end{align}
where $T$, $O$, and $I$ stand for tetrahedral, octahedral, and icosahedral, as $A_4$, $S_4$, and $A_5$ are the orientation-preserving symmetries of the tetrahedron, the octahedron, and the icosahedron, respectively.

Consider e.g.\ the VOA $V_T$. (The comments of this paragraph apply mutatis mutandis to $V_O$ and $V_I$.) We describe in Appendix \ref{app:exceptional} how $\Rep(V_T)$ can be realized as a modular fusion subcategory of the twisted quantum double $Z(\textsl{Vec}_{2.A_4}^{\tilde\omega})$ for a particular choice of $\tilde\omega \in H^3(2.A_4,U(1))$.  Accordingly, we will label simple modules of $V_T$ by pairs $[g,\chi]$, where $g$ is an element of $2.A_4$ (only the conjugacy class matters), and $\chi$ is the character of an irreducible representation of the centralizer of $g$ in $2.A_4$. The deeper meaning of these labels $[g,\chi]$ will be explained in Appendix \ref{app:exceptional}, using some of the technology developed in \cite{Gannon:2024tcl,Gannon:2026ttf}. For the purposes of this section, they can be regarded merely as indexing the different simple modules. 

The basic data of the simple modules of $V_T$, $V_O$, and $V_I$ are given in Table \ref{tab:VTreps}, Table \ref{tab:VOreps}, and Table \ref{tab:VIreps}, respectively. The modular S-matrices are recorded in Table \ref{tab:VTSmatrix}, Table \ref{tab:VOSmatrix}, and Table \ref{tab:VISmatrix}, respectively. We remark that every simple module of $V_O$ and $V_I$ is self-dual. The duals of simple $V_T$-modules are described in Equation \eqref{eqn:VTperm}.

\newpage

    \begin{table}
\begin{center}
    \begin{tabular}{c|c|c}
$M$ & $h(M)$ & $\mathrm{ch}_M(\tau)$ \\\toprule
$[1,\chi_{1,0}]$ & $0$ & $\frac{1}{12} \vartheta [\substack{0\\0}]+\frac{2}{3} \vartheta [\substack{0\\1/3}]+\frac{1}{4} \vartheta [\substack{0\\1/2}]$ \\
$[1,\chi_{1,i}]$  & $4$ & $\frac{1}{12} \vartheta[\substack{0\\0}]-\frac{1}{3} \vartheta [\substack{0\\1/3}]+\frac{1}{4} \vartheta [\substack{0\\1/2}]$ \\
$[1,\chi_{3}]$ & $1$ & $\frac14\vartheta[\substack{0\\0}]-\frac14\vartheta[\substack{0\\1/2}]$ \\
$[\gamma_4,\psi_4^{2j}]$ & $\sfrac{1}{16},\sfrac{9}{16}$ & $\frac12 \vartheta[\substack{1/4\\0}]+\frac{(-1)^j}{2} \vartheta[\substack{1/4\\1/2}]$ \\
$[\gamma_6^{1},\psi_6^{2m}]$ & $\sfrac1{36},\sfrac{49}{36},\sfrac{25}{36}$ & $\frac13\vartheta[\substack{1/6\\0}]+\frac13\omega^{-m}\vartheta[\substack{1/6\\1/3}]+\frac13\omega^{m}\vartheta[\substack{1/6\\2/3}]$\\
$[\gamma_6^{-1},\psi_6^{2m}]$ & $\sfrac1{36},\sfrac{25}{36},\sfrac{49}{36}$ & $\frac13\vartheta[\substack{1/6\\0}]+\frac13\omega^{m}\vartheta[\substack{1/6\\1/3}]+\frac13\omega^{-m}\vartheta[\substack{1/6\\2/3}]$\\
$[\gamma_6^{2},\psi_6^{2m}]$  & $\sfrac19,\sfrac{16}9,\sfrac49$ & $\frac13\vartheta[\substack{1/3\\0}]+\frac13\omega^{-m}\vartheta[\substack{1/3\\1/3}]+\frac13\omega^{m}\vartheta[\substack{1/3\\2/3}]$\\
$[\gamma_6^{-2},\psi_6^{2m}]$  & $\sfrac19,\sfrac49,\sfrac{16}9$ & $\frac13\vartheta[\substack{1/3\\0}]+\frac13\omega^{m}\vartheta[\substack{1/3\\1/3}]+\frac13\omega^{-m}\vartheta[\substack{1/3\\2/3}]$\\
$[z,\chi_{2,0}]$  & $\sfrac14$ & $\frac16\vartheta[\substack{1/2\\0}]-\frac23 \omega^2\vartheta[\substack{1/2\\1/3}]$ \\
$[z,\chi_{2,i}]$  & $\sfrac94$ & $\frac16\vartheta[\substack{1/2\\0}]+\frac13\omega^2\vartheta[\substack{1/2\\1/3}]$
    \end{tabular}\caption{The $21$ simple modules of $V_T$, labeled by the corresponding simple
    objects $[\gamma,\chi]$ of $Z(\mathrm{Vec}_{2.A_4}^{\tilde\omega})$, with $[1,\chi_{1,0}]=V_T$
    itself. Here $i\in \{1,2\}$, $j\in\{0,1\}$, $m\in \{0,1,2\}$, and $\omega=e^{2\pi \ri/3}$.}
    \label{tab:VTreps}
\end{center}
\end{table}

\begin{table}
\begin{center}
    \begin{tabular}{c|ccccc}
         & $[1,\chi_{1,i}]$ & $[1,\chi_{3}]$ & $[z,\chi_{2,i}]$ & $[\gamma_4,\psi_4^{2j}]$ & $[\gamma_6^{b},\psi_6^{2m}]$ \\\midrule
        $[1,\chi_{1,i'}]$ & $1$ & $3$ & $2$ & $6$ & $4\,\zeta^{-6bi'}$ \\
        $[1,\chi_{3}]$ &  & $9$ & $6$ & $-6$ & $0$ \\
        $[z,\chi_{2,i'}]$ &  &  & $-4$ & $0$ & $4\,(-1)^{b+1}\zeta^{6bi'}$ \\
        $[\gamma_4,\psi_4^{2j'}]$ &  &  &  & $6\sqrt{2}\,(-1)^{j+j'}$ & $0$ \\
        $[\gamma_6^{b'},\psi_6^{2m'}]$ &  &  &  &  & $4\,\zeta^{-bb'-6(mb'+m'b)}$
    \end{tabular}
    \caption{The modular S-matrix of $V_T$ up to the overall normalization $\frac{1}{12\sqrt{2}}$.
    Here $\zeta:=e^{2\pi \ri/18}$, and the indices run over $i,i'\in\{0,1,2\}$, $j,j'\in\{0,1\}$,
    $b,b'\in\{\pm1,\pm2\}$, and $m,m'\in\{0,1,2\}$.}
    \label{tab:VTSmatrix}
\end{center}
\end{table}

\newpage

\newpage

\begin{table}
\begin{center}
    \begin{tabular}{>{\small}c|>{\small}c|>{\small}c}
$M$ & $h(M)$ & $\mathrm{ch}_M(\tau)$ \\\toprule
$[1,\chi_{1,0}]$ & $0$ & $\frac{1}{24} \left(\vartheta[\substack{0\\0}] +3 \vartheta [\substack{0\\1/4}]+4 \vartheta [\substack{0\\1/3}]+9 \vartheta [\substack{0\\1/2}]+4 \vartheta [\substack{0\\2/3}]+3 \vartheta[\substack{0\\3/4}]\right)$\\
$[1,\chi_{1,1}]$ & $9$ & $\frac{1}{24} \left(\vartheta[\substack{0\\0}] -3 \vartheta [\substack{0\\1/4}]+12 \vartheta [\substack{0\\1/3}]-3 \vartheta [\substack{0\\1/2}]-4 \vartheta [\substack{0\\2/3}]-3 \vartheta[\substack{0\\3/4}]\right)$ \\
$[1,\chi'_{2}]$ & $4$ & $\frac{1}{12} \left(\vartheta [\substack{0\\0}]-4 \vartheta [\substack{0\\1/3}]+3 \vartheta [\substack{0\\1/2}]\right)$\\
$[1,\chi_{3,0}]$ & $1$ & $\frac{1}{8} \left(\vartheta[\substack{0 \\ 0}]+\vartheta [\substack{0\\1/4}]-3 \vartheta [\substack{0\\1/2}]+\vartheta [\substack{0\\3/4}]\right)$ \\
$[1,\chi_{3,1}]$ & $4$ & $\frac{1}{8} \left(\vartheta[\substack{0 \\ 0}]-\vartheta [\substack{0\\1/4}]+ \vartheta [\substack{0\\1/2}]-\vartheta [\substack{0\\3/4}]\right)$ \\
$[z,\chi_{2,0}]$ & $\sfrac14$ & $\frac{1}{24} \left(4 \left( \zeta ^*\vartheta[\substack{1/2\\2/3}]+\zeta  \vartheta[\substack{1/2\\1/3}]\right)+2 \vartheta [\substack{1/2\\0}]+3 \left((1+\ri) \vartheta[\substack{1/2\\1/4}]+(1-\ri) \vartheta [\substack{1/2\\3/4}]\right)\right)$ \\
$[z,\chi_{2,1}]$ & $\sfrac{25}4$ & $\frac{1}{24} \left(2 \vartheta[\substack{1/2\\0}]+12 \zeta \vartheta[\substack{1/2\\1/3}]-4 \zeta^* \vartheta [\substack{1/2\\2/3}]-3 ((1+\ri) \vartheta[\substack{1/2\\1/4}]+(1-\ri) \vartheta [\substack{1/2\\3/4}])\right)$\\
$[z,\chi_{4}]$ & $\sfrac94$ & $\frac{1}{6} \left(\vartheta [\substack{1/2\\0}]-2 e^{\pi \ri/3} \vartheta [\substack{1/2\\1/3}]\right)$ \\
$[\gamma_4,\psi_4^{2j}]$ & $\sfrac{1}{16},\sfrac9{16}$ & $\frac{1}{2} \left(\vartheta [\substack{1/4\\0}]+(-1)^j\vartheta[\substack{1/4\\1/2}]\right)$ \\
$[\gamma_6^{b},\psi_6^{2k}]$ & $\sfrac{r^2}{36}$ & $\mathrm{ch}_{V_{18,r}}\,,\ \ r\equiv \pm(b+6k) \bmod 18$ with $0\leq r\leq 9$\\
$[\gamma_8^{c},\psi_8^{2l}]$ & $\sfrac{s^2}{64}$ & $\mathrm{ch}_{V_{32,s}}\,,\ \ s\equiv \pm(c+8l) \bmod 32$ with $0\leq s \leq 16$\\
$[\gamma_8^{2},\psi_8^{0}]$ & $\sfrac{1}{16}$ & $\frac{1}{8} \left(\vartheta [\substack{1/4\\1/4}]-\vartheta [\substack{1/4\\0}]-\vartheta [\substack{1/4\\1/2}]+\vartheta [\substack{1/4\\3/4}]+3 \vartheta[\substack{3/4\\0}]+\ri \vartheta[\substack{3/4\\1/4}] -3 \vartheta [\substack{3/4\\1/2}]-\ri \vartheta[\substack{3/4\\3/4}] \right)$\\
$[\gamma_8^{2},\psi_8^{2}]$ & $\sfrac{25}{16}$ & $\frac{1}{8} \left(5\vartheta [\substack{1/4\\0}]-\ri\vartheta [\substack{1/4\\1/4}]-5\vartheta [\substack{1/4\\1/2}]+\ri\vartheta [\substack{1/4\\3/4}]-3 \vartheta[\substack{3/4\\0}]- \vartheta[\substack{3/4\\1/4}] -3 \vartheta [\substack{3/4\\1/2}]- \vartheta[\substack{3/4\\3/4}] \right)$\\
$[\gamma_8^{2},\psi_8^{4}]$ & $\sfrac{49}{16}$ & $\frac{1}{8} \left(5\vartheta [\substack{1/4\\0}]-\vartheta [\substack{1/4\\1/4}]+5\vartheta [\substack{1/4\\1/2}]-\vartheta [\substack{1/4\\3/4}]-3 \vartheta[\substack{3/4\\0}]-\ri \vartheta[\substack{3/4\\1/4}] +3 \vartheta [\substack{3/4\\1/2}]+\ri \vartheta[\substack{3/4\\3/4}] \right)$ \\
$[\gamma_8^{2},\psi_8^{6}]$ & $\sfrac9{16}$ & $\frac{1}{8} \left(\ri\vartheta [\substack{1/4\\1/4}]-\vartheta [\substack{1/4\\0}]+\vartheta [\substack{1/4\\1/2}]-\ri\vartheta [\substack{1/4\\3/4}]+3 \vartheta[\substack{3/4\\0}]+ \vartheta[\substack{3/4\\1/4}] +3 \vartheta [\substack{3/4\\1/2}]+ \vartheta[\substack{3/4\\3/4}] \right)$
    \end{tabular}
    \caption{The $28$ simple modules of $V_O$, labeled by the corresponding simple objects
    $[\gamma,\chi]$ of $Z(\mathrm{Vec}_{2.S_4}^{\tilde\omega})$, with $[1,\chi_{1,0}]=V_O$ itself.
    Here $j\in \{0,1\}$, $b\in\{1,2\}$, $k\in\{0,1,2\}$, $l\in\{0,1,2,3\}$,  and $c\in\{1,3\}$.
    Further, $\zeta=e^{\ri\pi/3}$. }\label{tab:VOreps}
\end{center}
\end{table}

\begin{table}
\begin{center}
    \begin{tabular}{>{\small}c|>{\small}c|>{\small}c}
$M$ & $h(M)$ & $\mathrm{ch}_M(\tau)$ \\\toprule
$[1,\chi_{1}]$ & $0$ & $\frac{1}{60} \vartheta[\substack{0\\0}]+ \frac{1}{3} \vartheta[\substack{0\\1/3}]+ \frac{1}{5} \vartheta[\substack{0\\1/5}]+ \frac{1}{5} \vartheta[\substack{0\\2/5}]+ \frac{1}{4} \vartheta[\substack{0\\1/2}]$ \\
$[z,\chi_{2,s}]$ & $\sfrac14,\sfrac{49}{4}$ & $\frac{1}{30} \vartheta[\substack{1/2\\0}]- \frac{1 + s\sqrt{5}}{10} \bar{\zeta}_5^{\,2} \vartheta[\substack{1/2\\1/5}]- \frac{1 - s\sqrt{5}}{10} \zeta_5 \vartheta[\substack{1/2\\2/5}]- \frac{\bar{\zeta}_3}{3} \vartheta[\substack{1/2\\1/3}]$\\
$[1,\chi_{3,s}]$ & $1,9$ & $\frac{1}{20} \vartheta[\substack{0\\0}]+ \frac{1 + s\sqrt{5}}{10} \vartheta[\substack{0\\1/5}]+ \frac{1 - s\sqrt{5}}{10} \vartheta[\substack{0\\2/5}]- \frac{1}{4} \vartheta[\substack{0\\1/2}]$\\
$[1,\chi_{4,+}]$ & $9$ & $\frac{1}{15} \vartheta[\substack{0\\0}]+ \frac{1}{3} \vartheta[\substack{0\\1/3}]- \frac{1}{5} \vartheta[\substack{0\\1/5}]- \frac{1}{5} \vartheta[\substack{0\\2/5}]$ \\
$[z,\chi_{4,-}]$ & $\sfrac94$ & $\frac{1}{15} \vartheta[\substack{1/2\\0}]- \frac{\bar{\zeta}_5^{\,2}}{5} \vartheta[\substack{1/2\\1/5}]- \frac{\zeta_5}{5} \vartheta[\substack{1/2\\2/5}]+ \frac{\bar{\zeta}_3}{3} \vartheta[\substack{1/2\\1/3}]$ \\
$[1,\chi_5]$ & $4$ & $\frac{1}{12} \vartheta[\substack{0\\0}]- \frac{1}{3} \vartheta[\substack{0\\1/3}]+ \frac{1}{4} \vartheta[\substack{0\\1/2}]$ \\
$[z,\chi_6]$ & $\sfrac{25}{4}$ & $\frac{1}{10} \vartheta[\substack{1/2\\0}]+ \frac{\bar{\zeta}_5^{\,2}}{5} \vartheta[\substack{1/2\\1/5}]+ \frac{\zeta_5}{5} \vartheta[\substack{1/2\\2/5}]$ \\
$[\gamma_4,\psi_4^{2i}]$ & $\sfrac{1}{16},\sfrac{9}{16}$ & $\mathrm{ch}_{V_{8,1+2i}}$\\
$[\gamma_6^b,\psi_6^{2j}]$ & $\sfrac{r^2}{36}$ & $\mathrm{ch}_{V_{18,r}}\,,\ \ r\equiv\pm(b+6j) \bmod 18$ with $0\leq r\leq 9$\\
$[\gamma_{10}^c,\psi_{10}^{2k}]$ & $\sfrac{t^2}{100}$ & $\mathrm{ch}_{V_{50,t}}\,,\ \ t\equiv\pm(c+10k) \bmod 50$ with $0\leq t \leq 25$
    \end{tabular}
    \caption{The $37$ simple modules of $V_I$, labeled by the corresponding simple objects
    $[\gamma,\chi]$ of $Z(\mathrm{Vec}_{2.A_5}^{\tilde\omega})$, with $[1,\chi_1]=V_I$ itself.
    Here $\zeta_5=e^{2\pi \ri/5}$, $\zeta_3=e^{2\pi \ri/3}$, $s\in\{\pm\}$ (identified with $\pm1$
    in formulas), $i\in\{0,1\}$, $b\in\{1,2\}$, $j\in\{0,1,2\}$, $c\in\{1,2,3,4\}$,
    $k\in\{0,1,2,3,4\}$. }\label{tab:VIreps}
\end{center}
\end{table}

\newpage

\begin{landscape}

\begin{table}
\begin{center}
    {\small
    \begin{tabular}{c|cccccccc}
         & $[1,\chi_{1,i}]$ & $[1,\chi'_{2}]$ & $[1,\chi_{3,i}]$ & $[z,\chi_{2,i}]$ & $[z,\chi_{4}]$ & $[\gamma_4,\psi_4^{2j}]$ & $[\gamma_6^{b},\psi_6^{2k}]$ & $[\gamma_8^{c},\psi_8^{2l}]$ \\\midrule
        $[1,\chi_{1,i'}]$ & $1$ & $2$ & $3$ & $2$ & $4$ & $12(-1)^{i'}$ & $8$ & $6(-1)^{i'c}$ \\
        $[1,\chi'_{2}]$ &  & $4$ & $6$ & $4$ & $8$ & $0$ & $-8$ & $6(1+(-1)^{c})$ \\
        $[1,\chi_{3,i'}]$ &  &  & $9$ & $6$ & $12$ & $12(-1)^{i'+1}$ & $0$ & $6(-1)^{i'c+c+1}$ \\
        $[z,\chi_{2,i'}]$ &  &  &  & $-4$ & $-8$ & $0$ & $8(-1)^{b+1}$ & $6\sqrt{2}\,(-1)^{i'}(2-c)$ \\
        $[z,\chi_{4}]$ &  &  &  &  & $-16$ & $0$ & $8(-1)^{b}$ & $0$ \\
        $[\gamma_4,\psi_4^{2j'}]$ &  &  &  &  &  & $12\sqrt{2}\,(-1)^{j+j'}$ & $0$ & $0$ \\
        $[\gamma_6^{b'},\psi_6^{2k'}]$ &  &  &  &  &  &  & $\alpha(bb'+6bk'+6b'k)$ & $0$ \\
        $[\gamma_8^{c'},\psi_8^{2l'}]$ &  &  &  &  &  &  &  & $\beta(cc'+8cl'+8c'l)$
    \end{tabular}}
    \caption{The modular S-matrix of $V_O$ up to the overall normalization
    $\frac{1}{24\sqrt{2}}$. Here $\alpha(x):=16\cos(\pi x/9)$, $\beta(x):=12\cos(\pi x/16)$, and
    the indices run over $i,i'\in\{0,1\}$, $j,j'\in\{0,1\}$, $b,b'\in\{1,2\}$,
    $k,k'\in\{0,1,2\}$, $c,c'\in\{1,2,3\}$, $l,l'\in\{0,1,2,3\}$.}\label{tab:VOSmatrix}
\end{center}
\end{table}

\begin{table}
\begin{center}
    {\small
    \begin{tabular}{c|ccccccccc}
         & $[1,\chi_{1}]$ & $[z,\chi_{2,s}]$ & $[1,\chi_{3,s}]$ & $[z^{i},\chi_{4,(-1)^i}]$ & $[1,\chi_{5}]$ & $[z,\chi_{6}]$ & $[\gamma_4,\psi_4^{2i}]$ & $[\gamma_6^{b},\psi_6^{2j}]$ & $[\gamma_{10}^{c},\psi_{10}^{2k}]$ \\\midrule
        $[1,\chi_{1}]$ & $1$ & $2$ & $3$ & $4$ & $5$ & $6$ & $30$ & $20$ & $12$ \\
        $[z,\chi_{2,s'}]$ &  & $-4$ & $6$ & $8(-1)^{i}$ & $10$ & $-12$ & $0$ & $20(-1)^{b+1}$ & $12(-1)^{c+1}\varphi_{u}$ \\
        $[1,\chi_{3,s'}]$ &  &  & $9$ & $12$ & $15$ & $18$ & $-30$ & $0$ & $12\,\varphi_{u}$ \\
        $[z^{i'},\chi_{4,(-1)^{i'}}]$ &  &  &  & $16(-1)^{ii'}$ & $20$ & $24(-1)^{i'}$ & $0$ & $20(-1)^{i'b}$ & $-12(-1)^{i'c}$ \\
        $[1,\chi_{5}]$ &  &  &  &  & $25$ & $30$ & $30$ & $-20$ & $0$ \\
        $[z,\chi_{6}]$ &  &  &  &  &  & $-36$ & $0$ & $0$ & $12(-1)^{c}$ \\
        $[\gamma_4,\psi_4^{2i'}]$ &  &  &  &  &  &  & $30\sqrt{2}\,(-1)^{i+i'}$ & $0$ & $0$ \\
        $[\gamma_6^{b'},\psi_6^{2j'}]$ &  &  &  &  &  &  &  & $\alpha(bb'+6bj'+6b'j)$ & $0$ \\
        $[\gamma_{10}^{c'},\psi_{10}^{2k'}]$ &  &  &  &  &  &  &  &  & $\beta(cc'+10ck'+10c'k)$
    \end{tabular}}
    \caption{The modular S-matrix of $V_I$ up to the overall normalization
    $\frac{1}{60\sqrt{2}}$. Here $\alpha(x):=40\cos(\pi x/9)$, $\beta(x):=24\cos(\pi x/25)$,
    $\varphi_{\pm}:=\frac{1\pm\sqrt5}{2}$, and $\varphi_u$ is evaluated at
    $u:=(-1)^{c(c-1)/2}\,s'$. The indices run over $s,s'\in\{\pm\}$, $i,i'\in\{0,1\}$,
    $b,b'\in\{1,2\}$, $j,j'\in\{0,1,2\}$, $c,c'\in\{1,2,3,4\}$, $k,k'\in\{0,1,2,3,4\}$.}
    \label{tab:VISmatrix}
\end{center}
\end{table}

\end{landscape}

\section{The modular data of cleft orbifolds like \texorpdfstring{$V_T$}{VT}, \texorpdfstring{$V_O$}{VO}, and \texorpdfstring{$V_I$}{VI}}\label{app:exceptional}

There are subtleties in computing modular data for the exceptional VOAs $V_T$, $V_O$, and $V_I$, in part because their characters are linearly dependent, and errors have appeared in the literature.
In this appendix, we provide our derivation. Our method works for any cleft orbifold, as we explain in Appendix \ref{subsec:cleft}. We then apply it to the tetrahedral VOA $V_T$. The same argument works for $V_O$ and $V_I$, and we therefore content ourselves with simply recording the final answers in Appendix \ref{app:bestiary}. For further treatment of the exceptional VOAs in the literature, we refer to \cite{Ginsparg:1987eb,Dijkgraaf:1989hb,Cappelli:2002wq,dong2013representations,wu2016representations,dong2015fusion,liao2026fusion}.

\subsection{Cleft orbifolds}\label{subsec:cleft}

Any automorphism $g$ of a VOA $V$ permutes the simple $V$-modules, $M\mapsto M^g$ (see the discussion around Equation \eqref{eqn:Mg}). We call this map from a group $G$ of automorphisms of $V$ to the group of permutations of the simples of $V$ the \textit{module-map}. It is a group homomorphism, and every permutation in its image is a symmetry of the fusion rules and modular data of $V$.

When this module-map is trivial, i.e.\ when $M^g\cong M$ for each automorphism $g\in G$ and each simple $V$-module $M$, we call $G$ \textit{cleft}. 
For example, when $V$ has only two simple modules (as is the case for $\widehat{\mathfrak{su}}(2)_1\cong V_2$), any group $G$ of $V$-automorphisms will necessarily be cleft.

Cleft automorphism groups are studied in Section 4 of \cite{Gannon:2024tcl} in the special case that all simple $V$-modules are simple currents (e.g.\ when $V$ is a lattice VOA). We review the theory here. In this case $V$ has a fusion ring which is isomorphic to the group ring of a finite abelian group $A$, so
\begin{align}
    \Rep(V)\cong \mathrm{Vec}_A^{\sigma,\omega}
\end{align}
as a braided fusion category, 
where $(\sigma,\omega)\in Z^3_{\mathrm{ab}}(A,U(1))$ is an Abelian 3-cocycle (see e.g.\ Example 2.2 of \cite{Galindo:2024qzg} for a concise explanation of Abelian cohomology and pointed braided fusion categories). In particular, the associators (i.e.\ F-symbols) of $\Rep(V)$ are characterized by
$[\omega]\in H^3(A,\mathbb{C}^\times)$. 

When $G$ is cleft and $\Rep(V)$ is pointed, the category of twisted modules is given as a fusion category by 
\begin{align}
    \Rep_G(V)\cong \mathrm{Vec}_\Gamma^{\tilde\omega},
\end{align}
where $\Gamma$ is a central extension of $G$ by $A$,
\begin{align}
    1\rightarrow A\hookrightarrow \Gamma\xrightarrow{q} G\twoheadrightarrow 1,
\end{align}
and $[\tilde{\omega}]\in H^3(\Gamma,\mathbb{C}^\times)$ restricts to $[\omega]$ on $A$. Call $M_\gamma$ the simple twisted modules of $V$, where $\gamma\in \Gamma$. For each $g\in G$, there are precisely $|A|$ inequivalent $g$-twisted simple $V$-modules, and these are parametrized by those $\gamma\in\Gamma$ with  $q(\gamma)=g$.

Which central extension $\Gamma$ to take is relatively straightforward in practice. Thanks to Eilenberg-MacLane, we know that the possibilities are parametrized by $H^2(G,A)$, where $G$ acts trivially on $A$. The most fundamental role of $\Gamma$ is in determining the fusions of twisted modules: 
\begin{align}
    M_\gamma\boxtimes M_{\gamma'}=M_{\gamma\gamma'}, \ \ \ \ \gamma,\gamma'\in \Gamma.
\end{align}

Determining the 3-cocycle $\tilde{\omega}$ is much more delicate. Its main role for us is in constraining the modular data of the fixed-point subVOA $V^G$. Indeed, we recall the standard fact that 
\begin{align}
    Z(\Rep_G(V))\cong \Rep(V^G)\boxtimes \overline{\Rep(V)},
\end{align}
where $\overline{\Rep(V)}$ is the braid- and twist-reverse of $\Rep(V)$. When $G$ is cleft and $\Rep(V)$ is pointed, this implies that
\begin{align}
    \Rep(V^G)\subset Z(\mathrm{Vec}_\Gamma^{\tilde \omega}).
\end{align} 
Thus, the modular data of $V^G$ is simply the restriction of the modular data of $Z(\mathrm{Vec}_\Gamma^{\tilde{\omega}})$ (which is determined by $\Gamma$ and $\tilde{\omega}$) to its $\Rep(V^G)$ subcategory.  (Though of course, the unitary $S$ matrix must be rescaled by $\sqrt{|A|}$ to account for the missing rows, and the $T$ matrix has to be multiplied by $e^{-2\pi \ri c/24}$ where $c$ is the central charge of $V$.)

More explicitly, the subcategory $\Rep(V^G)$ is obtained by first identifying the copy of $\overline{\Rep(V)}$ in $Z(\mathrm{Vec}_\Gamma^{\tilde{\omega}})$ and then taking its centralizer. Recall that simple objects of $Z(\mathrm{Vec}_\Gamma^{\tilde\omega})$ are given by pairs $[\gamma,\tilde\chi]$ where $\gamma$ runs through a set of conjugacy class representatives of $\Gamma$, and $\tilde\chi$ is the character of a simple projective representation of the centralizer $C_\Gamma(\gamma)$, where the 2-cocycle $c_\gamma$ defining the class of projective representations is determined from the 3-cocycle $\tilde{\omega}$ as in Equation (6.32) of \cite{Dijkgraaf:1989pz}. Because $\Rep(V)$ is pointed, every simple object of  $\overline{\Rep(V)}$ will be a simple current in  $Z(\mathrm{Vec}_\Gamma^{\tilde{\omega}})$. In fact, these simples will take the form $[z,\psi_z]$, where $z$ runs over elements of $A\subset \Gamma$, and $\psi_z$ is some 1-dimensional (projective) character of $\Gamma$, one for each $z$.

The centralizer of a simple current $J$ in a modular fusion category  $\mathcal{C}$ consists of all simples $x\in\mathcal{C}$ with $Q_J(x)=1$, where 
\begin{align}\label{eqn:Qcharge}
    Q_J(x)=S_{Jx}/S_{0x}
\end{align}
is a phase  sometimes called the charge of $x$ with respect to $J$. Also, $0$ is the vacuum and $S$ is the modular S-matrix of $\mathcal{C}$. 
    The simple objects of $\Rep(V^G)$ can be identified as 
\begin{align}
    \mathrm{Irr}(\Rep(V^G))=\{[\gamma,\tilde\chi]\in \mathrm{Irr}(Z(\mathrm{Vec}_\Gamma^{\tilde\omega}))\mid Q_{[z,\psi_z]}(\gamma,\tilde\chi)=1 \ \text{for all } z\in A \}.
\end{align}

By Proposition 4.2 of \cite{Gannon:2024tcl},
the simple $V^G$-module  $N_{[\gamma,\tilde\chi]}$ labeled by $[\gamma,\tilde\chi]$ appears in the decomposition of the twisted $V$-module $M_\gamma$ into (projective) representations of $C_\Gamma(\gamma)\times V^G$, 
\begin{align}
    M_\gamma \cong \bigoplus_{\substack{\tilde\chi\in \mathrm{Irr}_{c_\gamma}(C_\Gamma(\gamma)) \\ Q_{[z,\psi_z]}(\gamma,\tilde\chi)=1}}R_{\tilde\chi}\otimes N_{[\gamma,\tilde\chi]},
\end{align}
where $R_{\tilde\chi}$ is the $c_\gamma$-projective representation of $C_\Gamma(\gamma)$ with character $\tilde\chi$. Equivalently, using the perspective of \cite{Gannon:2026ttf}, one could analogously decompose the extended Hilbert space $\bigoplus_{\gamma\in \Gamma}M_\gamma$ into irreducible representations of the appropriate dome algebra \cite{Green:2023ork} (which in this case is a Mason-Ng algebra \cite{Mason:2014kea}).

Now, the modular data of any $Z(\mathrm{Vec}_\Gamma^{\tilde{\omega}})$ is computed in Equations (5.23) and (5.24) of \cite{Coste:2000tq}.
A special case (introduced in Equation (6.45) of \cite{Dijkgraaf:1989pz}) is when all the 2-cocycles $c_\gamma$ mentioned earlier are cohomologically trivial (i.e.\ coboundary). In particular, for each $\gamma\in \Gamma$, there are phases $\varepsilon_\gamma(g)$ with $g\in C_\Gamma(\gamma)$ such that 
\begin{align}\label{vareps}
c_\gamma(g,h)=\varepsilon_\gamma(g)\,\varepsilon_\gamma(h)\,\overline{\varepsilon_\gamma(gh)}, \ \ \ \ \text{for all }g,h\in C_\Gamma(\gamma).
\end{align}
This special case was called CT in \cite{Coste:2000tq}. In this case the projective characters $\tilde{\chi}$ become true characters $\chi=\overline{\varepsilon_\gamma}~\tilde{\chi}$ of the centralizer $C_\Gamma(\gamma)$ when rescaled by $\overline{\varepsilon_\gamma}$. Thus, we can and will  parametrize the simple objects of $Z(\mathrm{Vec}_\Gamma^{\tilde{\omega}})$  by pairs $[\gamma,\chi]$ where now $\chi$ is a true (i.e.\ not projective)  $C_\Gamma(\gamma)$-character. All the $Z(\mathrm{Vec}_\Gamma^{\tilde{\omega}})$ which are relevant in what follows are CT. 

The $\varepsilon_g(h)$ in \eqref{vareps} are uniquely determined up to a 1-dimensional representation of $C_\Gamma(g)$, but this ambiguity just amounts to a shuffling of the identification of $\mathrm{Irr}_{c_\gamma}(C_\Gamma(\gamma))$ with $\mathrm{Irr}(C_\Gamma(\gamma))$. We will exploit this ambiguity shortly. We can and will require 
\begin{align}
\begin{split}
    &\hspace{.7in}\varepsilon_e(\gamma)=1=\varepsilon_\gamma(e), \ \ \ \ \gamma\in \Gamma \\
    &\varepsilon_{g^{-1}\gamma g}(g^{-1}hg)=\varepsilon_\gamma(h)\,c_\gamma(g,g^{-1}hg)\,\overline{c_\gamma(h,g)}, \ \ \ \ h\in C_\Gamma(\gamma), \ g\in \Gamma.
\end{split}
\end{align} 
The phases 
 $\varepsilon_g(h)$ can be constrained by computing them for the restriction of the cocycle to the subgroup $\langle g,h\rangle$; those $\varepsilon_g(h)$ for $\Gamma$ will be a subset of those phases for the subgroup $\langle g,h\rangle$.

 In the CT case the modular data  of $Z(\mathrm{Vec}_\Gamma^{\tilde{\omega}})$ simplifies and we obtain (see e.g.\ Equations (6.48) and (6.49) of \cite{Dijkgraaf:1989pz})
\begin{align}\label{eqn:STct}
\begin{split}
    \overline{S_{[\gamma,\chi],[\gamma',\chi']}}&=\frac{1}{|C_\Gamma(\gamma)|\,|C_\Gamma(\gamma')|}\sum_{g\in \Gamma(\gamma,\gamma')}{\chi(g\gamma'g^{-1})}\,{\chi'(g^{-1}\gamma g)}\,{\varepsilon_\gamma(g\gamma' g^{-1})}\,{\varepsilon_{g\gamma' g^{-1}}(\gamma)},\\
    T_{[\gamma,\chi],[\gamma,\chi]}&=\frac{\chi(\gamma)}{\chi(e)}\varepsilon_\gamma(\gamma),
\end{split}
\end{align}
where $\Gamma(\gamma,\gamma')=\{g\in\Gamma\,|\,\gamma g\gamma'g^{-1}=g\gamma'g^{-1}\gamma\}$ and we have reported the complex conjugate of the S-matrix for convenience. A useful special case (see e.g.\ Equation (2.8) of \cite{Coste:2000tq}) is
\begin{equation}
\overline{{S}_{[z,\chi],[\gamma',\chi']}}= {\varepsilon_{z}(\gamma')}\,{\varepsilon_{\gamma'}(z)} \,{\chi(\gamma')}\,\frac{{\chi'(z)}}{|C_{\Gamma}(\gamma')|}\label{Sz}\end{equation}
when $z$ lies in the center of $\Gamma$. From this we can read off the phases $Q_{[z,\psi_z]}([\gamma,\chi])$ of \eqref{eqn:Qcharge}.

For a simple example, any lattice VOA containment $V_L\subset V_M$, where $L\subset M$ have the same dimension, is cleft, with Abelian group $G=L^*/M^*$ acting on $V_M$, fixing $V_L$. Here, $A=M^*/M$, and the central extension $\Gamma$ is $L^*/M$. The corresponding $Z(\mathrm{Vec}_\Gamma^{\tilde{\omega}})$ is CT. 
  The phases $\varepsilon$ 
are computed by first fixing a coset representative $\lambda\in L^*$ for every class in $L^*/M$. Then $\varepsilon_{[\lambda]}([\lambda'])=e^{\pi\mathrm{i}\lambda\cdot\lambda'}$ works. We will use this example shortly.

In the cleft and CT setting, it is trivial to identify $\overline{\Rep(V )}$ in $Z(\mathrm{Vec}^{\tilde{\omega}}_\Gamma)$.
 Recall that the $\varepsilon_z(g)$ are uniquely determined up to a 1-dimensional representation $\psi$ of $C_\Gamma(z) = \Gamma$. We can consume this freedom by fixing each $\psi_z$ to be 1: without loss
of generality the simples of $\overline{\Rep(V )}$ can then be identified with $[z, 1]$ for each $z \in A \le \Gamma$.

In the following subsections, we apply this cleft machinery to compute the modular data of the exceptional $c=1$ VOA $V_T$.

\subsection{Generalities on the exceptionals}

First, we recall that the three exceptional finite subgroups of $SO(3)$ are 
\begin{align}\label{eqn:Gs}
    A_4, \ \ S_4, \ \ A_5.
\end{align}
Their preimages under the natural map $SU(2)\to SU(2)/\mathbb{Z}_2\cong SO(3)$ give rise to the following three exceptional finite subgroups of $SU(2)$,
\begin{align}\label{eqn:binarygroups}
    2.A_4\cong \textsl{SL}(2,\mathbb{Z}_3), \ \ \ 2.S_4, \ \ \ 2.A_5\cong \textsl{SL}(2,\mathbb{Z}_5).
\end{align}
For each of the ADE subgroups $\Gamma\le SU(2)$, we note that 
\begin{align}
    H^2(\Gamma,\mathbb{C}^\times)=0, \ \ \ \ H^3(\Gamma,\mathbb{C}^\times)\cong\mathbb{Z}/|\Gamma|.
\end{align}

Given any finite group $\Gamma$ and subgroup $H$, the restriction of $H^k(\Gamma,\mathbb{C}^\times)$ to $H^k(H,\mathbb{C}^\times)$ followed by corestriction back to $\Gamma$ is multiplication of the cocycle  by the index $|\Gamma|/|H|$ (see e.g. Proposition III.9.5 (ii) in \cite{brown1982cohomology}). This implies that the restriction 
\begin{align}
    H^k(\Gamma,\mathbb{C}^\times)\to \prod_p H^k(\Gamma_p,\mathbb{C}^\times)
\end{align}
is injective,
where the direct product is over all distinct primes dividing $|\Gamma|$, and  $\Gamma_p$ is a Sylow  $p$-subgroup of $\Gamma$. In the case of the ADE groups (and their subgroups), because $H^3(\Gamma,\mathbb{C}^\times)\cong\mathbb{Z}/|\Gamma|$, this restriction will in fact be an isomorphism. So we can identify $\tilde{\omega}$ on $\Gamma$ with its restrictions to the $\Gamma_p$. These  restrictions to $\Gamma_p$ can be obtained by the orbifolds $V^{G_p}$ where $G_p=(\Gamma/A)_p$.

Let $\Gamma$ be one of the groups in \eqref{eqn:binarygroups}. As explained in the previous subsection, we take the copy of $\overline{\Rep(V_2)}$ in $Z(\mathrm{Vec}^{\tilde{\omega}}_\Gamma)$ to be the vacuum $[1,1]$  together with the  order-2 simple current $[z,1]$, where  $z$ is the nontrivial central element in $\Gamma$. This means that the simple modules of $V_{2}^G$ for $G$ in \eqref{eqn:Gs} will be those $[\gamma,\chi]$ in $Z(\mathrm{Vec}^{\tilde{\omega}}_\Gamma)$ with charge $Q_{[z,1]}([\gamma,\chi])=1$.

\subsection{Tetrahedral}

Now focus on the tetrahedral VOA $V_T$, i.e.\ the case that $G = A_4$, with central extension $\Gamma = \textsl{SL}(2, \mathbb{Z}_3)$. The octahedral and icosahedral VOAs $V_O$ and $V_I$ can be treated analogously.

The group $\Gamma$ has precisely 7 conjugacy classes: $\{1\}$ and $\{z= -1\}$ (both with centralizer $\Gamma$), two conjugacy classes $\mathbf{3}_\pm$ of order 3 elements (both with centralizer $\mathbb{Z}_6$), $\mathbf{4}$ containing the order 4 elements (with centralizer $\mathbb{Z}_4$), and two ($\mathbf{6_\pm}$) containing order 6 elements (both with centralizer $\mathbb{Z}_6$). The character table of $\Gamma$ is recorded in Table \ref{tab:SL(2,Z3)}. 

Note that the Schur multipliers of these centralizers are all trivial. Therefore, for any 3-cocycle $\tilde{\omega}$,  the twisted double $Z(\mathrm{Vec}^{\tilde{\omega}}_{\Gamma})$ is CT.  Any  $Z(\mathrm{Vec}^{\tilde{\omega}}_{\Gamma})$ will therefore have exactly $7+7+6+6+4+6+6=42$ simples,
 and hence $V_2^{A_4}$ will have exactly $42/2=21$ simple modules, 
 \begin{align}
     \mathrm{rk}(Z(\mathrm{Vec}_{2.A_4}^{\tilde\omega}))=42, \ \ \ \ \mathrm{rk}(\Rep(V_2^{A_4}))=21.
 \end{align}

Fix any elements $\gamma_6$ and $\gamma_4$ from the conjugacy classes $\mathbf{6_+}$ and $\mathbf{4}$, respectively. We see from the character table that $\gamma_6^{\pm2}$ and $\gamma_6^{-1}$ lie in the conjugacy classes $\mathbf{3_{\mp}}$ and $\mathbf{6_-}$, respectively. We use $\psi_6$ and $\psi_4$ to denote the 1-dimensional representations of $\langle\gamma_6\rangle\cong \mathbb{Z}_6$ and $\langle\gamma_4\rangle\cong\mathbb{Z}_4$, respectively, satisfying $\psi_6(\gamma_6)=\xi_6$ and $\psi_4(\gamma_4)=\xi_4$.

\begin{table}
 \begin{center}
\begin{tabular}{l|cccc}
\toprule
class    & $\{z^a\}$ & $\mathbf{3}_{\pm}$ & $\mathbf{4}$ & $\mathbf{6}_{\pm}$ \\
$C(g)$   & $\Gamma$  & $\mathbb{Z}_6$     & $\mathbb{Z}_4$ & $\mathbb{Z}_6$   \\
\midrule
$\chi_{1,i}$ & $1$        & $\xi_3^{\pm i}$  & $1$  & $\xi_3^{\pm i}$ \\
$\chi_{2,i}$ & $(-1)^a 2$ & $-\xi_3^{\mp i}$ & $0$  & $\xi_3^{\mp i}$ \\
$\chi_{3}$   & $3$        & $0$              & $-1$ & $0$             \\
\bottomrule
\end{tabular}
\caption{The character table of $\textsl{SL}(2,\mathbb{Z}_3)$. Here, $a\in\{0,1\}$, $i\in \{0,1,2\}$, and $\xi_n=e^{2\pi \mathrm{i}/n}$.}\label{tab:SL(2,Z3)}
\end{center}
\end{table}

 We need to identify the 3-cocycle $\tilde{\omega}$ relevant to $V_2^{A_4}$. The only information we need regarding it (at least as far as determining the modular data is concerned) are the phases $\varepsilon_g(h)$ appearing in \eqref{vareps},  when $g,h$ commute. 
Recall that
$\tilde{\omega}$  can be identified with its restrictions to 
the Sylow subgroups $Q_8$ and $\mathbb{Z}_3$ of $\Gamma$.  The former is the 3-cocycle for the orbifold $V_2^{\mathbb{Z}_2\times\mathbb{Z}_2}=V_8^+$ while the latter is that for $V_2^{\mathbb{Z}_3}=V_{18}$.

For example let's compute the phases $\varepsilon_{\gamma_4^k}(\gamma_4^j)$. As always, we take $\varepsilon_e(g)=\varepsilon_g(e)=1$. 
  To compute the others, consider the $\mathbb{Z}_2$ orbifold of $V_2$ by the automorphism corresponding to the projection $\bar{\gamma_4}$  of $\gamma_4$ into $G=A_4$, i.e.\ the VOA $V_2^{\bar{\gamma_4}}=V_{8}$. So we have here a containment of lattice VOAs. Copying the $\varepsilon$ from the example at the end of Appendix \ref{subsec:cleft}, we obtain $\varepsilon_{\gamma_4^k}(\gamma_4^j)=\exp(\pi\ri \frac{k}{\sqrt{8}}\frac{j}{\sqrt{8}})=\xi_{16}^{jk}$. However, for that $\mathbb{Z}_2$ orbifold, which
as we know sits inside the twisted double of $\mathbb{Z}_4$, those $\varepsilon_{\gamma_4^k}(\gamma_4^j)$  could be multiplied by any 1-dimensional representation $\psi_4^m$
 (where $m$ depends only on $k$). We’re interested in those $\varepsilon_{\gamma_4^k}(\gamma_4^j)$
for the full $A_4$ orbifold, and these are ambiguous up to a 1-dimensional representation of the
centralizer $C_\Gamma(\gamma_4^k)$.  But for $k$ odd that centralizer is also $\mathbb{Z}_4$, so $m=0$ works. For $k=2$, i.e.\ for $\varepsilon_z(\gamma_4^j)$, that $\psi_4^m$ is more elusive. We fix the value of $m$ by requiring that there be a module of
$V_2^{A_4}$  with conformal weight $\sfrac{1}{16}$, corresponding to the restriction of the $\gamma_4$-twisted $V_2$-module  down to
$V_2^{\bar{\gamma_4}}= V_8$ and thence to $V_2^{A_4}$. We find $m=-1$.
  
  The other $\varepsilon$ can be computed similarly. Once these are known, the modular data for $Z(\mathrm{Vec}^{\tilde{\omega}}_{\Gamma})$ is obtained directly from the character tables and \eqref{eqn:STct}. The S-matrix for the double is given in Table \ref{tab:doublebinarytetS}. In this table, $i,i'\in \{0,1,2\}$, $j,j'\in\{0,1,2,3\}$, $k\in\{0,1,\dots,5\}$, $a,a'\in\{0,1\}$, and $b,b'\in\{\pm1,\pm2\}$. From this, the S-matrix of $V_T$ can be read off, resulting in Table \ref{tab:VTSmatrix}. The values of $e^{2\pi\mathrm{i}h_x}$ for the simples, in the same order they appear in Table \ref{tab:doublebinarytetS}, are $(-\mathrm{i})^a$, $\mathrm{i}^a$, $(-\mathrm{i})^a$, $\xi_{16}^{1+4j}$, and $\xi_{36}^{b^2+6bk}$, respectively.

 \begin{table}
\centering
\setlength{\tabcolsep}{4pt}
\begin{tabular}{c|ccccc}
\toprule

  & $[z^a,\chi_{1,i}]$ & $[z^a,\chi_{2,i}]$ & $[z^a,\chi_{3}]$
  & $[\gamma_4,\psi_4^j]$ & $[\gamma_6^b,\psi_6^k]$ \\
\midrule
$[z^{a'},\chi_{1,i'}]$
  & $(-1)^{aa'}$ & $2(-1)^{aa'+a'}$ & $3(-1)^{aa'}$
  & $6(-1)^{ja'}$ & $4(-1)^{a'k}\xi_3^{-bi'}$ \\
$[z^{a'},\chi_{2,i'}]$
  & & $4(-1)^{aa'+a+a'}$ & $6(-1)^{aa'+a}$
  & $0$ & $4(-1)^{b+1+a'k}\xi_3^{bi'}$ \\
$[z^{a'},\chi_{3}]$
  & & & $9(-1)^{aa'}$
  & $6(-1)^{ja'+1}$ & $0$ \\
$[\gamma_4,\psi_4^{j'}]$
  & & & & $6\sqrt{2}\,(-\ri)^{\,(j+j')(j+j'+1)}$ & $0$ \\
$[\gamma_6^{b'},\psi_6^{k'}]$
  & & & & & $4\,\xi_{18}^{-bb'-3kb'-3k'b}$ \\
\bottomrule
\end{tabular}
\caption{The S-matrix for $Z(\textsl{Vec}_{2.A_4}^{\tilde\omega})$, multiplied by $24$.}
\label{tab:doublebinarytetS}
\end{table}

\section{Automorphisms of nice \texorpdfstring{$c=1$}{c=1} vertex operator algebras}\label{app:automorphisms}

In this appendix, we discuss the invertible symmetries of the nice $c=1$ VOAs. Specifically, after reviewing some technical generalities in Section \ref{subsec:generalities}, we proceed to tabulate the automorphism groups of the nice $c=1$ VOAs (extracted from \cite{dong1999automorphism,Dong:1997id,dong1999rank} and summarized in Table \ref{tab:aut}), and calculate how each automorphism permutes modules. This is important both for our classification of nice $c=1$ VOAs and VOSAs in the main text, and our classification of nice  $c=1$ CFTs in \cite{grCFTs}.

In Appendix \ref{app:QuOpTopLin}, we give a complementary treatment of the non-invertible symmetry structure of the nice $c=1$ VOAs, at a physics level of rigor.

\subsection{Generalities}\label{subsec:generalities}

Before specializing to $c=1$, we collect some general facts which will be used in the sequel.

We recall that, if $V^G$ is strongly rational, then $\Rep_G(V)$ is a $G$-crossed ribbon tensor category in the sense of \cite{Turaev:2000ug}. Part of the data of a $G$-crossed ribbon category is a $G$-action, and part of the data of a (right) $G$-action is a tensor equivalence $g_\ast:\Rep_G(V)\to\Rep_G(V)$ (sending $\Rep_h(V)\to \Rep_{g^{-1}hg}(V)$) for each $g\in G$. We denote the image of an $h$-twisted $V$-module $M$ under the action of $g$ as $M^g$. The $g^{-1}hg$-twisted module $M^g$ by definition has the same underlying vector space as $M$, but with the action of $V$ modified as 
\begin{align}\label{eqn:Mg}
    Y_{M^g}(v,z) = Y_{M}(g^{-1}v,z), 
\end{align}
where $v\in V$ and $Y_M(v,z)=\sum_n v_nz^{-n-1}$ with $v_n\in\mathrm{End}(M)$. 

There are several useful constraints on the action of $G$ on twisted modules. For example, a module and its $g$-image always have the same character, $\mathrm{ch}_M(\tau)=\mathrm{ch}_{M^g}(\tau)$. Also, $g_\ast$ restricts to a ribbon tensor equivalence of $\Rep(V)$, called the \emph{module map} in Appendix \ref{app:exceptional}, which in particular means that it is required to respect the modular data of $V$, 
\begin{align}\label{eqn:Aut(S,T)}
    S_{M^g,N^g}=S_{M,N}, \ \ \ \ T_{M^g,N^g}=T_{M,N}, \ \ \ \ \text{ for all }M,N\in \Rep(V).
\end{align}
Also, it follows from the fact that $g_\ast$ is a ribbon tensor equivalence that $(M\oplus N)^g\cong M^g\oplus N^g$.

\begin{table}
\begin{center}
    \begin{tabular}{c|c|c}
    $V$ & $\mathrm{Aut}(V)$ & $\mathrm{Aut}^\dagger(V)$ \\\midrule 
    $V_{2}\cong \widehat{\mathfrak{su}}(2)_1$ & $\textsl{PSL}_2(\mathbb{C})$ & $SO(3)$ \\
    $V_{2m}$ ($m>1$) & $\mathbb{C}^\times \rtimes  \mathbb{Z}_2^{\mathrm{C}}$ & $O(2)$ \\
    $V_{2m}^+$ ($m\neq 1,4$) & $\mathbb{Z}_2$ & $\mathbb{Z}_2$ \\
    $V_{8}^+$ & $S_3$ & $S_3$ \\
    $V_T$ & $\mathbb{Z}_2$ & $\mathbb{Z}_2$ \\
    $V_O$, $V_I$ & $\mathds{1}$ & $\mathds{1}$ 
    \end{tabular}
\end{center}
\caption{For $V$ a $c=1$ VOA, $\mathrm{Aut}(V)$ is the automorphism group and $\mathrm{Aut}^\dagger (V)$ is the subgroup of the automorphism group which preserves the unitary structure on $V$.}\label{tab:aut}
\end{table}

The following general construction will be useful for understanding the automorphisms of $c=1$ VOAs. 

\begin{proposition}\label{prop:normaut}
Suppose that $V$ is a simple vertex operator algebra, and $W=V^F$ is the fixed-point subalgebra of $V$ with respect to a subgroup $F\subset \mathrm{Aut}(V)$. Then there is a well-defined homomorphism    
    \begin{align}\label{eqn:normhom}
    \begin{split}
        N_{\mathrm{Aut}(V)}(F)&\to \mathrm{Aut}(V^F) \\
        g&\mapsto \bar g := g\vert_{V^F},
    \end{split}
    \end{align}
    where $N_{\mathrm{Aut}(V)}(F)$ is the normalizer of $F$ inside $\mathrm{Aut}(V)$. The kernel of \eqref{eqn:normhom} contains $F$. If $N_{\mathrm{Aut}(V)}(F)$ is finite, then the kernel is exactly equal to $F$, and hence $V^F$ admits a faithful action of $N_{\mathrm{Aut}(V)}(F)/F$ by automorphisms.
\end{proposition} 

We remark that it is likely possible to relax the assumption that $N_{\mathrm{Aut}(V)}(F)$ is finite, but we will not need the more general statement. 

\begin{proof} To check that the map $g\mapsto \bar g = g\vert_{V^F}$ is well-defined, one needs only confirm that any element $g\in N_{\mathrm{Aut}(V)}(F)$ maps $V^F$ back to itself. And indeed, by definition, for every $f\in F$, there is some $f'\in F$ such that $fg=gf'$, from which it follows that 
\begin{align}
    fgv = gf'v=gv, \ \ \  \text{ for all }f\in F,v\in V^F,
\end{align}
so $v$ belonging to $V^F$ implies that $gv\in V^F$ for all $g\in N_{\mathrm{Aut}(V)}(F)$. 

Of course, $F\subset N_{\mathrm{Aut}(V)}(F)$ acts trivially on $V^F$, and so is a subset of the kernel of \eqref{eqn:normhom}. On the other hand, if $N_{\mathrm{Aut}(V)}(F)$ is finite, by the quantum Galois theory of \cite{10.1215/S0012-7094-99-09720-X,DONG199992}, the map $H\mapsto V^H$ is a bijection between subgroups of $N_{\mathrm{Aut}(V)}(F)$ and subVOAs of $V$ which contain $V^{N_{\mathrm{Aut}(V)}(F)}$. Thus, if $h$ is any automorphism in $N_{\mathrm{Aut}(V)}(F)$ outside of the subgroup $F$, then the fixed points of $H=\langle F,h\rangle$ are strictly smaller than those of $F$, i.e.\ $V^{\langle F,h\rangle}\subsetneq V^F$. Hence, $h$ must act non-trivially on $V^F$ and is therefore outside of the kernel of \eqref{eqn:normhom}. We conclude that the kernel is exactly $F$.
\end{proof}

It is also useful to note the following constraint on how automorphisms of $V^F$ in the image of \eqref{eqn:normhom} permute around its modules. (The proof is trivial.)
\begin{proposition}\label{prop:fixedptperm}
Let $V$ be a simple VOA and $W=V^F$ be the fixed-point subalgebra of $V$ with respect to a finite subgroup $F\subset \mathrm{Aut}(V)$. Then, for any $M\in \Rep_F(V)$,
\begin{align}
    \mathrm{Res}(M)^{\bar g}= \mathrm{Res}(M^g), \ \ \ \ \text{ for all } g\in N_{\mathrm{Aut}(V)}(F),
\end{align}
where $\bar g\in\mathrm{Aut}(V^F)$ is the image of $g$ under the map \eqref{eqn:normhom}, and $\mathrm{Res}(M)$ denotes the restriction of $M$ to a $V^F$-module.
\end{proposition}

Finally, recall from \cite{dong1996compact} that if $F$ is compact and $V$ is simple, then $V$ decomposes into irreducible $F\times V^F$-modules as 
\begin{align}
    V=\bigoplus_{\chi \in \Rep(F)}R_\chi\otimes M_\chi =: \bigoplus_{\chi\in \Rep(F)}V^\chi,
\end{align}
where $R_\chi$ is an irreducible representation of $F$ with character $\chi$, the space $M_\chi$ is a simple $V^F$-module labeled by $\chi$, and $V^\chi$ is the $\chi$-isotypical component of $V$. We note the following characterization of how automorphisms in the image of \eqref{eqn:normhom} permute around the $M_\chi$.

\begin{proposition}\label{prop:permVsubmodules}
 Suppose $V$ is a simple VOA and $F$ is a compact subgroup of $\mathrm{Aut}(V)$. Then for any $g\in N_{\mathrm{Aut}(V)}(F)$, the corresponding automorphism $\bar g \in \mathrm{Aut}(V^F)$ from \eqref{eqn:normhom} permutes the simple $V^F$-modules $M_\chi$ according to
    \begin{align}
        M_\chi^{\bar g} \cong M_{\chi^g},
    \end{align}
    where $\chi^g(f):=\chi(g^{-1}fg)$.
\end{proposition}

\begin{proof}
    Since $g$ normalizes $F$, it sends any $F$-stable subspace $W\subset V$ to another $F$-stable subspace $g(W)$. The action of $F$ on $g(W)$ is 
    \begin{align}
        f(gw)=g(g^{-1}fg)w, \ \ \ \ \text{ for all }f\in F,w\in W.
    \end{align}
    In particular, applying this to the $\chi$-isotypical component $W=V^\chi$ and taking the trace, we see that 
    \begin{align}
        \mathrm{Tr}_{g(V^\chi)}fq^{L_0-c/24}=\mathrm{Tr}_{V^\chi}(g^{-1}fg)q^{L_0-c/24}=\mathrm{Tr}_{V^{\chi^g}}fq^{L_0-c/24}, \ \ \ \ \text{ for all }f\in F.
    \end{align}
    This implies that $g(V^\chi)=V^{\chi^g}$ and in particular that $M_\chi^{\bar g} \cong M_{\chi^g}$.
\end{proof}
We remark that the map $F\to F$ sending $f\mapsto g^{-1} f g$ for some $g\in N_{\mathrm{Aut}(V)}(F)$ is a (not necessarily inner) automorphism of $F$, and $\chi\mapsto \chi^g$ is the permutation of representations of $F$ induced by $g$.

\subsection{Lattice vertex operator algebras}

Let us apply some of this general machinery to the lattice VOAs $V_{2m}$.

For any simple VOA $\widehat{\mathfrak{g}}_k$ of Lie type, with $\mathfrak{g}$ a simple Lie algebra, the automorphism group is given by the automorphism group of the underlying Lie algebra $\mathfrak{g}$,
\begin{align}
    \mathrm{Aut}(\widehat{\mathfrak{g}}_k)=\mathrm{Aut}(\mathfrak{g}).
\end{align}
In the case of $V_{2}\cong\widehat{\mathfrak{su}}(2)_1$, this is
\begin{align}
    \mathrm{Aut}(V_2)\cong \textsl{SL}(2,\mathbb{C})/\mathbb{Z}_2\cong \textsl{PSL}(2,\mathbb{C}).
\end{align}
The subgroup of $\textsl{PSL}(2,\mathbb{C})$ which acts unitarily on the theory is $SO(3)$. 

Every automorphism $g\in \mathrm{Aut}(V_2)$ induces the trivial braided auto-equivalence of $\Rep(V_2)$, simply because $\Rep(V_2)$ does not possess any non-trivial braided auto-equivalences. (For physicists, the semion TQFT, a.k.a.\ $SU(2)_1$ Chern-Simons theory, does not have any non-trivial topological surface defects.) Hence, the automorphisms of $V_2$ do not permute around its modules in a non-trivial way.

The finite subgroups of $\mathrm{Aut}(V_2)$ are known to admit an ADE classification. In particular, they are (up to conjugacy)
\begin{align}
    \mathsf{A}_{n\geq 0}: \mathbb{Z}_{n+1}, \ \ \ \ \  \mathsf{D}_{n\geq 4}: D_{2(n-2)}, \ \ \ \ \  \mathsf{E}_6: A_4, \ \ \ \ \ \mathsf{E}_7: S_4, \ \ \ \ \ \mathsf{E}_8: A_5,
\end{align}
where $D_{2\ell}$ is the dihedral group of order $2\ell$.
The corresponding fixed-point subalgebras are 
\begin{align}
    V_{2}^{\mathbb{Z}_n} \cong V_{2n^2} , \ \ \ V_{2}^{D_{2n}} \cong V^+_{2n^2}, \ \ \ V_{2}^{A_4}=V_T, \ \ \ V_{2}^{S_4} = V_O, \ \ \ V_{2}^{A_5} = V_I.
\end{align}

In the case of $V_{2m}$ with $m>1$, the automorphism group is 
\begin{align}
    \mathrm{Aut}(V_{2m}) \cong \mathbb{C}^\times \rtimes \mathbb{Z}_2^{\mathrm{C}}\cong O(2,\mathbb{C}), \ \ \ \ (m>1).
\end{align}
The automorphisms in the subgroup $\mathbb{C}^\times$ are obtained by exponentiating the zero-modes of the $\mathfrak{u}(1)$ currents, and the non-identity connected component is reached by incorporating $\mathbb{Z}_2^{\mathrm{C}}$ charge conjugation symmetry. The subgroup which acts unitarily is $U(1)\rtimes \mathbb{Z}_2^{\mathrm{C}} \cong O(2)$.

The automorphisms in the subgroup $\mathbb{C}^\times$ are continuously connected to the identity and hence induce the trivial braided auto-equivalence of $\Rep(V_{2m})$. On the other hand, the automorphisms in the non-identity connected component permute around the modules of $V_{2m}$ according to 
\begin{align}
    V_{2m,r}^g\cong  V_{2m,-r}
\end{align}
where we understand $r$ to be defined modulo $2m$. 

For $m>1$, the finite subgroups of $\mathrm{Aut}(V_{2m})$, up to conjugacy, are simply the obvious $\mathbb{Z}_n$ and $D_{2n}$ subgroups of $O(2)$. The fixed-point subalgebras in this case are 
\begin{align}\label{eqn:latticefp}
    V_{2m}^{\mathbb{Z}_n}\cong V_{2mn^2}, \ \ \ \ V_{2m}^{D_{2n}}\cong V^+_{2mn^2}.
\end{align}
Summarizing the discussion above, we record the following concerning the $G$-fixed points of lattice VOAs for later use (cf.\ Example 3.26 of \cite{Moller:2024xtt}).

\begin{proposition}\label{prop:Gfixedpts}
    A VOA $V$ is isomorphic to $V_{2m}$, $V_{2m}^+$, $V_T$, $V_O$, or $V_I$ if and only if $V\cong V_{2t_\star}^G$ for some square-free integer $t_\star$ and some finite subgroup $G\subset \mathrm{Aut}(V_{2t_\star})$.
\end{proposition}

\subsection{Charge conjugation orbifolds}

For charge-conjugation orbifolds, one has that 
\begin{align}
    \mathrm{Aut}(V_{2m}^+) = \begin{cases}
    \mathbb{Z}_2, & m\neq 1,4 \\
    S_3, & m=4.
    \end{cases}
\end{align}
Note that $V_{2}^+\cong V_8$ and so its automorphisms follow from the analysis of the previous subsection. We assume in what follows that $m>1$.

Let us understand the symmetries of $V_{2m}^+$ in the generic case $(m\neq 1,4)$ using Proposition \ref{prop:normaut}. We apply the map \eqref{eqn:normhom} to the case $V=V_{2m}$ and $F=\mathbb{Z}_2^{\mathrm{C}}$, with the understanding that $V_{2m}^+\cong V_{2m}^{\mathbb{Z}_2^{\mathrm{C}}}$. The normalizer group is readily computed to be
\begin{align}
    N_{\mathrm{Aut}(V_{2m})}(\mathbb{Z}_2^{\mathrm{C}})\cong \mathbb{Z}_2\times\mathbb{Z}_2^{\mathrm{C}}\subset O(2),
\end{align} 
where one can think of the first factor as the $\mathbb{Z}_2$ subgroup of the connected component of the identity in $\mathrm{Aut}(V_{2m})$. Thus, $V_{2m}^+$ inherits an action of this  $\mathbb{Z}_2\cong N_{\mathrm{Aut}(V_{2m})}(\mathbb{Z}_2^{\mathrm{C}})/\mathbb{Z}_2^{\mathrm{C}}$. 

Using the elementary fact that the fixed points of $V^F$ with respect to the image of $N_{\mathrm{Aut}(V)}(F)$ under the map \eqref{eqn:normhom} are identical to the fixed points of $V$ with respect to $N_{\mathrm{Aut}(V)}(F)$, it is easy to see that 
\begin{align}
    (V_{2m}^+)^{\mathbb{Z}_2} \cong V_{2m}^{\mathbb{Z}_2\times \mathbb{Z}_2^{\mathrm{C}}}\cong V_{8m}^+ , \ \ \ \ (m\neq 1,4),
\end{align}
where the second isomorphism follows from e.g.\ \eqref{eqn:latticefp}.

Let us determine how this involution permutes the modules of $V_{2m}^+$.

\begin{proposition}\label{prop:V2m+perm}
    The non-trivial order 2 automorphism of $V_{2m}^+$ (for $m\neq 1,4$) swaps the modules
    \begin{align}\label{eqn:m+genericperm}
        V_{2m,m}^+ \leftrightarrow V_{2m,m}^-, \ \ \ \ V_{2m}^{T_1,+}\leftrightarrow V_{2m}^{T_2,+}, \ \ \ \ V_{2m}^{T_1,-}\leftrightarrow V_{2m}^{T_2,-},
    \end{align}
    and stabilizes all others.
\end{proposition}

\begin{proof}
    Let $g$ be the $\mathbb{Z}_2$ symmetry of $V_{2m}$ in the connected component of the identity automorphism, and let $\bar g$ be its image under \eqref{eqn:normhom}, i.e.\ $\bar g$ is the non-trivial involution of $V_{2m}^+$. Recall from earlier that $g$ stabilizes all $V_{2m}$-modules. 

    First, $\bar g$ stabilizes the vacuum $V_{2m}^+$ e.g.\ because ribbon auto-equivalences must fix the tensor unit. On the other hand, Proposition \ref{prop:fixedptperm} tells us that 
    \begin{align}\label{eqn:vacvsm}
        V_{2m}^+\oplus V_{2m}^-\cong \mathrm{Res}(V_{2m})\cong \mathrm{Res}(V_{2m}^g)\cong (V_{2m}^+)^{\bar g}\oplus (V_{2m}^-)^{\bar g}\cong V_{2m}^+\oplus (V_{2m}^-)^{\bar g},
    \end{align} 
    from which it follows that $(V_{2m}^-)^{\bar g}\cong V_{2m}^-$. 
    
    Noting that the simple $V_{2m}$-module $V_{2m,r}$ with $1\leq r\leq m-1$ restricts to a simple $V_{2m}^+$-module, also somewhat abusively labeled $V_{2m,r}$, it follows from Proposition \ref{prop:fixedptperm} that $V_{2m,r}^{\bar g}\cong V_{2m,r}$ for $1\leq r \leq m-1$.

    From a similar calculation to the one in \eqref{eqn:vacvsm}, one obtains that $(V_{2m,m}^+)^{\bar g} \oplus (V_{2m,m}^-)^{\bar g} \cong V_{2m,m}^+\oplus V_{2m,m}^-$, which implies that $V_{2m,m}^\pm$ are either swapped or stabilized by $\bar g$. To settle which actually occurs, let us pick an actual isomorphism $\varphi:V_{2m,m}^{ g} \to V_{2m,m}$ of $V_{2m}$-modules and see what it does to the subspaces $V_{2m,m}^\pm\subset V_{2m,m}$. The map $\varphi$ is required to satisfy
    \begin{align}
        \varphi Y_{V_{2m,m}}(g^{-1}v,z)\varphi^{-1}=Y_{V_{2m,m}}(v,z), \ \ \ \ v\in V_{2m}
    \end{align}
    and is fixed by this condition up to an overall normalization.
    Denoting the states in $V_{2m,m}$ as $|u,\lambda\rangle$ where $\lambda \in \sqrt{2m}\mathbb{Z}+\sqrt{m/2}$ and $u$ denotes the oscillator contributions, we pick $\varphi$ to act as 
    \begin{align}
        \varphi |u,\sqrt{2m}k+\sqrt{m/2}\rangle = (-1)^k |u,\sqrt{2m}k+\sqrt{m/2}\rangle.
    \end{align}
    On the other hand, charge conjugation acts as 
    \begin{align}
    \begin{split}
        C|u,\sqrt{2m}k+\sqrt{m/2}\rangle &=(-1)^{n(u)}|u,-\sqrt{2m}k-\sqrt{m/2}\rangle \\
        &=(-1)^{n(u)}|u,\sqrt{2m}(-k-1)+\sqrt{m/2}\rangle ,
        \end{split}
    \end{align}
    where $n(u)$ simply counts the number of oscillators. 
    It is straightforward to check that 
    \begin{align}
        C\varphi = -\varphi C,
    \end{align}
    which implies that $\varphi$ swaps states in $V_{2m,m}^+$ and $V_{2m,m}^-$. In particular, we conclude that $(V_{2m,m}^{\pm })^{\bar g} \cong V_{2m,m}^\mp$.

    Finally, to determine what $\bar g$ does to the $V_{2m}^+$-modules $V_{2m}^{T_i,\pm}$, it is easiest to first determine what $g$ does to the charge-conjugation twisted $V_{2m}$-modules $V_{2m}^{T_i}$. Since $g$ sends a $C$-twisted module to a $g^{-1}Cg=C$-twisted module again, it follows that $(V_{2m}^{T_i})^g$ is isomorphic either to $V_{2m}^{T_i}$ or to $V_{2m}^{T_{3-i}}$, where $i=1,2$. The $\mathbb{Z}_2$ character $T_i$ governs the sign with which e.g.\ $Y_{V_{2m}^{T_i}}(|\sqrt{2m}\rangle,z)$ acts on $V_{2m}^{T_i}$. On the other hand, 
    \begin{align}
       Y_{(V_{2m}^{T_i})^g}(|\sqrt{2m}\rangle,z) = Y_{V_{2m}^{T_i}}(g^{-1}|\sqrt{2m}\rangle,z)=-Y_{V_{2m}^{T_i}}(|\sqrt{2m}\rangle,z),
    \end{align}
    so this sign is swapped in $(V_{2m}^{T_i})^g$, which implies that $(V_{2m}^{T_i})^g\cong V_{2m}^{T_{3-i}}$.

    Now, invoking Proposition \ref{prop:fixedptperm}, we learn that 
    \begin{align}
        (V_{2m}^{T_i,+})^{\bar g}\oplus (V_{2m}^{T_i,-})^{\bar g}\cong \mathrm{Res}((V_{2m}^{T_i})^g)\cong \mathrm{Res}(V_{2m}^{T_{3-i}})\cong V_{2m}^{T_{3-i},+}\oplus V_{2m}^{T_{3-i},-}. 
    \end{align}
    Using the fact that $(V_{2m}^{T_i,\pm})^{\bar g}$ must have the same character as $V_{2m}^{T_i,\pm}$, it is forced that $(V_{2m}^{T_i,\pm})^{\bar g}\cong V_{2m}^{T_{3-i},\pm}$.
\end{proof}

Moving on to the exceptional theory $V_8^+$, we have that 
\begin{align}
    \mathrm{Aut}(V_8^+)\cong S_3.
\end{align}
We understand this $S_3$ worth of automorphisms by invoking Proposition \ref{prop:normaut}, using the fact that $V_8^+\cong V^{F}$ where $V=V_2$ and $F=\mathbb{Z}_2\times\mathbb{Z}_2$ is the unique (up to conjugacy) subgroup of $\mathrm{Aut}(V_{2})=\textsl{PSL}(2,\mathbb{C})$ isomorphic to $\mathbb{Z}_2\times \mathbb{Z}_2$. The normalizer of this $\mathbb{Z}_2\times \mathbb{Z}_2$ in $\textsl{PSL}(2,\mathbb{C})$ and its quotient by the subgroup $F$ are, respectively,
\begin{align}
    N_{\mathrm{Aut}(V_{2})}(\mathbb{Z}_2\times\mathbb{Z}_2)\cong S_4, \ \ \ \ \ N_{\mathrm{Aut}(V_{2})}(\mathbb{Z}_2\times\mathbb{Z}_2)/(\mathbb{Z}_2\times \mathbb{Z}_2)\cong S_3.
\end{align}
By Proposition \ref{prop:normaut}, this $S_3$ acts faithfully, and by \cite{Dong:1997id} it realizes the entire automorphism group of $V_{8}^+$. The fixed point subalgebras with respect to subgroups of this $S_3$ are straightforward to obtain:
\begin{align}
    (V_{8}^+)^{\mathbb{Z}_2}\cong V_{32}^+ , \ \ \ \ \ \ (V_{8}^+)^{\mathbb{Z}_3} \cong V_T, \ \ \ \ \ \  (V_{8}^+)^{S_3} \cong V_O.
\end{align}

To calculate how $S_3$ permutes the modules of $V_8^+$, we first determine what the symmetries of its modular data are.

\begin{lemma}
    The group of permutations of the simple modules of $V_8^+$ which preserve its modular data in the sense of \eqref{eqn:Aut(S,T)} is $S_3$. This $S_3$ is generated by the transposition in \eqref{eqn:m+genericperm} together with
    \begin{align}
        V_{8}^-\leftrightarrow V_{8,4}^+, \ \ \ V_{8,1}\leftrightarrow V_{8}^{T_1,+}, \ \ \ V_{8,3}\leftrightarrow V_{8}^{T_1,-}.
    \end{align}
\end{lemma}
\begin{proof}
    Any such permutation must preserve the vacuum since it is the unique row which is entrywise positive. By imposing that the T-matrix is preserved and that the vacuum is invariant, we learn that any symmetry of the modular data of $V_{8}^+$ must permute modules within the sets 
\begin{align}
    A=\{V_8^-,V_{8,4}^+,V_{8,4}^-\}, \ \ \ \  B=\{V_{8,1},V_{8}^{T_1,+},V_{8}^{T_2,+}\}, \ \ \ \ C= \{V_{8,3},V_{8}^{T_1,-},V_8^{T_2,-}\}.
\end{align}
A priori, any permutation in $S_3^A\times S_3^B\times S_3^C$ is possible, but imposing that permutations preserve the modular S matrix winnows down the allowed permutations to the diagonal $S_3$ subgroup. 
\end{proof}

Call $\mathrm{Aut}(S_{V_8^+},T_{V_8^+})$ the permutations of the simple modules of $V_8^+$ which preserve its modular data. Because every automorphism of $V_8^+$ induces such a permutation, we have a homomorphism $\mathrm{Aut}(V_8^+)\to \mathrm{Aut}(S_{V_8^+},T_{V_8^+})$. In fact:

\begin{proposition}
    The map $\mathrm{Aut}(V_8^+)\to \mathrm{Aut}(S_{V_8^+},T_{V_8^+})$ is an isomorphism. That is, every permutation of simple $V_8^+$-modules preserving the modular data is uniquely induced by an automorphism of $V_8^+$. 
\end{proposition}

\begin{proof}
    Because $\mathrm{Aut}(V_8^+)$ and $\mathrm{Aut}(S_{V_8^+},T_{V_8^+})$ are both isomorphic to $S_3$ it suffices to show that the map is surjective.

    We study the decomposition of $V_2$ into $(\mathbb{Z}_2\times \mathbb{Z}_2)\times V_8^+$-modules, 
    \begin{align}
    \begin{split}
       \mathrm{Res}(V_2)&\cong \bigoplus_{\chi\in\Rep(\mathbb{Z}_2\times\mathbb{Z}_2)} R_\chi \otimes M_\chi \\
      &\cong  R_{\chi_0}\otimes  V_8^+\oplus R_{\chi_1}\otimes V_{8,4}^+\oplus R_{\chi_2}\otimes V_{8,4}^-\oplus R_{\chi_3}\otimes V_8^-,
      \end{split}
    \end{align}
    where the $R_{\chi_i}$ are the irreducible representations of $F:=\mathbb{Z}_2\times \mathbb{Z}_2$, with $\chi_0$ the trivial representation. Let $\bar g\in S_3\cong \mathrm{Aut}(V_8^+)$ be an automorphism of $V_8^+$, and let $g\in S_4\cong N_{\mathrm{Aut}(V_2)}(F)\subset  \mathrm{Aut}(V_2)$ be an automorphism of $V_2$ whose image under the map \eqref{eqn:normhom} is $\bar g$. By Proposition \ref{prop:permVsubmodules}, we learn that $M_\chi^{\bar g} \cong M_{\chi^g}$, where $\chi^g(f)=\chi(g^{-1}fg)$. It suffices to show that every permutation of $\{\chi_1,\chi_2,\chi_3\}$ takes the form $\chi \mapsto \chi^g$ for some $g\in S_4$.
    
    It is an elementary fact about $F=\mathbb{Z}_2\times\mathbb{Z}_2$ that each of its automorphisms $F\to F$ is obtained by embedding $\mathbb{Z}_2\times\mathbb{Z}_2\subset S_4$ and taking $f\mapsto g^{-1}fg$ with $g\in S_4$. Of course, if $g$ is taken to be in the subgroup $F$, then $g^{-1}fg=f$ for all $f\in F$, so the automorphism of $\mathbb{Z}_2\times\mathbb{Z}_2$ defined by $g$ depends only on its image $\bar g \in S_3 \cong S_4/\mathbb{Z}_2\times\mathbb{Z}_2$ under the natural quotient map. The $S_3$ worth of automorphisms of $\mathbb{Z}_2\times\mathbb{Z}_2$ act as permutations of the three irreducible representations $\{\chi_1,\chi_2,\chi_3\}$, so every element of $\mathrm{Aut}(S_{V_8^+},T_{V_8^+})$ is induced by some element of $\mathrm{Aut}(V_8^+)$. 
\end{proof}

\subsection{Exceptional models}

Finally, for the exceptional $c=1$ VOAs, one has that
\begin{align}
    \mathrm{Aut}(V_T)=\mathbb{Z}_2, \ \ \mathrm{Aut}(V_O)=\mathds{1}, \ \ \mathrm{Aut}(V_I)=\mathds{1}.
\end{align}
The $\mathbb{Z}_2$ which acts on $V_T$ can be understood as being inherited from $V_{2}$, using the fact that $V_T$ is the fixed-point subalgebra of $V_{2}$ with respect to $F=A_4\subset \mathrm{Aut}(V_{2})$. In this case, 
\begin{align}
    N_{\mathrm{Aut}(V_{2})}(A_4) =S_4, \ \ \ \ N_{\mathrm{Aut}(V_{2})}(A_4)/A_4 \cong \mathbb{Z}_2,
\end{align}
so Proposition \ref{prop:normaut} tells us that there is a homomorphism $\mathbb{Z}_2\to \mathrm{Aut}(V_T)$ which does not have a kernel. The fixed point subalgebra with respect to this $\mathbb{Z}_2$ is 
\begin{align}
    V_T^{\mathbb{Z}_2}\cong V_O.
\end{align}

In a moment, we will determine how this $\mathbb{Z}_2$ permutes the modules of $V_T$. We establish the following intermediate result first. (See Table \ref{tab:VTreps} for our labeling conventions for the modules of $V_T$, and Table \ref{tab:VTSmatrix} for the modular S-matrix.)

\begin{lemma}\label{lem:VTperm}
    The only non-trivial permutation of the simple modules of $V_T$ which preserves the modular data in the sense of \eqref{eqn:Aut(S,T)} is $M\mapsto M^\ast$, i.e.\ 
   \begin{align}\label{eqn:VTperm}
       [1,\chi_{1,1}]\leftrightarrow [1,\chi_{1,2}], \ \ \ \ [z,\chi_{2,1}]\leftrightarrow [z,\chi_{2,2}], \ \ \ \ [\gamma_6^{+b},\psi_6^{2m}]\leftrightarrow [\gamma_6^{-b},\psi_6^{-2m}],
   \end{align}
   where $b\in\{1,2\}$ and $m\in \mathbb{Z}_3$.
\end{lemma}

\begin{proof}
    Call $g$ a permutation preserving the modular data of $V_T$. First, note that by imposing that the T-matrix and the quantum dimensions are invariant, one learns that the symmetry can only move modules around within the following sets:
    \begin{align}
    \begin{split}
        &\{[1,\chi_{1,0}],[1,\chi_{1,1}], [1,\chi_{1,2}]\}, \ \ \{[z,\chi_{2,0}],[z,\chi_{2,1}], [z,\chi_{2,2}]\},\\
        &\hspace{.2in}\{[\gamma_6^{+b},\psi_6^{2m}], [\gamma_6^{-b},\psi_6^{-2m}]\},  \ \ \ \ b=1,2, \ \ m=0,1,2.
    \end{split}
    \end{align}
    Any permutation preserving the modular data must fix the vacuum: indeed, the vacuum is the unique row that is entrywise positive. Imposing this splits the first set into two sets, $\{[1,\chi_{1,0}]\}$ and $\{[1,\chi_{1,1}],[1,\chi_{1,2}]\}$. Imposing that the permutation preserve the rows as multisets, i.e.\ 
    \begin{align}
        \{S_{M^g,N}\mid N\in \Rep(V_T)\}=\{S_{M,N}\mid N\in \Rep(V_T)\},
    \end{align} 
    the set $\{[z,\chi_{2,0}],[z,\chi_{2,1}], [z,\chi_{2,2}]\}$ splits into $\{[z,\chi_{2,0}]\}$ and $\{[z,\chi_{2,1}],[z,\chi_{2,2}]\}$. Thus, $g$  must be a composition of disjoint transpositions.

    Suppose that $g$ exchanges $[1,\chi_{1,1}]$ and $[1,\chi_{1,2}]$. By imposing that $g$ preserve the S-matrix elements between $[1,\chi_{1,i}]$ and $[\gamma_6^b,\psi_6^{2m}]$, it is forced that $[\gamma_6^b,\psi_6^{2m}]^g=[\gamma_6^{-b},\psi_6^{-2m}]$. Then, by imposing that $g$ preserve the S-matrix elements between $[z,\chi_{2,i}]$ and $[\gamma_6^b,\psi_6^{2m}]$, it is forced that $g$ exchanges $[z,\chi_{2,1}]\leftrightarrow [z,\chi_{2,2}]$. Thus, $g$ is forced to act as in \eqref{eqn:VTperm}.

    Suppose instead that $g$ fixes $[1,\chi_{1,1}]$ and $[1,\chi_{1,2}]$. Then following logic identical to that of the previous paragraph shows that $g$ must also fix every other module. Thus, the only permutations preserving the modular data of $V_T$ are the trivial permutation and \eqref{eqn:VTperm}.
\end{proof}

\begin{proposition}
   The non-trivial automorphism of $V_T$ permutes its simple modules according to \eqref{eqn:VTperm}.
\end{proposition}

\begin{proof}
    Since automorphisms of a strongly rational VOA induce ribbon auto-equivalences of its representation category, which in particular preserve the modular data in the sense of \eqref{eqn:Aut(S,T)}, it follows from Lemma \ref{lem:VTperm} that the non-trivial automorphism of $V_T$ must either fix every module of $V_T$ or permute them according to \eqref{eqn:VTperm}. To settle which actually occurs, it suffices to determine where the automorphism sends the simple $V_T$-module labeled by $[1,\chi_{1,1}]$. 
    
    Note that by the interpretation of the simple $V_T$-modules afforded by Appendix \ref{app:exceptional}, and using the Schur-Weyl decomposition of \cite{dong1996compact}, we have that the decomposition of $V_2\cong \widehat{\mathfrak{su}}(2)_1$ into $A_4\times V_T$-modules is
    \begin{align}
        V_2\cong\bigoplus_{\chi\in \Rep(A_4)} R_\chi\otimes M_{[1,\chi]}, 
    \end{align}
    where $R_\chi$ is the $A_4$ representation with character $\chi$ and  $M_{[1,\chi]}$ is the $V_T$-module labeled by $[1,\chi]$ (see Table \ref{tab:VTreps}). The sum is over the irreducible representations of $A_4$, with $\chi_{1,0}$ the trivial representation, $\chi_{1,1}$ and $\chi_{1,2}$ the two non-trivial one-dimensional representations, and $\chi_3$ the unique three-dimensional irreducible representation. 

    Recall that the non-trivial automorphism $\bar g$ of $V_T$ is in the image of the map $S_4\to S_4/A_4\cong \mathbb{Z}_2$ obtained by taking $V=V_2$ and $F=A_4$ in Proposition \ref{prop:normaut}. Let $g$ be any element of $S_4\subset\mathrm{Aut}(V_2)$ which maps to the non-trivial involution $\bar g$ acting on $V_T$, i.e.\ $g$ is any odd permutation in $S_4$. By Proposition \ref{prop:permVsubmodules}, we obtain that $M^{\bar g}_{[1,\chi_{1,1}]}=M_{[1,\chi_{1,1}^g]}$. Using the fact that, when $g$ is an odd permutation in $S_4$, the map $A_4\to A_4$ sending $a\mapsto g^{-1}ag$ 
    is an outer automorphism which exchanges the two non-trivial 1-dimensional representations of $A_4$, we learn that $\chi_{1,1}^g=\chi_{1,2}$. Hence, since $\bar g$ permutes simple $V_T$-modules non-trivially, it must permute them according to \eqref{eqn:VTperm}. \end{proof}

\section{Quantum operations and topological lines}\label{app:QuOpTopLin}

In this appendix, we discuss the non-invertible symmetries of the nice $c=1$ VOAs using some of the ideas developed in \cite{Gannon:2026ttf}.  This appendix (and this appendix only) is speculative and written at a physics level of rigor because the requisite notions (quantum operations and generalized twisted modules) have not yet been fully defined for VOAs, though see \cite{Rie22,Dong:2025ttr} for recent work in this direction. We will be somewhat telegraphic and content ourselves with just sketching the calculations. Our main motivation for including this content is to provide an impressionistic picture of what these structures are expected to look like. We leave the important goal of making these computations mathematically precise to the future.

In physics language, for each lattice VOA $V_{2m}$, we characterize both its set of quantum operations (a generalization of the notion of an automorphism of a VOA) and its topological line operators (intuitively, a module of a VOA which is ``twisted'' by a quantum operation). The remaining nice $c=1$ VOAs can be treated using nearly identical arguments to the ones we provide for $V_{2m}$. In the context of conformal nets, which are believed to be equivalent to unitary VOAs \cite{Carpi:2015fga,Carpi:2023onx,henriques2025every}, these objects have been made rigorous: quantum operations were studied in \cite{Bischoff:2016jmy,Bischoff:2022fxf}, and twisted modules/topological line operators are called solitons (see e.g.\ \cite{longo2004topological,henriques2017chernsimons,henriques2017bicommutant}).\footnote{For conformal net aficionados, we comment that we do not study the most general kinds of solitons, but rather only those that preserve the Virasoro subnet.} Thus, our calculations may be viewed as conjectures about the behavior of quantum operations and soliton categories of conformal nets. See also \cite{Marin-Salvador:2025stc,Marin-Salvador:2026jdy} for recent closely related work.

 We refer to \cite{Gannon:2026ttf} for further details behind our perspective on non-invertible symmetries of chiral algebras. Related ideas will also appear in \cite{flatroads,ncb}, mainly in the context of full CFTs with both left- and right-movers.

\subsection{Generalities}
The basic object we are interested in understanding is $\mathrm{Sym}^\dagger(V)$, the category of unitary topological line defects supported on $V$, thought of as a gapless chiral boundary condition of a 3D TQFT. We assume unitarity because much of what we explain below is informed by the conformal net literature, where unitarity is built in from the beginning, though it is likely that this assumption can eventually be dropped for VOAs. We expect that $\mathrm{Sym}^\dagger(V)$ is a (non-rigid) tensor category (and perhaps even a bicommutant category \cite{henriques2017bicommutant}), though we will be mostly agnostic about the details of its structure in what follows.

The category $\mathrm{Sym}^\dagger(V)$ includes $\Rep(V)$ as a fusion subcategory: indeed, $\Rep(V)$ can be thought of as the category of anyons/topological line operators in the bulk, and these can always be pushed to the boundary. See e.g.\ Figure 12 of \cite{Gannon:2026ttf}. 

Every object $X\in\mathrm{Sym}^\dagger(V)$ is expected to define a kind of ``generalized twisted module'' $V_X$ of $V$, though this has not yet been rigorously defined for VOAs. Physically, $V_X$ is the Hilbert space of local operators which live at the endpoint of the line $X$, see Figure 18 of \cite{Gannon:2026ttf}. In the conformal net literature, the analog of $\mathrm{Sym}^\dagger(V)$ is sometimes called $\mathrm{Sol}_{\mathrm{Vir}}(\mathcal{A})$, the category of solitons of a conformal net $\mathcal{A}$ which preserve the Virasoro subnet. 

The category $\mathrm{Sym}^\dagger(V)$ should act on the local operators in $V$, but it will not do so faithfully. Indeed, the objects inside the $\Rep(V)$ subcategory are in the kernel of this action. For this reason, we will also be interested in understanding the quotient
\begin{align}
    \mathrm{QuOp}(V)=\mathrm{Sym}^\dagger(V)\sslash \Rep(V),
\end{align}
consisting of the double cosets of $\Rep(V)$ inside of $\mathrm{Sym}^\dagger(V)$, which does act faithfully on $V$. Alternatively, note that $\Rep(V)$ acts on $\mathrm{Sym}^\dagger(V)$: one might understand the elements of $\mathrm{QuOp}(V)$ as corresponding to the indecomposable $\Rep(V)$-module categories which arise in the decomposition of $\mathrm{Sym}^\dagger(V)$. In nice circumstances, this will be a hypergroup, but not always, as we will see. See e.g.\ \cite{Bischoff:2016jmy,Rie22,Gannon:2026ttf} for reviews of hypergroups and double coset hypergroups.

We comment that $\mathrm{QuOp}(V)$ always contains the automorphisms which preserve the unitary structure of $V$, i.e.\ $\mathrm{Aut}^\dagger(V)\subset \mathrm{QuOp}(V)$. Another nice feature is that the quantum operations are expected to grade the category of topological lines, i.e.\ $\mathrm{Sym}^\dagger(V)$ admits a decomposition of the form
\begin{align}\label{eqn:quopgraded}
    \mathrm{Sym}^\dagger(V) = \bigoplus_{k\in \mathrm{QuOp}(V)}\Rep_k(V), \ \ \ \ \ \ \Rep_1(V)=\Rep(V).
\end{align}
Elements of $\mathrm{QuOp}(V)$ can be understood physically as being implemented by bulk topological surface operators terminating on boundary topological line junctions, see e.g. Figure 22 of  \cite{Gannon:2026ttf}. They also appear in the conformal net literature with the same name \cite{Bischoff:2016jmy,Bischoff:2022fxf}. 

Let us recall from \cite{Gannon:2026ttf} the following characterization of how $\mathrm{Sym}^\dagger(V)$ changes under taking conformal extensions and passing to $G$ fixed points (called topological manipulations in op.\ cit.\ because they can be implemented by gauging one-form and zero-form symmetries). Strictly speaking, these constructions make the most sense when they are applied to fusion categories, finite groups, etc. But we will apply them with reckless abandon in wilder settings below, and see that we still get sensible results.
\begin{proposal}\label{proposal:SymTopMan}
If $W$ is a conformal extension of $V$, and $A$ is the corresponding commutative algebra in $\Rep(V)$ (cf.\ Proposition \ref{prop:ExtAlg}), then 
\begin{align}
    \mathrm{Sym}^\dagger(W)\cong \mathrm{Sym}^\dagger(V)_A,
\end{align}
where $\mathrm{Sym}^\dagger(V)_A$ denotes the category of $A$-modules in $\mathrm{Sym}^\dagger(V)$.
\end{proposal}
\begin{proposal}\label{proposal:Gequiv}
If $G$ is a subgroup of $\mathrm{Aut}^\dagger(V)$, then $G$ acts on $\mathrm{Sym}^\dagger(V)$ and 
\begin{align}
    \mathrm{Sym}^\dagger(V^G)\cong \mathrm{Sym}^\dagger(V)^G,
\end{align}
where $\mathrm{Sym}^\dagger(V)^G$ denotes the $G$-equivariantization of $\mathrm{Sym}^\dagger(V)$. 
\end{proposal}
We can also describe how the quantum operations of $V$ change under passing to the fixed points of a hypergroup action.
\begin{proposal}\label{proposal:quopVH}
    Suppose $H\subset \mathrm{QuOp}(V)$ are both hypergroups. Then the quantum operations of $V^H$ (the subVOA invariant under $H$) are given by the double cosets of $H$ in $\mathrm{QuOp}(V)$,
    \begin{align}
        \mathrm{QuOp}(V^H)=\mathrm{QuOp}(V)\sslash H.
    \end{align}
\end{proposal}
Let us compare this to Proposition \ref{prop:normaut}. Noting that the (unitary) automorphisms of $V$ are always contained inside of the quantum operations, $\mathrm{Aut}^\dagger(V)\subset \mathrm{QuOp}(V)$, suppose that $H$ is taken to be a group of automorphisms of $V$. Then it follows from Proposal \ref{proposal:quopVH} that 
\begin{align}
    N_{\mathrm{Aut}^\dagger(V)}(H)\sslash H\subset \mathrm{Aut}^\dagger(V)\sslash H \subset \mathrm{QuOp}(V)\sslash H = \mathrm{QuOp}(V^H).
\end{align}
But since $H$ is always normal inside of its normalizer, and since the double cosets of a normal subgroup coincide with the single cosets, which define a group, we learn that the group $N_{\mathrm{Aut}^\dagger(V)}(H)/H\subset \mathrm{QuOp}(V^H)$ and therefore acts on $V^H$ via automorphisms, in agreement with Proposition \ref{prop:normaut}. What this also shows is that the remaining automorphisms of $V$ do not disappear after passing to the fixed points $V^H$: rather, they become ``non-invertible'' symmetries described by the double-coset hypergroup $\mathrm{Aut}^\dagger(V)\sslash H$.

\subsection{The Witt class of \texorpdfstring{$\widehat{\mathfrak{su}}(2)_1$}{SU(2) level 1}}
Let us apply some of these considerations to the nice $c=1$ VOAs. The easiest class of nice $c=1$ VOAs to analyze are those which conformally embed into $V_2\cong \widehat{\mathfrak{su}}(2)_1$. In the language of \cite{Moller:2024xtt}, such VOAs are in the same Witt class as $V_2$, in the sense that they have the same central charge as $V_2$ and their representation categories are Witt equivalent to that of $V_2$.

The starting point from which everything else flows is the observation (already made in \cite{Bischoff:2022fxf} in the context of conformal nets) that the quantum operations and topological lines of $V_2\cong \widehat{\mathfrak{su}}(2)_1$ are entirely group-theoretical. 

\begin{proposal}[Symmetries of $\widehat{\mathfrak{su}}(2)_1$]
    The quantum operations and topological line operators of $V_2$ are given by 
    \begin{align}
    \begin{split}
    \mathrm{QuOp}(V_2)&=SO(3),\\
    \mathrm{Irr}(\mathrm{Sym}^\dagger(V_2))&=SU(2).
    \end{split}
\end{align}
\end{proposal}

Indeed, $V_2$ has a known $SU(2)$ worth of topological lines which come from exponentiating its spin-1 currents. The fact that these exhaust all of the lines in $\mathrm{Sym}^\dagger(V_2)$ basically follows from the fact that $\mathrm{Sym}^\dagger(V_2)\sslash\Rep(V_2) = \mathrm{Sym}^\dagger(V_2)\sslash \mathbb{Z}_2= SO(3)$ by definition. See \cite{flatroads,ncb} for further discussion.

We can then calculate $\mathrm{QuOp}(V)$ and $\mathrm{Sym}^\dagger(V)$ for every other nice $c=1$ VOA $V$ by relating $V$ to $\widehat{\mathfrak{su}}(2)_1$ via a sequence of topological manipulations. The VOAs which are in the Witt class of $\widehat{\mathfrak{su}}(2)_1$ are easiest to achieve this for since they can all be reached from $\widehat{\mathfrak{su}}(2)_1$ by passing to $G$ fixed points, for $G$ a finite group of automorphisms. Thus, we can invoke Proposal \ref{proposal:Gequiv}.

Analogous computations were carried out in \cite{Gannon:2026ttf} for the $\widehat{\mathfrak{u}}(1)$ Kac-Moody algebra (i.e.\ the Heisenberg VOA), and our treatment below is very similar to that case. 

We content ourselves with illustrating the ideas for the circle branch VOAs $V_{2k^2}=\widehat{\mathfrak{su}}(2)_1^{\mathbb{Z}_k}$. The rest of the VOAs in the Witt class of $\widehat{\mathfrak{su}}(2)_1$ (i.e.\ $V_{2k^2}^+$, $V_T$, $V_O$, and $V_I$) can be treated using completely analogous methods. Formally applying Proposal \ref{proposal:Gequiv} and Proposal \ref{proposal:quopVH} gives the following.

\begin{proposal}[Symmetries of $V_{2k^2}$]\label{proposal:V2k^2}
    Let $\zeta_r=e^{2\pi i/r}$. The full category of unitary topological line operators of the chiral boson theory $V_{2k^2}$ has the following collection of simple objects,
    \begin{align}
    \begin{split}
        \mathrm{Irr}(\mathrm{Sym}^\dagger(V_{2k^2}))&=\{L_{a,b}\mid (a,b)\in \mathbb{C}^2,~ |a|<1, ~|a|^2+|b|^2=1,b\sim \zeta_kb\} \\
        &\ \ \ \ \ \ \ \ \ \ \  \cup \{M_\lambda \mid \lambda \in \mathbb{R}/k\mathbb{Z}\},
        \end{split}
    \end{align}
    where $b\sim \zeta_kb$ indicates that we identify the lines $L_{a,b}$ and $L_{a,\zeta_kb}$. These lines obey the fusion rules 
    \begin{align}\label{eqn:V2k^2 fusion}
    \begin{split}
        M_\lambda \otimes M_{\lambda'}\cong M_{\lambda+\lambda'}, \ \ \ L_{a,b}\otimes L_{a',b'}\cong \bigoplus_{\ell=0}^{k-1}L_{a_\ell,b_\ell}, \\
        M_\lambda \otimes L_{a,b} \cong L_{e^{2\pi \ri\lambda}a,e^{2\pi \ri\lambda}b}, \ \ \  L_{a,b}\otimes M_\lambda \cong L_{e^{2\pi \ri\lambda}a,e^{-2\pi \ri\lambda}b},
        \end{split}
    \end{align}
    where we define $L_{e^{2\pi \ri\alpha},0}:=\bigoplus_{n=0}^{k-1}M_{\alpha+ n}$ and
    \begin{align}
    \begin{split}
        a_\ell &= aa'-b\overline{b'}\zeta_k^{-\ell}, \ \ \ \ b_\ell = \overline{a'}b+ ab'\zeta_k^\ell .
    \end{split}
    \end{align}
    The tensor unit is $M_0$. The ordinary module  $V_{2k^2,r}$ of $V_{2k^2}$ is identified with the Hilbert space of boundary local operators at the end of the line $M_{\sfrac r{2k}}$. The effective hypergroup which acts faithfully on boundary local operators is 
    \begin{align}
        \mathrm{QuOp}(V_{2k^2})=SO(3)\sslash \mathbb{Z}_k \cong SU(2) \sslash \mathbb{Z}_{2k},
    \end{align}
    whose elements and multiplication rule are described in Equation \eqref{eqn:SU(2)//Zk double cosets} and Equation \eqref{eqn:SU(2)//Zk multiplication}.
\end{proposal}

\noindent Let us give the derivation of this proposal, explaining further details as we go.

We start with the quantum operations, which are described by a hypergroup. By applying Proposal \ref{proposal:quopVH}, we learn that it is $SO(3)\sslash \mathbb{Z}_k$. For the purposes of understanding how this hypergroup grades $\mathrm{Sym}^\dagger(V_{2k^2})$, it is more convenient to describe it in $SU(2)$ variables. To this end, we note that we can lift up to $SU(2)$ and describe it as $SU(2)\sslash \mathbb{Z}_{2k}$. We parametrize $SU(2)$ elements as 
\begin{align}
    g(a,b) = \left(\begin{array}{cc} a & b \\ -\bar b & \bar a\end{array}\right),  \ \ \ a,b\in\mathbb{C}, \ \ \ |a|^2+|b|^2=1.
\end{align}
The double cosets $[a,b]:=\mathbb{Z}_{2k} g(a,b) \mathbb{Z}_{2k}$ are
\begin{align}\label{eqn:SU(2)//Zk double cosets}
    [a,b] =\begin{cases} 
\{ g(a\zeta_{2k}^{u},b\zeta_{2k}^{v})\mid u,v\in\mathbb{Z}_{2k}, ~ u=v~\mathrm{mod}~2\}, & 0<|a|<1 \\ 
\{ g(a\zeta_{2k}^r,0)\mid r\in\mathbb{Z}_{2k}\} ,& |a|=1 \\
\{ g(0,b\zeta_{2k}^r)\mid r\in\mathbb{Z}_{2k}\} ,& |a|=0,
    \end{cases}
\end{align}
so that $[a,b]\sim [a',b']$ if and only if $a'=a\zeta_{2k}^{u}$ and $b'=b\zeta_{2k}^{v}$ for some integers $u,v$ which are either both even or both odd. 
A straightforward calculation of the hypergroup multiplication, given by the standard formula 
\begin{align}
    [g]\star [g'] = \frac{1}{|H|}\sum_{h\in H}[ghg']
\end{align}
for general $G\sslash H$ with $H$ finite, gives 
\begin{align}\label{eqn:SU(2)//Zk multiplication}
    [a,b]\star [a',b']=\frac{1}{k}\sum_{\ell=0}^{k-1} [aa'-b\bar{b'}\zeta_{k}^{-\ell},\bar{a'}b+ab'\zeta_{k}^{\ell}],
\end{align}
where some of the terms in the sum may be equal to each other (though generically they are not). For example, one has that 
\begin{align}
\begin{split}
    [e^{\ri\alpha},0]\star [e^{\ri\beta},0] = [e^{\ri(\alpha+\beta)},0], \ \ \ \ \ [0,e^{\ri\alpha}]\star [0,e^{\ri\beta}] = [e^{\ri(\alpha-\beta)},0], \\ [e^{\ri\alpha},0]\star [0,e^{\ri\beta}]=[0,e^{\ri(\alpha+\beta)}], \ \ \ \ \ [0,e^{\ri\beta}]\star [e^{\ri\alpha},0] = [0,e^{\ri(\beta-\alpha)}],
\end{split}
\end{align}
where we make the identifications $[e^{\ri\alpha},0]\sim [e^{\ri(\alpha+\pi/k)},0]$ and $[0,e^{\ri\beta}]\sim[0,e^{\ri(\beta+\pi/k)}]$. In other words, the full hypergroup possesses an $O(2)$ subgroup, 
\begin{align}
    O(2)\cong \{[e^{\ri\alpha},0]\mid \alpha\in[0,\pi/k)\}\cup \{ [0,e^{\ri\beta}]\mid \beta\in [0,\pi/k)\}
\end{align}
which we identify with the unitary automorphism group of the chiral algebra, $\mathrm{Aut}^\dagger(V_{2k^2})\cong O(2)$ when $k>1$.

Now, let us turn to the topological line operators on the boundary. As follows from Proposal \ref{proposal:Gequiv}, we can formally describe $\mathrm{Sym}^\dagger(V_{2k^2})$ as the $\mathbb{Z}_k$ equivariantization of $\mathrm{Sym}^\dagger(V_2)$. Note that the $\mathbb{Z}_k$ acts on $SU(2)$ by conjugation, in particular it maps
\begin{align}
    g(a,b) \xmapsto{r\in\mathbb{Z}_k} g(a,b \zeta_k^r).
\end{align}
In a $\mathbb{Z}_k$ equivariantization, we expect that the simple objects are parametrized by pairs $(X,\rho)$. Here, $X$ is a direct sum over simple objects of $\mathrm{Sym}^\dagger(V_2)$ (i.e.\ elements of $SU(2)$) in the orbit of some element $g(a,b)$ under the action of $\mathbb{Z}_k$ by conjugation, and $\rho$ is a representation of the stabilizer group $G_{a,b} = \{r\in\mathbb{Z}_k\mid g(a,b\zeta_k^r)=g(a,b)\}$, which is $\mathbb{Z}_k$ if $|a|=1$ and trivial if $|a|<1$. In particular, we obtain lines of the form
\begin{align}
    L_{a,b} \equiv \left(\bigoplus_{r\in\mathbb{Z}_k}g(a,b\zeta_k^r),\rho=1\right), \ \ \ \ a,b\in\mathbb{C}, \ \ |a|^2+|b|^2=1, \ \ |a|<1,
\end{align}
where we must identify  $L_{a,b}\sim L_{a,b\zeta_k}$. We also obtain lines of the form
\begin{align}
    X_{a,n} \equiv (g(a,0),\rho=n), \ \ \ a\in\mathbb{C}, \ \ |a|=1, \ \ n\in\widehat{\mathbb{Z}_k}=\mathrm{Hom}(\mathbb{Z}_k,U(1)).
\end{align}
It will actually be convenient in what follows to repackage the $X_{a,n}$ as
\begin{align}
    M_\lambda\equiv X_{e^{2\pi \ri\lambda},\lfloor \lambda\rfloor}, \ \ \ \lambda \in [0,k).
\end{align}
In the limit that $(a,b)\to (e^{2\pi \ri\alpha},0)$, the line $L_{a,b}$ splits into a direct sum of the form
\begin{align}\label{eqn:Lablimiting}
    L_{a,b}\xmapsto{a,b\to e^{2\pi \ri\alpha},0}\bigoplus_{r=0}^{k-1} M_{\alpha+ r}.
\end{align}
It is clear that the quantum dimensions are
\begin{align}
    \dim(L_{a,b})=k, \ \ \ \dim(M_{\lambda}) = 1.
\end{align}
Indeed, these assignments  are consistent with the limiting behavior in Equation \eqref{eqn:Lablimiting}.

The twisted module of $V_{2k^2}$ associated to $L_{a,b}$ (i.e.\ the Hilbert space of local operators living at the end of the line $L_{a,b}$) is isomorphic to the twisted module $\widehat{\mathfrak{su}}(2)_{g(a,b)}$ of $\widehat{\mathfrak{su}}(2)$, restricted to $V_{2k^2}$, 
\begin{align}
    (V_{2k^2})_{L_{a,b}}\cong \mathrm{Res}(\widehat{\mathfrak{su}}(2)_{g(a,b)}).
\end{align}
Standard formulas from the literature (see e.g.\ \cite{burciu2013fusion}) allow us to determine the fusion rules on $\mathrm{Sym}^\dagger(V_{2k^2})$. In particular, we recover precisely the fusion rules reported in Equation \eqref{eqn:V2k^2 fusion}.

Our general expectation is that the category $\mathrm{Sym}^\dagger(V_{2k^2})$ is graded by the hypergroup $SO(3)\sslash \mathbb{Z}_k$, as in \eqref{eqn:quopgraded}. The simples in the graded component labeled by $[a,b]$ are just the objects $(X,\rho)\in \mathrm{Sym}^\dagger(V_{2k^2})$ for which $X\subset [a,b]$, where we write $X\subset [a,b]$ if all the $SU(2)$ elements arising as summands of $X$ are contained in the double coset $[a,b]$. Unpacking this more explicitly, we have
\begin{align}
   \mathrm{Irr}( \mathrm{Sym}^\dagger(V_{2k^2})_{[a,b]})=
   \begin{cases}
       \{M_{\alpha+ r/2k}\mid r=0,\dots,2k^2-1\}, & \text{if }[a,b]\sim [e^{2\pi\ri\alpha},0], \\
       \{L_{0,b},L_{0,b\zeta_{2k}}\}, & \text{if }[a,b]\sim[0,b], \\
       \{L_{a\zeta_{k}^r,b}\}_{r\in\mathbb{Z}_k} \cup \{L_{a\zeta_{2k}\zeta_k^r,b\zeta_{2k}}\}_{r\in\mathbb{Z}_k} & \text{otherwise.}
   \end{cases}
\end{align}
Recall that the full hypergroup possesses an $O(2)$ subgroup. The tensor subcategory obtained by keeping just the graded components corresponding to $O(2)$ elements is generated by $L_{0,b}$ and $M_\lambda$.

\subsection{Outside of the Witt class of \texorpdfstring{$\widehat{\mathfrak{su}}(2)_1$}{su(2) level 1}}

Describing the quantum operations and topological line operators of $V_{2m}$ requires much more care when $m$ is not a perfect square. Indeed, in this situation, the results of \cite{Moller:2024xtt} imply that we cannot connect $V_{2m}$  to $V_2$ via iterated topological manipulations using only finite symmetries. We can, however, use infinite symmetries. As we will see, the quantum operations and topological lines of $V_{2m}$ with $m\neq k^2$ look qualitatively quite different from those appearing in Proposal \ref{proposal:V2k^2}. 

It is natural to connect $V_2$ to $V_{2m}$ in two steps. In the first step, we can pass from $V_2$ to the $\widehat{\mathfrak{u}}(1)$ Kac--Moody algebra by projecting onto operators of $V_2$ which are invariant under the Cartan of the automorphism group. In the second step, we extend $\widehat{\mathfrak{u}}(1)$ to $V_{2m}$ using a $\mathbb{Z}$ simple current extension. In each step, we keep track of the quantum operations and topological line operators using Proposal \ref{proposal:SymTopMan} and Proposal \ref{proposal:Gequiv}.

The first step was actually carried out already in \cite{Gannon:2026ttf}. We recall the result. 

\begin{proposal}[Symmetries of $\widehat{\mathfrak{u}}(1)$]
    The quantum operations of $\widehat{\mathfrak{u}}(1)$ are given by 
    \begin{align}
        \mathrm{QuOp}(\widehat{\mathfrak{u}}(1))=SO(3)\sslash SO(2),
    \end{align}
    and the simple topological line operators are given by 
    \begin{align}
        \mathrm{Irr}(\mathrm{Sym}^\dagger(\widehat{\mathfrak{u}}(1)))=\{L_a\mid a\in\mathbb{C},~|a|<1\}\cup \{M_\lambda \mid \lambda \in \mathbb{R}\},
    \end{align}
    with fusion rules 
    \begin{align}\label{eqn:diskfusion}
    \begin{split}
       & M_\lambda\otimes M_{\lambda'}\cong M_{\lambda+\lambda'}, \ \ \ M_\lambda \otimes L_a \cong L_a\otimes M_\lambda \cong L_{ae^{ 2\pi \ri \lambda}}, \\
        &\hspace{.6in}  L_a\otimes L_b \cong \int_{\theta\in[0,2\pi)}^\oplus[d\theta]~ L_{ab-r(a,b)e^{\ri\theta}}.
    \end{split}
   \end{align}
In the above formulas, $[d\theta]$ is a measure which we remain mostly agnostic about, and we have defined $r(a,b)=\sqrt{(1-|a|^2)(1-|b|^2)}$ and
    \begin{align}
        L_{e^{2\pi \ri \lambda}}\cong \bigoplus_{\lambda' \in \lambda +\mathbb{Z}}M_{\lambda'}.
    \end{align}
    The line $M_\lambda$ corresponds to the simple module of $\widehat{\mathfrak{u}}(1)$ with conformal dimension $h(M_\lambda)=\lambda^2$. 
\end{proposal}

The result of formally carrying out the second step is the following speculative proposal.

\begin{proposal}[Symmetries of $V_{2m}$ with $m\neq k^2$]
    The simple topological line operators are 
    \begin{align}
    \begin{split}
        &\mathrm{Irr}(\mathrm{Sym}^\dagger(V_{2m}))= \{\widetilde{M}_\lambda \mid \lambda \in \mathbb{R}/\sqrt{m}\mathbb{Z}\} \\
        &\hspace{.3in}\cup \{\widetilde{L}_a\mid a\in\mathbb{C},~0<|a|<1,~a\cong a e^{2\pi \ri \sqrt{m}}\}\cup \{\widetilde{N}_\alpha\mid \alpha \in \mathbb{R}/\tfrac{1}{\sqrt{m}}\mathbb{Z}\}.
    \end{split}
    \end{align}
    The following is a subset of the fusion rules:
    \begin{align}\label{eqn:nonsquaremfusion}
    \begin{split}
        &\hspace{.9in}\widetilde{M}_\lambda\otimes \widetilde{M}_{\lambda'}\cong \widetilde{M}_{\lambda+\lambda'}, \ \ \ \ \ \ \widetilde{M}_{\lambda}\otimes \widetilde{N}_\alpha\cong \widetilde{N}_{\lambda+\alpha}\cong \widetilde{N}_{\alpha}\otimes \widetilde{M}_{-\lambda}, \\
&\hspace{.5in}\widetilde{M}_\lambda \otimes \widetilde{L}_a\cong \widetilde{L}_{ae^{2\pi \ri \lambda}}\cong \widetilde{L}_a\otimes \widetilde{M}_\lambda, \ \ \ \ \ \ \widetilde{N}_\alpha\otimes \widetilde{N}_\beta \cong \bigoplus_{\ell=0}^{m-1} \widetilde{M}_{\alpha-\beta+(\ell+\delta_m/2)/\sqrt{m}},
\end{split}
    \end{align}
where $\delta_m=1$ if $m$ is odd and $\delta_m=0$ if $m$ is even. We have the following limiting behaviors:
\begin{align}
    \widetilde{L}_{e^{2\pi i \lambda}}=\bigoplus_{\lambda'\in \lambda+\mathbb{Z}}\widetilde{M}_{\lambda'}, \ \ \ \ \ \widetilde{L}_0=\int_{\alpha\in [0,\tfrac{1}{\sqrt{m}})}^{\oplus} [d\alpha] \widetilde{N}_\alpha.
\end{align}
The quantum operations are given in \eqref{quop:V2m}.
\end{proposal}

We emphasize that the simple lines $\widetilde{L}_{a}$ with a fixed value of $0<|a|<1$ are labeled by quite a wild set: namely, the circle modded out by an infinite order rotation. This is part of the reason we have avoided trying to write out all the fusion rules involving these lines.

Let us now sketch how this proposal is derived. It is important that we assume throughout that $m$ is not a perfect square. We begin by calculating the topological line operators. The starting point is the observation that $V_{2m}$ can be obtained from $\widehat{\mathfrak{u}}(1)$ via a simple current extension. The corresponding commutative algebra that implements the extension is 
\begin{align}
    A^{(m)}= \bigoplus_{\lambda \in \sqrt{m}\mathbb{Z}}M_{\lambda},
\end{align}
where $M_\lambda\in\mathrm{Sym}^\dagger(\widehat{\mathfrak{u}}(1))$. Thus, Proposal \ref{proposal:SymTopMan} predicts that the topological line operators of $V_{2m}$ are given by $A^{(m)}$-modules in $\mathrm{Sym}^\dagger(\widehat{\mathfrak{u}}(1))$, i.e.\ 
\begin{align}
    \mathrm{Sym}^\dagger(V_{2m})=\mathrm{Sym}^\dagger(\widehat{\mathfrak{u}}(1))_{A^{(m)}}.
\end{align}
Borrowing standard results about the modules of simple current algebras (which hold at least in the setting of fusion categories), we expect that the simple lines of $\mathrm{Sym}^\dagger(V_{2m})$ are given by pairs $(X,\rho)$, where $X=X_1\oplus \cdots \oplus X_n$ is an object in $\mathrm{Sym}^\dagger(\widehat{\mathfrak{u}}(1))$ which is invariant under fusion with $M_{\sqrt{m}}$ (with $X_i$ its simple constituents), and $\rho$ is an irreducible representation of the stabilizer group $G_{X_1}$, i.e.\ the subgroup of $\mathbb{Z}=\langle M_{\sqrt{m}}\rangle$ which stabilizes $X_1$. 

Unpacking this explicitly, we are led to the following collection of simple topological line operators (we decorate them with tildes in order to distinguish them from the lines of $\widehat{\mathfrak{u}}(1)$):
\begin{align}
\begin{split}
    \widetilde{M}_\lambda &\equiv \left(\bigoplus_{n\in\mathbb{Z}}M_{\lambda+\sqrt{m}n},~\mathrm{trivial}\right), \ \ \ \  \lambda \in \mathbb{R}, \ \lambda\cong \lambda+\sqrt{m}\\
    \widetilde{L}_a&\equiv \left(\bigoplus_{n\in\mathbb{Z}}L_{ae^{2\pi \ri \sqrt{m}n}},~\mathrm{trivial} \right), \ \ \ \ a\in\mathbb{C}, \ 0<|a|<1, \ a\cong a e^{2\pi \ri \sqrt{m}}, \\
    \widetilde{N}_\alpha &\equiv \left(  L_0, ~n\mapsto e^{2 \pi \ri \sqrt{m}\alpha n}  \right), \ \ \ \ \ \ \ \ \ \ ~ \alpha \cong \alpha+\frac{1}{\sqrt{m}}.
\end{split}
\end{align}
The fusion rules in most cases are straightforward to calculate using the fact that there is a natural induction tensor functor  $I:\mathrm{Sym}^\dagger(\widehat{\mathfrak{u}}(1))\to \mathrm{Sym}^\dagger(\widehat{\mathfrak{u}}(1))_{A^{(m)}}$ mapping
\begin{align}
\begin{split}
    I(M_\lambda)&=\widetilde{M}_\lambda, \\
    I(L_a)&=\widetilde{L}_a ,\\
    I(L_0)&=\int_{[0,\tfrac{1}{\sqrt{m}})}^{\oplus} [d\alpha]\widetilde{N}_\alpha,
\end{split}
\end{align}
see \cite{Gannon:2026ttf} for a review and some worked examples. In the remaining cases, one must do a more direct calculation. The result is the fusion rules reported in \eqref{eqn:nonsquaremfusion}.

The quantum operations can be determined from these lines by computing the double cosets of $\Rep(V_{2m})$ inside of $\mathrm{Sym}^\dagger(V_{2m})$. The objects of $\Rep(V_{2m})$ correspond to the lines $\widetilde{M}_{n/(2\sqrt{m})}$ with $n\in\mathbb{Z}_{2m}$. It is then not difficult to see that 
\begin{align}\label{quop:V2m}
    \mathrm{QuOp}(V_{2m})=\mathrm{Sym}^\dagger(V_{2m})\sslash \Rep(V_{2m}) = O(2)~\sqcup ~(0,1)\times S^1/\langle \pi/\sqrt{m}\rangle,
\end{align}
where $S^1/\langle \pi /\sqrt{m}\rangle$ denotes the real numbers $x\in\mathbb{R}$ subject to the identifications $x\cong x+2\pi$ and $x\cong x+\pi/\sqrt{m}$.
The $O(2)$ part corresponds to the automorphism group. If we call the rotation elements $r_\lambda$ and the reflection elements $s_\lambda$ with $\lambda \cong \lambda+1/(2\sqrt{m})$, then the lines in the double cosets labeled by $r_\lambda$ and $s_\lambda$ are 
\begin{align}
    [r_\lambda]=\{\widetilde M_{\lambda+n/2\sqrt{m}}\mid n\in\mathbb{Z}_{2m}\}, \ \ \ \ [s_\lambda]=\{\widetilde{N}_\lambda, \widetilde{N}_{\lambda+1/2\sqrt{m}}\}.
\end{align}
For elements $(t,\theta)\in (0,1)\times S^1/\langle \pi/\sqrt{m}\rangle$, the double cosets are 
\begin{align}
    [t,\theta]=\{ \widetilde{L}_{te^{\ri\theta+\pi \ri n/\sqrt{m}}}\mid n\in\mathbb{Z}_{2m}\}.
\end{align}

\bibliographystyle{ytphys}
\bibliography{main}

\end{document}